\documentclass[11pt]{article}
\usepackage{style}
\usepackage{shortcuts}
	
\title{New Applications of 3SUM-Counting\\in Fine-Grained Complexity and Pattern Matching}
\author{
	Nick Fischer\thanks{INSAIT, Sofia University ``St. Kliment Ohridski''. \texttt{nick.fischer@insait.ai}. Partially funded by the Ministry of Education and Science of Bulgaria's support for INSAIT, Sofia University ``St. Kliment Ohridski'' as part of the Bulgarian National Roadmap for Research Infrastructure. Parts of this work were done while the author was at Weizmann Institute of Science.}\and 
	Ce Jin\thanks{Massachusetts Institute of Technology. \texttt{cejin@mit.edu}. Supported by NSF grants CCF-2129139 and CCF-2127597, and a Simons Investigator Award.}\and
	Yinzhan Xu\thanks{Massachusetts Institute of Technology. \texttt{xyzhan@mit.edu}. Supported by NSF Grant CCF-2330048, a Simons Investigator Award and HDR TRIPODS Phase II grant 2217058.}}
\date{}

\begin{document}

\maketitle
\begin{abstract}
\noindent
The 3SUM problem is one of the cornerstones of fine-grained complexity. Its study has led to countless lower bounds, but as has been sporadically observed before---and as we will demonstrate again---insights on 3SUM can also lead to \emph{algorithmic} applications.

The starting point of our work is that we spend a lot of technical effort to develop new algorithms for 3SUM-type problems such as approximate 3SUM-counting, small-doubling 3SUM-counting, and a deterministic subquadratic-time algorithm for the celebrated Balog-Szemer\'edi-Gowers theorem from additive combinatorics. All of these are relevant in their own right and may prove useful for future research on 3SUM.

But perhaps even more excitingly, as consequences of these tools, we derive diverse new results in fine-grained complexity and pattern matching algorithms, answering open questions from many unrelated research areas. Specifically:
\begin{itemize}[itemsep=\smallskipamount]
	\item A recent line of research on the ``short cycle removal'' technique culminated in tight 3SUM-based lower bounds for various graph problems via randomized fine-grained reductions [Abboud, Bringmann, Fischer; STOC '23] [Jin, Xu; STOC '23]. In this paper we derandomize the reduction to the important \emph{4-Cycle Listing} problem, answering a main open question of [Fischer, Kaliciak, Polak; ITCS '24].
	\item We establish that \#3SUM and 3SUM are fine-grained equivalent under deterministic reductions, derandomizing a main result of [Chan, Vassilevska Williams, Xu; STOC '23].
	\item We give a deterministic algorithm for the \emph{$(1+\epsilon)$-approximate Text-to-Pattern Hamming Distances} problem in time $n^{1+o(1)} \cdot \epsilon^{-1}$. While there is a large body of work addressing the randomized complexity, this is the first deterministic improvement over Karloff's $\tilde O(n \epsilon^{-2})$-time algorithm in over 30 years.
	\item In the \emph{$k$-Mismatch Constellation} problem the input consists of two integer sets $A, B \subseteq [N]$, and the goal is to test whether there is a shift $c$ such that $|(c + B) \setminus A| \leq k$ (i.e., whether~$B$ shifted by $c$ matches $A$ up to $k$ mismatches). For moderately small $k$ the previously best running time was $\tilde O(|A| \cdot k)$ [Cardoze, Schulman; FOCS '98] [Fischer; SODA '24]. We give a faster $|A| \cdot k^{2/3} \cdot N^{o(1)}$-time algorithm in the regime where $|B| = \Theta(|A|)$. This result also has implications for the \emph{$k$-Mismatch String Matching with Wildcards} problem.
\end{itemize}
\end{abstract}

\setcounter{page}{0}
\thispagestyle{empty}
\clearpage

\setcounter{page}{0}
\thispagestyle{empty}
\tableofcontents
\clearpage

\section{Introduction}
The famous 3SUM problem is to test whether a given set $A$ of $n$ integers  contains a solution to the equation $a + b = c$. It was hypothesized long ago that the naive $\Order(n^2)$-time algorithm for 3SUM is essentially best-possible~\cite{GajentaanO95}, and this hypothesis has since evolved into a corner stone of fine-grained complexity theory, forming the basis of countless conditional \emph{lower bounds}. In a recent result, Chan, Jin, Vassilevska Williams, and Xu~\cite{focs23} showcased that, quite surprisingly, some of the carefully crafted lower bound techniques can be turned into \emph{algorithms} for seemingly unrelated problems, such as pattern matching on strings.

In this work we explore this paradigm further. We develop a colorful palette of new tools tailored towards 3SUM-type questions, often with a particular focus on \emph{deterministic} algorithms, and often aided by tools from \emph{additive combinatorics} such as the celebrated Balog-Szemer\'edi-Gowers (BSG) theorem. Then we apply these tools to derive various new fine-grained lower bounds---but also to design new and faster algorithms for pattern matching problems. We start with an exposition of our applications.

\subsection{Application 1: Deterministic 3SUM-hardness}
Recently, there has been an emerging interest in derandomizing fine-grained reductions from the 3SUM problem~\cite{ChanHe,ChanX24,fischer3sum}. This line of research follows two conceptual goals: (1) understanding the power of randomization in fine-grained complexity, and (2) establishing lower bounds that endure even if there turns out to be a truly subquadratic \emph{randomized} algorithm for 3SUM. Formally, the goal is to establish conditional lower bound under the following \emph{deterministic} hypothesis:

\begin{hypo}[Deterministic 3SUM Hypothesis]
\label{hypo:deter-3sum}
For any $\eps > 0$, there is no deterministic algorithm that solves 3SUM on $n$ numbers in $O(n^{2-\eps})$ time. 
\end{hypo}

This is a weaker hypothesis than the usual 3SUM hypothesis, which commonly also allows randomized algorithms. See~\cite{fischer3sum} for a detailed discussion on the motivations of studying deterministic reductions from 3SUM. 

Historically, the first result for derandomizing reductions from 3SUM is due to Chan and He~\cite{ChanHe}, who derandomized the reduction from 3SUM to 3SUM Convolution. Said reduction, originally due to P{\u{a}}tra{\c{s}}cu~\cite{Patrascu10}, plays a key role in complexity as it first related 3SUM to a non-geometric problem. Recently, Chan and Xu~\cite{ChanX24} showed a deterministic reduction from 3SUM to the All-Edges Sparse Triangle problem\footnote{Given an $m$-edge graph, determine whether each edge in the graph is in a triangle.}, whose randomized version was again shown in~\cite{Patrascu10}.\footnote{Technically, the reduction in~\cite{ChanX24} is from \emph{Exact Triangle} to All-Edges Sparse Triangle (derandomizing an earlier reduction from~\cite{DBLP:conf/focs/WilliamsX20}). Combined with prior work~\cite{ChanHe, DBLP:journals/siamcomp/WilliamsW13} this entails the claimed deterministic reduction from 3SUM to All-Edges Sparse Triangle.} Independent to~\cite{ChanX24}, Fischer, Kaliciak, and Polak~\cite{fischer3sum} studied deterministic reductions from 3SUM in a more systematic way. They derandomized the majority of formerly randomized reductions from 3SUM, including the reduction from 3SUM to All-Edges Sparse Triangle, $3$-Linear Degeneracy Testing, Local Alignment and many more. This leads to an almost-complete derandomization, with the notable exception of the hardness results from three recent papers~\cite{DBLP:conf/stoc/AbboudBKZ22, AbboudBF23,JinX23}.

These three works~\cite{DBLP:conf/stoc/AbboudBKZ22, AbboudBF23,JinX23} establish strong lower bounds for important graphs problems such as 4-Cycle Listing and Distance Oracles. They commonly rely on the ``short cycle removal'' technique, a technique that was first proposed by Abboud, Bringmann, Khoury and Zamir~\cite{DBLP:conf/stoc/AbboudBKZ22} and independently optimized in~\cite{AbboudBF23,JinX23}. It is based on a structure-versus-randomness paradigm, on graphs or on numbers, which seems inherently randomized. Specifically, the key step in the refined proof~\cite{AbboudBF23,JinX23} is to reduce 3SUM to 3SUM instances with small \emph{additive energy}\footnote{The additive energy $\sfE(A)$ of an integer set $A$ is the number of tuples $(a, b, c, d) \in A^4$ with $a + b = c + d$. }. This key step utilizes the popular BSG theorem from additive combinatorics, and obtaining a deterministic algorithm for the BSG theorem is the main technical challenge for desired derandomization~\cite{fischer3sum}. Jin and Xu~\cite{JinX23} additionally showed 3SUM-hardness for 4-Linear Degeneracy Testing, which also remained underandomized. 

In this paper we derandomize the key step in~\cite{AbboudBF23,JinX23} and successfully design a deterministic reduction from general 3SUM instances to 3SUM instances with small additive energy:

\begin{theorem}
\label{thm:smalldoub3sumhardnessmain}
For any constant $\delta>0$, there is a deterministic fine-grained reduction from 3SUM to 3SUM with additive energy at most $n^{2+\delta}$.
\end{theorem}

Naturally one of the main ingredients of the above theorem is a \emph{deterministic subquadratic-time} algorithm for the constructive BSG theorem, which may be of independent interest (see \cref{sec:tools}).

With some more work, \cref{thm:smalldoub3sumhardnessmain} implies the following lower bound for the \emph{4-Cycle Listing} problem (where the goal is to list all $4$-cycles in a given undirected graph), resolving a main open question left by~\cite{fischer3sum}.

\begin{restatable}{theorem}{fourcycle}
\label{thm:4cyclehardnessmain}
Under the deterministic 3SUM hypothesis, there is no deterministic algorithm for 4-Cycle Listing on $n$-node $m$-edge graphs in $O(n^{2-\delta} + t)$ time or $O(m^{4/3 - \delta} + t)$ time for any $\delta > 0$, where $t$ is the number of $4$-cycles in the graph. 
\end{restatable}

These lower bounds tightly match the upper bounds for 4-Cycle Listing~\cite{JinX23,DBLP:conf/fsttcs/AbboudKLS23}. 

We remark that the additional work needed to derive \cref{thm:4cyclehardnessmain} from \cref{thm:smalldoub3sumhardnessmain} does not easily follow from previous work, and is in fact somewhat technically involved. In~\cite{AbboudBF23,JinX23}, the way to derive a randomized version of \cref{thm:4cyclehardnessmain} from a randomized version of \cref{thm:smalldoub3sumhardnessmain} is to first reduce low-energy 3SUM to All-Edges Sparse Triangle with few $4$-cycles (morally, this is a ``low-energy'' variant of triangle detection in graphs). These reductions in~\cite{AbboudBF23, JinX23} are both randomized. Even though deterministic reductions from 3SUM to All-Edges Sparse Triangle are known~\cite{ChanX24, fischer3sum}, they do not seem to preserve low energies. Our alternative reduction looks more like a direct derandomization of~\cite{AbboudBF23}, which makes the reduction more technically heavy. For instance, as a key technical tool it relies on derandomized version of the \emph{almost-additive hash functions} used in~\cite{AbboudBF23, JinX23}. 

Utilizing our deterministic subquadratic-time algorithm for the constructive BSG theorem (which was a main ingredient in \cref{thm:smalldoub3sumhardnessmain}) we also show a deterministic fine-grained equivalence between 3SUM and \#3SUM---the problem of \emph{counting} the number of 3SUM solutions. This equivalence was first shown by Chan, Vassilevska Williams and Xu~\cite{ChanWX23} (by means of randomized reductions). As a technical bonus, and in contrast to~\cite{ChanWX23}, our new reduction avoids the use of fast matrix multiplication (it relies on FFT instead). 

\begin{theorem}
\label{thm:count3sumequivalencenessmain}
\#3SUM and 3SUM are subquadratically equivalent, under deterministic reductions. 
\end{theorem}

\subsection{Application 2: Pattern matching algorithms}
Inspired by the previously observed connection~\cite{focs23} between tools in the fine-grained complexity of 3SUM on the one hand, and the design of pattern matching algorithms on the other hand, we have investigated how our toolset can further be algorithmically exploited. We offer two algorithms for natural pattern matching tasks.

\paragraph{Problem 1: Text-to-Pattern Hamming Distances.}
In the \emph{Text-to-Pattern Hamming Distances} problem (also known as \emph{Sliding Window Hamming Distances}), we are given a pattern string~$P$ of length $m$ and a text $T$ of length $n$ over some common alphabet $\Sigma$, and the goal is to compute the Hamming distance (i.e., the number of mismatches) between $P$ and $T[i \mathinner{.\,.} i + m-1]$ for every shift $i \in [n - m]$. This fundamental text-processing problem has been extensively studied in various settings~\cite{FischerP74,Abrahamson87,Karloff93,FredrikssonG13,AtallahGW13,KopelowitzP15,KopelowitzP18,ChanGKKP20}.

Of particular interest is \emph{approximate} setting, where the typical goal is to compute $(1+\eps)$-approximations of the Hamming distances. The history of this problem is rich, starting with a classical algorithm due to Karloff~\cite{Karloff93} running in time $\widetilde\Order(n \epsilon^{-2})$. At the time this was believed to be best-possible (in light of certain $\Omega(\epsilon^{-2})$ communication lower bounds~\cite{Woodruff04,JayramKS08}), and it was surprising when Kopelowitz and Porat~\cite{KopelowitzP15} came up with an algorithm in time $\widetilde\Order(n \epsilon^{-1})$; see also~\cite{KopelowitzP18} for a more streamlined algorithm. Then, in a very recent breakthrough, Chan, Jin, Vassilevska Williams, and Xu~\cite{focs23} gave an algorithm in $\tilde O(n/\eps^{0.93})$ time, improving the dependence on $\epsilon^{-1}$ to \emph{sublinear}. However, all of these algorithms are randomized.

Karloff~\cite{Karloff93} also designed a deterministic algorithm running in time $O(n\eps^{-2}\log^3 m)$. In stark contrast to the progress by randomized algorithms, however, Karloff's deterministic algorithm from more than 30 years ago has remained the state of the art. In our work we finally improve this $\eps^{-2}$ factor to $\eps^{-1}$ (albeit at the cost of incurring an extra $m^{o(1)}$ factor):

\begin{theorem}
    \label{thm:dethdmain}
  There is a deterministic algorithm solving $(1+\eps)$-approximate Text-to-Pattern Hamming Distances in $nm^{o(1)}/\eps$ time.
\end{theorem}

This answers an open question explicitly asked by Kociumaka in a video talk on~\cite{ChanGKKP20}\footnote{\url{https://youtu.be/WEiQjjTBX-4?t=2820}} and by~\cite{focs23}. 

\paragraph{Problem 2: \boldmath$k$-Mismatch Constellation.}
The \emph{Constellation} problem is a natural geometric pattern matching task: Given two sets $A, B \subseteq \Int^d$ of at most $n$ points, the goal is to find all shifts~\makebox{$c \in \Int^d$} satisfying that $B + c \subseteq A$. Intuitively, this models the task of identifying a certain constellation of stars~$B$ in the night sky~$A$. But what if the constellation contains some mistakes? In the generalized \emph{$k$-Mismatch Constellation} problem the goal is to report all shifts $c$ with~\makebox{$|(c + B) \setminus A| \leq k$} (i.e., all occurrences that match the constellation in the night sky up to $k$ mismatches). It is known that we can assume all input points are in $\Int$ instead of $\Int^d$ without loss of generality \cite{CardozeS98, nickconstellation}. 

The study of this problem was initiated by Cardoze and Schulman~\cite{CardozeS98}. Assuming the natural technical condition $k\le (1-\eta)|B|$ for any fixed constant $\eta>0$ (which implies the output size is at most $n/\eta$), they designed a randomized algorithm in time $\tilde O(nk)$. Recently, Fischer~\cite{nickconstellation} revisited this problem and gave a deterministic $\tilde O(nk) \cdot \polylog(N)$ time algorithm. Both of these results left open whether it is possible to achieve better algorithms, in time $\Order(n k^{1-\epsilon})$, say, or perhaps even in near-linear time~$\widetilde\Order(n)$.

In fact, a side result by Chan, Vassilevska Williams, and Xu~\cite[Lemma 8.4]{ChanWX23} implies an $n^{2-\Omega(1)}$-time algorithm for this problem, which beats time $\tilde O(nk)$ whenever $k$ is very large in terms of~$n$. As our second algorithmic contribution we significantly widen the range of $k$ for which we can improve upon the time $\widetilde\Order(n k)$:

\begin{theorem}\label{thm:constelargebmain}
Let $A, B \subseteq [N]$ and let $1\le k\le 0.3|B|$. There is a deterministic algorithm that solves $k$-Mismatch Constellation in $ |A|\cdot k^{2/3} \cdot (|A|/|B|)^{2/3} \cdot N^{o(1)}$ time. If randomization is allowed, then the $N^{o(1)}$ factor can be replaced by $\polylog(N)$.
\end{theorem}

Notably, this algorithm runs in time $n k^{2/3} N^{o(1)}$ in the regime where $|B| = \Theta(|A|)$ (and $k$ is arbitrary). We improve upon the $\widetilde\Order(n k)$ baseline in the broader range $|B| \gg |A| / \sqrt{k}$. (Counterintuitively, this problem indeed becomes harder the smaller~$B$ is as then there could be completely separate structured regions in $A$ in which~$B$ could be approximately matched.)

But can we do even better? In particular, there is no evidence that the problem cannot even be solved in near-linear time $\widetilde\Order(n)$, irrespective of $k$. Relatedly, in the \emph{$k$-Bounded} Text-to-Pattern Hamming Distances problem (where all distances exceeding $k$ can be reported as $\infty$ in Text-to-Pattern Hamming Distances), an algorithm in time $\tO(n + k \sqrt{n})$ is known~\cite{GawrychowskiU18}, in the setting where the text and pattern both have length $\Theta(n)$. This running time becomes near-linear when $k$ is small ($k \leq n^\eps$, say). These two problems appear similar in nature in that for both problems the task is about shifts that result in at most $k$ mismatches. So can we achieve a similar result for $k$-Mismatch Constellation?

Perhaps surprisingly, we prove that this is not possible---at least not by \emph{combinatorial}\footnote{In the context of fine-grained complexity and algorithm design, we informally say that an algorithm is \emph{combinatorial} if it does not rely on algebraic fast matrix multiplication. It has emerged as a powerful paradigm to prove lower bounds for combinatorial based on the \emph{BMM hypothesis} stating that there is no combinatorial algorithm for Boolean matrix multiplication in time $\Order(n^{3-\epsilon})$ (for all $\epsilon > 0$). See e.g.~\cite{AbboudW14,williams2018some,GawrychowskiU18} and the references therein.} algorithms:

\begin{restatable}{theorem}{BMMLowerBound}
    \label{thm:BMMLowerBound}
    Under the BMM hypothesis, combinatorial algorithms solving $k$-Mismatch Constellation requires $\sqrt{k} n^{1-o(1)}$ time for sets of size $\Theta(n)$. 
\end{restatable}

This lower bound leaves a gap of $k^{1/6}$ to our algorithm. We leave it as an open question to close this gap, and to make further progress on the setting where $|B| \ll |A|$.

\paragraph*{Problem 3: $k$-Mismatch String Matching with Wildcards.}
Interestingly, our algorithm for $k$-Mismatch Constellation also implies a new algorithm for the following well-studied problem \cite{CliffordEPR09,CliffordEPR10,CliffordP10,NicolaeR17,BathieCS24}: given a text string $T \in \Sigma^{n}$, a pattern string $P \in (\Sigma \cup \{\diamondsuit\})^{m}$, and integer $1\le k \le m$, find all occurrences of $P$ inside $T$ allowing up to $k$ Hamming mismatches, where the wildcard $\diamondsuit \notin \Sigma$ can match any single character. 
Note that the wildcards appear only in the pattern and not in the text.

By a simple reduction from this problem to the $k$-Mismatch Constellation problem, we obtain the following result as an immediate corollary of \cref{thm:constelargebmain}:
\begin{corollary}\label{cor:wildcardbasic}
    There is a deterministic algorithm that solves $k$-Mismatch String Matching with Wildcards (in the pattern only) in time $k^{2/3} n^{1+o(1)}$. 
    If randomization is allowed, then the $n^{o(1)}$ factor can be replaced by $\polylog(n)$.
\end{corollary}
The proof of this corollary from \cref{thm:constelargebmain} is short and straightforward (see \cref{subsec:wildcard}). Moreover, we can get a faster algorithm by adjusting parameters and simplifying a few steps in the proof of \cref{thm:constelargebmain}:
\begin{theorem}\label{thm:wildcardbetter}
    There is a deterministic algorithm that solves $k$-Mismatch String Matching with Wildcards (in the pattern only) in time $k^{1/2} n^{1+o(1)}$. 
    If randomization is allowed, then the $n^{o(1)}$ factor can be replaced by $\polylog(n)$.
\end{theorem}

We now briefly compare \cref{thm:wildcardbetter} with previous algorithms for $k$-Mismatch String Matching with Wildcards (see \cite{BathieCS24} for a more comprehensive overview). When wildcards are allowed in  both the text and the pattern, \cite{CliffordEPR09,CliffordEPR10} gave an algorithm in $\tilde O(\min\{nk, n\sqrt{m}\})$ time. When wildcards are allowed either in the text only or in the pattern only, \cite{CliffordP10} gave an algorithm in $\tilde O(n(mk)^{1/3})$ time. Our \cref{thm:wildcardbetter} is faster than these algorithms but only applies when all wildcards are in the pattern. More recent works \cite{NicolaeR17,BathieCS24} gave fast algorithms parameterized by the number of wildcards and/or the number of contiguous segments of wildcards in the input.

\subsection{Technical tools and overview}
\label{sec:tools}

In this section, we list some of our new technical tools for proving the aforementioned results and also discuss some intuition for proving some of the results.  

\paragraph{Tool I: Approximate 3SUM counting.}
The first tool is an efficient deterministic algorithm for approximate 3SUM counting: Given integer (multi-)sets $A$ and $B$, it approximately counts the number of times each number $c$ can be represented as a sum $a + b$ for $a \in A$ and $b \in B$.

Formally, we write $1_A$ to denote the indicator vector of a (multi-)set $A$, and use $\conv$ between two vectors to denote their \emph{convolution} (see formal definitions in \cref{sec:prelim}). In particular, $1_A \conv 1_B$ denotes a vector whose $c$-th coordinate is the number of choices for picking $(a, b) \in A \times B$ so that $a + b = c$. 

\begin{restatable}[Deterministic Popular Sums Approximation]{theorem}{thmdeterministicapproxpop} 
\label{thm:introdeterministicapprox3sum}
Let $A, B \subseteq [N]$ be multisets of size at most $N$ and let $\epsilon > 0$. There is a deterministic algorithm that runs in time $(\epsilon^{-1} |A| + |B|) \cdot N^{\order(1)}$ and computes a vector $f$ of sparsity $\norm{f}_0 \leq \Order(\epsilon^{-1} |A| \log^2 N)$, such that $\norm{f - (1_A \conv 1_B)}_\infty \leq \epsilon |B|$.
\end{restatable}

To the best of our knowledge, this lemma has not explicitly appeared in the literature (even with randomization allowed). A randomized version of \cref{thm:introdeterministicapprox3sum} is simple,\footnote{We can randomly pick a prime $p = \Theta(\eps^{-1} |A| \log N)$, and then compute $\sum_{i \equiv s \pmod{p}}(1_A\conv 1_B)[i]$ for every $s \in \F_p$ using FFT in $O(p \log p) = \tO(\eps^{-1} |A|)$ time. For every index $i$, the expected value of $\sum_{i' \equiv i \pmod{p}, i \ne i'}(1_A\conv 1_B)[i']$ can be bounded by $O(\frac{|A||B|}{p / \log N}) = O(\eps|B|)$. Thus, by repeating $O(\log N)$ times, there must be one iteration where the mod $p$ bucket of $i$ is an additive $O(\eps |B|)$ approximation of $i$, with high probability (we also have to recover the index of $i$ given its bucket, which can be achieved, say, by random sampling)} but the deterministic version needs more care.

\cref{thm:introdeterministicapprox3sum} is used as a subroutine in almost every result in this paper. For example, one useful application is to estimate the number of 4SUM solutions (which can be used for estimating the additive energy of a set) on inputs $A,B,C,D$.  To do this, we simply approximate $1_A\conv 1_B$ by $f$ and approximate $1_C\conv 1_D$ by $g$ using \cref{thm:introdeterministicapprox3sum}, and sum up $\sum_x f[x] g[-x]$.   \cref{thm:introdeterministicapprox3sum} is also used in the proof of our deterministic BSG theorem (to be described later).

As another important application, combining \cref{thm:introdeterministicapprox3sum} with the previous techniques due to Gawrychowski and Uzna\'{n}ski~\cite{GawrychowskiU18} also implies our deterministic Text-to-Pattern Hamming Distances algorithm (\cref{thm:dethdmain}). For the details, we refer to \cref{sec:ham}. But to illustrate the intuition behind the connection to \cref{thm:introdeterministicapprox3sum}, we will describe a very simple $\eps^{-1}n^{1+o(1)}$ time algorithm computing an \emph{additive} $\eps m$-approximation here.

For each symbol $s\in \Sigma$, let $A_s:= \{j: T[j]=s\}$ and $B_s:=\{k: P[k]=s\}$. Then the Hamming distance between $T[i+1\dd i+m]$ and $P[1\dd m]$ can be expressed as 
\begin{equation*}
    |P| - \sum_{s\in \Sigma} (1_{A_s} \conv 1_{-B_s})[i].
\end{equation*}
For each $s\in \Sigma$, we can use \cref{thm:introdeterministicapprox3sum} to compute in $(\eps^{-1}|A_s|+|B_s|)\cdot n^{o(1)}$ time a sparse vector~$f_s$ satisfying that $\|f_s - 1_{A_s} \conv 1_{-B_s}\|_\infty\le \eps |B_s|$. Thus, for each index $i$ we can approximate the Hamming distance by $|P| - \sum_{s\in \Sigma}f_s[i]$ with an additive error of at most 
\begin{equation*}
    \sum_{s\in \Sigma} \big \lvert \big (f_s - 1_{A_s} \conv 1_{-B_s}\big )[i]\big \rvert\le \sum_{s\in \Sigma}\eps |B_s| = \eps m.
\end{equation*}
The total time complexity is $\sum_{s\in \Sigma} (\eps^{-1}|A_s| + |B_s|) \cdot n^{o(1)} = \eps^{-1}n^{1+o(1)}$.

\paragraph{Tool II: Small-doubling 3SUM counting.}
The next technical tool is also for 3SUM counting, but in a different setting. This time we insist on \emph{exact} counts, but we assume that the input is \emph{additively structured.} Specifically, recall the definition of the sumset $A + A = \set{a + b : a, b \in A}$. The \emph{doubling constant} of a set is defined as the ratio~\smash{$K = \frac{|A + A|}{|A|}$}. Intuitively, highly additively-structured sets (such as intervals or arithmetic progressions) have small doubling constant, whereas unstructured sets (such as random sets) have large doubling constant. The following theorem states that we can solve 3SUM counting exactly, provided that at least one input set has sufficiently small doubling constant:

\begin{theorem}[Deterministic small-doubling 3SUM-counting; simplified from \cref{thm:detsmalldoublethreesum}]
   Given sets $A,B,C\subseteq [N]$, where $|A+A|=K|A|$, we can deterministically compute $(1_C\conv 1_{-B})[a]$  for all $a\in A$ in time complexity $O\big(K\sqrt{|A||B||C|}\cdot N^{o(1)} \big )$.
   \label{thm:introdetsmalldouble}
\end{theorem}
In particular, when $|A|,|B|,|C|\le n$, the time complexity becomes $O(Kn^{1.5}N^{o(1)})$, which for small $K$ improves upon the standard quadratic-time algorithm for 3SUM. Previously, Jin and Xu~\cite{JinX23} showed a randomized 3SUM algorithm in this regime in $O(Kn^{1.5}\polylog N)$ time (and Abboud, Bringmann, and Fischer~\cite{AbboudBF23} also independently showed a randomized algorithm in $\poly(K)n^{2-\Omega(1)}$ time).
Our \cref{thm:detsmalldoublethreesum} simultaneously (i) derandomizes these results, (ii) extends them to the counting version of 3SUM, and (iii) also allow $A,B,C$ to differ in size. The latter two new features turn out to be crucial for our $k$-Mismatch Constellation algorithm.

The randomized version of this small-doubling 3SUM algorithm is a key step in the randomized fine-grained reductions in~\cite{AbboudBF23,JinX23} from 3SUM to small-energy 3SUM (and thus a stepping stone in the chain of reductions to 4-Cycle Listing and Approximate Distance Oracles). Small doubling instances for $k$-SUM and integer programming have also been studied by Randolph in his thesis~\cite{randolph2024exact}.

\paragraph{Tool III: Almost-additive hashing with small seed length.}
A family of hash functions $h$ is called \emph{additive} if it satisfies $h(x + y) = h(x) + h(y)$, and \emph{almost-additive} if~\makebox{$h(x + y) - h(x) - h(y)$} takes only constantly many different values. (Almost-)additive hashing has been playing an important role for 3SUM for a long time~\cite{BaranDP08,Patrascu10}, and also the modern fine-grained reductions oftentimes rely on such hash functions. The most common use cases is to reduce the universe size in the 3SUM problem (typically at the cost of introducing some few false solutions). In~\cite{fischer3sum} it was shown that this application can be derandomized building on the family of hash functions $h(x) = x \mod p$ (where $p$ is a random prime).

Unfortunately, for our purposes this derandomization result is often not sufficient. The reason is technical: For some applications, specifically the reductions based on arithmetic short cycle removal~\cite{AbboudBF23,JinX23}, we additionally need the hash function to be 3-wise independent. While it is provably impossible for a hash family to be (almost-)linear and 3-wise independent at the same time, this can at least approximately be achieved (see \cref{sec:dethash}). To derandomize this more constrained hash family, we come up with a creative construction that besides the hash function $h(x) = x \mod p$ from before, also throws derandomization by \emph{$\epsilon$-biased sets}~\cite{AlonGHP92} into the mix.

\paragraph{Tool IV: Deterministic BSG theorem.}
Putting together the previous three tools, we manage show a deterministic subquadratic-time algorithm to compute the Balog-Szemer\'{e}di-Gowers (BSG) theorem:

\begin{theorem}[Deterministic subquadratic-time BSG theorem; simplified from \cref{thm:derandbsg}]
\label{thm:bsgintro}
 Given set $A\subseteq [N]$ and $K\ge 1$ such that the additive energy is $\sfE(A)\ge |A|^3/K$, we can deterministically find a subset $A' \subseteq A$ in $(|A|^2/ K^{2} + K^{13}|A|)N^{o(1)}$ time, such that 
         $|A'| \ge |A|/(64K)$ and
             $|A'+A'|\le K^{16}|A|\cdot N^{o(1)}$.
\end{theorem}

In words, the BSG theorem~\cite{balog1994statistical,gowers2001new} states that a set with high additive energy must contain a large subset of small doubling.
The constructive BSG theorem was first used by Chan and Lewenstein~\cite{ChanL15}, and later played a crucial role in the randomized fine-grained reductions of~\cite{AbboudBF23, JinX23}, and different versions of BSG theorem were proved in~\cite{ChanWX23}. 
All previous proofs for (versions of) the constructive BSG theorem were randomized, except the one in~\cite{ChanL15}: they showed a deterministic algorithm that runs in $O(n^\omega)$ time for constant values of $K$, where~\makebox{$\omega < 2.372$}~\cite{DBLP:conf/focs/DuanWZ23, WXXZ24} is the matrix multiplication exponent. But this running time is not efficient enough for our purposes; we really need subquadratic time.
All other previous proofs for constructive BSG theorem use randomness crucially. One (less crucial) use of randomness was for approximating the degree of nodes in certain bipartite graphs defined by the additive relation of the input sets; for this part, the degrees actually correspond to 3SUM solution counts, and hence can be derandomized using our 
\cref{thm:introdeterministicapprox3sum}. However, a more crucial use of randomness in the previous proof is due to the probabilistic method, specifically the method of \emph{dependent random choice}. For this part, a naive derandomization would require enumerating all possibilities, taking (at least) quadratic time.

Although our deterministic constructive BSG is slower than known randomized versions (which run in $\tilde O(|A|\poly(K))$ time~\cite{ChanL15,AbboudBF23,JinX23}), and provide worse quantitative bounds, it still runs in subquadratic time when $K$ is not too small.

Roughly speaking, the proof of our deterministic BSG theorem uses the following strategy: We do not know how to overcome the aforementioned quadratic barrier for derandomizing the probabilistic method, but we can afford to run this slow derandomization on some much smaller subset $A_0\subseteq A$. Then, we will use the results computed on this small subset $A_0$ to extract the desired small-doubling large subset $A'\subseteq A$. Naturally, if $A_0$ is chosen smaller, then this derandomization is faster, at the cost of worsening the doubling constant of the extracted subset $A'$. In order to make such extraction possible, we need to make subset $A_0\subseteq A$ inherit the high additive energy of $A$ (roughly speaking, if $\sfE(A)\ge |A|^3/K$, then we would like $\sfE(A_0)\ge |A_0|^{3-o(1)}/K$). 
Simply uniformly subsampling $A_0\subseteq A$ would cause too much loss, and we need to do something more technical: we (deterministically) pick a certain almost-additive hash function (see \cref{sec:dethash}), mapping $A$ to roughly $|A|/|A_0|$ buckets, and then pick one of the buckets as our $A_0$.
This is exactly where our tool III, the deterministic almost-additive hash functions, come into play. We find it interesting that these hash families find an application here that is completely different from their original intent~\cite{AbboudBF23,JinX23}.

\paragraph{Tool V: 3SUM counting for popular sums.}
The fifth and final tool we introduce is yet another variant of 3SUM counting: This time the goal is to \emph{exactly} count 3SUM solutions, but only those involving the most \emph{popular} elements (i.e., those elements $s$ that are contained in many 3-sums):

\begin{restatable}[Deterministic 3SUM-counting for popular sums]{theorem}{thmpopular}
  \label{thm:introdetpopsum}
   Given sets $A,B\subseteq [N]$, $|A|,|B|\le n$, and a parameter $k\ge 1$, there is a deterministic algorithm that returns $S = \{s: (1_{A}\conv 1_B)[s] \ge |A|/k\}$ and exactly computes $(1_{A}\conv 1_B)[s]$ for all $s\in S$ in $O(k^2 |A|^{2-3/128}\cdot N^{o(1)})$ time.
\end{restatable}

Chan, Vassilevska Williams and Xu~\cite[Lemma 8.4]{ChanWX23} gave a randomized subquadratic-time algorithm for this task using fast matrix multiplication. Our \cref{thm:introdetpopsum} derandomizes~\cite[Lemma 8.4]{ChanWX23}, and does not require fast matrix multiplication, but the running time has a worse exponent (but is still subquadratic).

In fact, establishing \cref{thm:introdetpopsum} immediately completes the proof of \cref{thm:count3sumequivalencenessmain}. Recall that this theorem was originally proven in~\cite{ChanWX23} (by means of randomized reductions), and tracing through their proofs it is easy to check that replacing their Lemma~8.4 by our \cref{thm:count3sumequivalencenessmain} replaces the only randomized step by a deterministic alternative.

\paragraph{Techniques for the \boldmath$k$-Mismatch Constellation algorithm.}
Throughout the overview we have encountered hints for the proofs of all our algorithmic applications, except for the algorithm for $k$-Mismatch Constellation; let us finally provide some overview for this result. Similar to previous work of Fischer~\cite{nickconstellation}, we solve $k$-Mismatch Constellation by solving a \emph{Partial Convolution} instance. That is, given sets $A, B, C \subseteq [N]$, where $C$ is a set of candidate answers such that every $c\in C$ is promised to satisfy $|(c+B)\setminus A|\le O(k)$, and the goal is to exactly compute $|(c+B)\setminus A|$ for all~\makebox{$c\in C$}.
The algorithm of Fischer~\cite{nickconstellation} (which can be viewed as a certain transposed version of sparse convolution) solves this task in $\tilde O(|A|+|B+C|)\le \tilde O(|A|+k|C|)$ time, which is at most $\tilde O(k|A|)$ since $|C|= O(|A|)$.
Our algorithm uses a win-win argument: If $|C|$ is small, then Fischer's algorithm is already faster than $\tilde O(k|A|)$. Otherwise, if $|C|$ is large, then we show that the input instance must have very rich additive structure. In particular, $C$ must have low doubling constant, and there is very large subset $B'\subseteq B$ of size $|B'|\ge |B|-O(k)$ such that the sumset $B'+C$ is small. Then, we separately compute with the contribution of $B'$ and $B\setminus B'$ to the answers:
\begin{itemize}
    \item For $B'$, based on the small sumset size $|B'+C|$, we use the technique of Fischer~\cite{nickconstellation}  to compute their contribution.
    \item For $B\setminus B'$, we take advantage of the small doubling of $C$, and use our small-doubling \#3SUM algorithm (\cref{thm:introdetsmalldouble}) to compute their contribution.
\end{itemize}

\subsection{Discussions and open questions}
Our work further explores the role of additive structure in the context of fine-grained complexity and pattern matching problems. We believe the various technical tools developed in this paper can find future applications.

We conclude with some natural open questions arising from our work.
\begin{enumerate}[itemsep=\medskipamount]
    \setlength\parindent{1.6em}
    \item  Can we deterministically solve $(1+\eps)$ Text-to-Pattern Hamming Distances in time complexity $O((1/\eps)^{1-c} n^{1+o(1)})$ for some constant $c>0$? The recent breakthrough of~\cite{focs23} achieved such a running time using randomization. One of their key steps is to compute the multiplicities of popular sums up to very high precision, using the techniques from the randomized algorithm of~\cite{ChanWX23} that solves \#3SUM for popular sums.  This part can already be derandomized by our \cref{thm:introdetpopsum}. However, another step in the proof of~\cite{focs23} involves analyzing the error of random subsampling by bounding the variance, and we do not know how to derandomize such arguments.
    
    \item  Our derandomized reduction from low-energy 3SUM problem to the 4-Cycle Listing problem turned out to be much more technical than its randomized counterpart in~\cite{AbboudBF23,JinX23}.
Due to similar technical reasons, \cref{thm:smalldoub3sumhardnessmain} does not immediately show conditional lower bound for Distance Oracles under the deterministic 3SUM hypothesis, and it seems harder to derandomize the reduction in~\cite{AbboudBF23, JinX23} to Distance Oracles because the reduction needs to bound the number of longer cycles, which seems more difficult to control. 
We leave it as an interesting open problem  to derandomize the reduction from 3SUM to Distance Oracles.

\item Another underandomized result in~\cite{JinX23} is the 3SUM hardness for 4-Linear Degeneracy Testing (whereas 3-Linear Degeneracy Testing has deterministic 3SUM hardness, as shown earlier in~\cite{ldt}). We expect a larger barrier in derandomizing this reduction.  The reason is that, in order to show hardness for 4-Linear Degeneracy Testing,~\cite{JinX23} first showed hardness for 3SUM instances whose additive energy is $2n^2 - n$ (i.e., Sidon sets). This suggests that one might need to first improve the bound of \cref{thm:smalldoub3sumhardnessmain} from $O(n^{2+\delta})$ to $2n^2 - n$, which seems challenging using our current method. (In contrast, it is totally reasonable that \cref{thm:smalldoub3sumhardnessmain} could be used in a black-box way, combined with additional ideas, to show lower bound for Distance Oracles under the deterministic 3SUM hypothesis.)

    \item Our algorithm for $k$-Mismatch Constellation only achieves improvement in the $|B|\approx |A|$ regime. Can we show faster algorithms when $|B|$ can be much smaller than $|A|$?
    \item Can we relax \cref{thm:wildcardbetter} to allow wildcards in both the pattern and the text?
\end{enumerate}
\section{Preliminaries}
\label{sec:prelim}
Let $\Z$ and $\N = \Z_{\ge 0}$ denote the integers and the non-negative integers respectively. Let $[n] = \{0,1,\dots,n-1\}$. We use $\tilde O(f)$ to denote $O(f\cdot \polylog f)$.

For two sets $A,B$, define their sumset $A+B = \{a+b : a\in A, b\in B\}$ and difference set $A-B = \{a-b : a\in A, b\in B\}$. Let $-A= \{-a: a\in A\}$.
We define the iterative sumset $kA$ by the recursive definition $kA=(k-1)A + A$, and $1A=A$.

For a vector $v$, we use $v[i]$ to refer to its $i$-th coordinate, and denote $\supp(v)= \{i: v[i]\neq 0\}$, $\|v\|_0 = |\supp(v)|$, $\|v\|_\infty = \max_i |v[i]|$, and $\|v\|_1 = \sum_i |v[i]|$. We say $v$ is $s$-sparse if $\|v\|_0\le s$.

The \emph{convolution} of two length-$N$ vectors $u,v$ is defined as the vector $u\conv v$ of length $2N-1$ with
\[ (u\conv v)[k] = \sum_{\substack{i,j\in [N]\\i+j=k}}u[i]\cdot v[j].\]
The \emph{cyclic convolution} $u\conv_m v$ is the length-$m$ vector with 
\[ (u\conv_m v)[k] = \sum_{\substack{i,j\in [N]\\i+j\equiv k\,(\mathrm{mod}\, m)}}u[i]\cdot v[j].\]

In this paper we often work with \emph{multisets}, in which an element is allowed to occur multiple times. For a multiset $A\subseteq [N]$,\footnote{By ``multiset $A \subseteq [N]$'' we mean each $a\in A$ is in $[N]$ but may have multiplicity greater than $1$.} let $1_A \in \Z_{\geq 0}^N$ be its indicator vector where $1_A[i]$ denotes the multiplicity of $i$ in $A$. Following this notion, the size of the multiset $A$ is~$|A|=\norm{1_A}_1$, $(1_A \conv 1_B)[i]$ denotes the number of representations $i = a + b$ for~\makebox{$a \in A$} and~\makebox{$b \in B$}, and $(1_A \conv_m 1_B)[i]$ denotes the number of representations $i \equiv a + b \mod m$ for $a \in A$ and~\makebox{$b \in B$}.

The \emph{additive energy} between two sets $A,B\subset \Z$ is
\begin{align*}
\sfE(A,B) &= |\{(a,a',b,b')\in A\times A\times B\times B: a+b=a'+b'\}|\\ &= \sum_{s}(1_A\conv 1_{-A})[s]\cdot (1_B\conv 1_{-B})[s] =\sum_{s}((1_A\conv 1_{B})[s])^2.
\end{align*}
 The additive energy of a set $A\subset \Z$ is $\sfE(A)=\sfE(A,A)$.
We have $\sfE(A)\sfE(B) \ge \sfE(A,B)^2$ by Cauchy-Schwarz inequality.

We need the deterministic nonnegative sparse convolution algorithm by Bringmann, Fischer, and Nakos \cite{BringmannFN22}.
\begin{lemma}[Sparse convolution {\cite[Theorem 1]{BringmannFN22}}]
    \label{thm:detnonnegsparsecovo}
   Given two nonnegative vectors  $x,y\in \N^{N}$, there is a deterministic algorithm that computes $x\conv y$ in $O\big (\|x\conv y\|_0 \cdot\polylog(N\|x\|_\infty\|y\|_\infty)\big )$ time.
\end{lemma}

We need the following partial convolution algorithm by Fischer \cite{nickconstellation}.
\begin{lemma}[Partial convolution {\cite[Theorem 4.1]{nickconstellation}}]
    \label{lem:nickcount}
   Given sets $A,B,C\subseteq [N]$, and $x,y\in \Z^N$ with $\supp(x)\subseteq A, \supp(y)\subseteq B$, there is a deterministic algorithm that computes $(x\conv y)[c]$ for all $c\in C$ in $O\big ((|A|+|C-B|)\cdot \polylog(N\|x\|_\infty \|y\|_\infty)\big )$ time.
\end{lemma}
\section{Tool I: Deterministic Approximate 3SUM Counting}

The goal of this section is to prove \cref{thm:introdeterministicapprox3sum}, restated below.

\thmdeterministicapproxpop*

\newcommand\thresh{\operatorname{thresh}}
\newcommand\coll{\operatorname{coll}}

As an immediate corollary, this theorem implies the following result for 3SUM:

\begin{corollary}[Approximate 3SUM Counting] \label{lem:detapx3sumcount-lemma}
Let multisets $A, B \subseteq [N]$, set $C\subseteq [N]$, and let $\epsilon > 0$. There is a deterministic algorithm that runs in time $\widetilde\Order((\epsilon^{-1} |A| + |B| + |C|) \cdot N^{\order(1)})$ and computes approximations $f(c)$, for all $c \in C$, such that
\begin{equation*}
    |f(c) - (1_A \conv 1_B)[c]| \leq \epsilon |B|.
\end{equation*}
\end{corollary}

For the proof of \cref{thm:introdeterministicapprox3sum}, we rely on the following weaker lemma to compute approximations with bounded total error but unbounded coordinate-wise error:

\begin{lemma} \label{lem:3sum-approx-fft}
Let $A, B \subseteq [N]$ be multisets, let $C \subseteq [N]$ be a set and let $\epsilon > 0$. There is a deterministic algorithm that runs in time $\widetilde\Order((\epsilon^{-1} |A| + |B| + |C|) \cdot N^{\order(1)})$ and computes an approximation $g(c)$ for each $c \in C$ such that:
\begin{itemize}
    \item $g(c) \geq (1_A \conv 1_B)[c]$, and
    \item $\sum_{c \in C} g(c) - (1_A \conv 1_B)[c] \leq \min(\epsilon |B| \, |C|, |A| \, |B| \log N)$.
\end{itemize}
\end{lemma}
\begin{proof}
We follow the deterministic hashing procedure as by Chan and Lewenstein~\cite{ChanL15} and later applied in~\cite{ChanHe,fischer3sum}. The idea is to select a modulus $m \approx \epsilon^{-1}|A| + |C|$ as the product of several small primes. Our goal is to achieve two objectives: (1) We want to isolate a constant fraction of elements in $C$ modulo $m$, and (2) we want to reduce the number of \emph{pseudo-solutions} modulo $m$ (that is, triples $(a, b, c) \in A \times B \times C$ with $a + b \neq c$ and $a + b \equiv c \mod m$) to $\approx \epsilon |B| \, |C|$. Then, for all isolated elements $c \in C$ we can approximate $(1_A \conv 1_B)[c]$ by the number of pseudo-solutions involving $c$ with total error $\approx \epsilon |B| \, |C|$.

We now describe this algorithm in detail. Let $P$ be a parameter to be fixed later. Throughout we maintain a set~\makebox{$C^* \subseteq C$} which is initially~\makebox{$C^* \gets C$}. We call a pair $(c^*, c) \in C^* \times C$ a \emph{collision under $m$} if $c^* \neq c$ and~\makebox{$c^* \equiv c \mod m$}, and denote the number of collisions under $m$ by $\coll_{C^*}(m)$. We say that $c^*$ is \emph{isolated under $m$} if there is no collision $(c^*, c)$ under $m$. Note that given any $m$ we can compute $\coll_{C^*}(m)$ in time $\widetilde\Order(|C| + m)$. While $C^* \neq \emptyset$, we repeat the following steps:
\begin{enumerate}
    \item Initialize $m := 1$
    \item  Repeat the following steps $R := \ceil{\log(\epsilon^{-1} |A| \log N + 2 |C|) / \log P}$ times:
    \begin{enumerate}[label=2.\arabic*]
        \item \label{item:stepselectp} Test all primes $p \in [10(P \log_P N) \log(P\log_P N), 20(P \log_P N) \log(P\log_P N)]$ exhaustively and select one
        \begin{itemize}[label=--]
            \item such that $\coll_{C^*}(m p) \leq \frac{\coll_{C^*}(m)}{P}$, and
            \item such that $\#\set{ (a, b, c) \in A \times B \times C : a + b \equiv c \mod{m p}}$ is minimized\\(note that this quantity can be computed by FFT in time $\widetilde\Order(m p)$)
        \end{itemize}
        \item Update $m := m \cdot p$
    \end{enumerate}
    \item Compute $g'(c) = \#\set{(a, b) \in A \times B : a + b \equiv c \mod m}$ for all $c \in C$ by FFT in time $\widetilde O(m)$
    \item Compute the subset of elements $c^* \in C^*$ that is isolated under $m$
    \item For each isolated element $c^* \in C^*$:
    \begin{enumerate}[label=5.\arabic*.]
        \item Set $g(c^*) := g'(c^*)$
        \item Remove $c^*$ from $C^*$
    \end{enumerate}
\end{enumerate}

\medskip\par\noindent
\emph{Correctness.}
As a first step we argue that (1) we can always select a prime in step~\ref{item:stepselectp} satisfying the two conditions, and (2) for the final modulus (obtained when step 2 ends) we have that
\begin{equation*}
    S(m) := \#\set{(a, b, c) \in A \times B \times C : a + b \neq c, a + b \equiv c \mod m} \leq \frac{\epsilon}{\log N} \cdot |B| \, |C|.
\end{equation*}

By a crude quantitative version of the Prime Number Theorem, there are at least $4P\log_P N$ primes in the interval $[10(P \log_P N) \log(P\log_P N), 20(P \log_P N) \log(P\log_P N)]$ (see e.g.~\cite[Corollary~3]{RosserS62}).
So focus on any inner iteration of step~\ref{item:stepselectp}, and suppose that we select prime $p$ from this interval uniformly at random. For (1), for any fixed pair $(c_0, c)$ ($c_0\neq c$), since $c-c_0$ has at most $\log_P N$ prime factors larger than $P$, we know $(c_0,c)$ becomes a collision modulo $p$ with probability at most~$\frac{\log_P N}{4P (\log_P N)} = \frac{1}{4P}$. Thus, the expected number of collisions modulo $p$ among the $\coll(m)$ collisions modulo $m$ is at most~\smash{$\frac{\coll(m)}{4P}$}, and by Markov's bound the prime satisfies~\smash{$\coll(mp) \leq \frac{\coll(m)}{P}$} with probability at least $\frac34$. For the final modulus $m$, in particular we have that
\begin{equation*}
    \coll(m) \leq \frac{|C^*| \, |C|}{P^R} \leq \frac{|C^*| \, |C|}{2|C|} \leq \frac{|C^*|}{2}.
\end{equation*}

For (2), one can similarly argue about the number of false positives. A random prime $p$ reduces the number of false positives modulo $m p$ in expectation by a factor $\frac{1}{4P}$ and thus by a factor $\frac{1}{P}$ with probability at least $\frac34$. By a union bound both events happen simultaneously with positive probability. In particular, we can pick the prime minimizing
\begin{equation*}
    \#\set{(a, b, c) \in A \times B \times C : a + b \equiv c \mod{mp}},
\end{equation*}
which equivalently minimizes
\begin{equation*}
    \#\set{(a, b, c) \in A \times B \times C : a + b \neq c, a + b \equiv c \mod{mp}},
\end{equation*}
to ultimately obtain that
\begin{equation*}
    S(m) \leq \frac{|A| \, |B| \, |C|}{P^R} \leq \frac{|A| \, |B| \, |C|}{\epsilon^{-1} |A| \log N} = \frac{\epsilon}{\log N} \cdot |B| \, |C|.
\end{equation*}

Next, we analyze the number of iterations of the algorithm. From the bound on the number of collisions above we infer that at least half of the elements in $C^*$ are isolated. As in each iteration we remove all isolated elements from $C^*$, it follows that the number of iterations is bounded by $\log |C| \leq \log N$.

Finally, we analyze the approximations $g(c^*)$. Focus on any iteration of the algorithm. On the one hand, since we set $g(c^*) := g'(c^*)$ which equals the pseudo-solutions equal to $c^*$, we must have that $g(c^*) \geq (1_A \conv 1_B)[c^*]$. On the other hand, note that each pseudo-solution contributes at most $1$ to the sum $\sum_{\text{$c^* \in C^*$ isolated}} g'(c^*)$ and hence
the total contribution of the pseudo-solutions is at most $\frac{\epsilon}{\log N} \cdot |B| \, |C|$. But as in each iteration we only assign the values $g(c^*)$ for the isolated elements~\makebox{$c^* \in C^*$} and as the algorithm terminates after at most $\log N$ iterations, it follows that indeed
\begin{equation*}
    \sum_{c \in C} g(c) - (1_A \conv 1_B)[c] \leq \epsilon |B| \, |C|.
\end{equation*}

\medskip\par\noindent
\emph{Running Time.}
It remains to analyze the running time. The running time of step 2 can be bounded by~\smash{$\widetilde\Order(|A| + |B| + |C| + m P)$}. The remaining steps~3,~4 and~5 can easily be implemented in time $\widetilde\Order(|A| + |B| + |C| + m)$. Taking the $\log N$ iterations into account, and bounding
\begin{equation*}
    m \leq \Order((P \log_P N) \log(P\log_P N))^R \leq (\epsilon^{-1} |A| \log N + |C|)^{1+\Order(\frac{\log\log N}{\log P})},
\end{equation*}
the total running time can be bounded by
\begin{equation*}
    \Order(P \cdot (\epsilon^{-1} |A| \log N + |B| + |C|)^{1+\Order(\frac{\log\log N}{\log P})}) = \Order((\epsilon^{-1} |A| + |B| + |C|) \cdot N^{\order(1)})
\end{equation*}
by choosing, say, $P = 2^{\sqrt{\log N}}$ (here we assumed that $|A|,|B|,|C|, \eps^{-1}\le N$; if this is not the case, then we could simply use FFT to exactly compute $1_A\conv 1_B$ in $O(N\log N)$ time, which is within the desired time bound).
\end{proof}

\begin{proof}[Proof of \cref{thm:introdeterministicapprox3sum}]
We design a recursive algorithm. In the base case, when $N$ is smaller than some constant threshold we can trivially solve the problem in constant time. For the recursive case, let $r$ be a parameter to be determined later and let $N' = N / r$ (assuming for simplicity that $r$ divides $N$). Consider the multisets $A' = (A \bmod N') = \set{a \bmod N' : a \in A}$ and $B' = (B \bmod N')$ (in particular, we have that $|A'| = |A|$ and $|B'| = |B|$). We compute recursively, with error parameter~\makebox{$\epsilon' = \frac{\epsilon}{16}$}, a length-$2N'$ vector $f'$ that approximates $1_{A'} \conv 1_{B'}$. Then in order to compute the desired vector $f$ we proceed as follows; here we write $\thresh_\alpha(x)$ for the function that is $x$ for all~\makebox{$x \geq \alpha$} and zero everywhere else.
\begin{enumerate}
    \item Initialize $f(c) := 0$ for all $c \in [2N]$
    \item Initialize $X := X_0 := \set{x \in [N'] : f'(x) + f'(x + N') > 0}$
    \item While $X \neq \emptyset$:
    \begin{enumerate}[label=3.\arabic*.]
        \item Apply \cref{lem:3sum-approx-fft} on $A, B, C := X + \set{0, N', \dots, (2r-1) N'}$ and with parameter $\frac{\epsilon}{32r}$ to compute a function~$g$
        \item For each $x \in X$ satisfying that $|f'(x) + f'(x + N') - \sum_{i \in [2r]} g(x + i N')| \leq \frac{\epsilon}{4} \cdot |B|$:
        \begin{enumerate}[label=3.2.\arabic*.]
            \item Set $f(x + i N') := \thresh_{\frac{\epsilon}{2} \cdot |B|} (g(x + i N'))$ for all $i \in [2r]$
            \item Remove $x$ from $X$
        \end{enumerate}
    \end{enumerate}
\end{enumerate}

\medskip\par\noindent
\emph{Correctness of $\norm{f - (1_A \conv 1_B)}_\infty \leq \epsilon |B|$.}
The proof is by induction. The base case is clearly correct, so focus on the inductive case. To prove that indeed $|f(c) - (1_A \conv 1_B)[c]| \leq \epsilon \cdot |B|$ we distinguish two cases. First, suppose that $c \not\in X_0 + [2r]$ in which case the algorithm assigns $f(c) = 0$. The induction hypothesis implies that $(1_{A'} \conv 1_{B'})[x] , (1_{A'} \conv 1_{B'})[x + N'] \leq \frac{\epsilon}{16} \cdot |B|$ for $x = c \bmod N'$. It immediately follows that $(1_A \conv 1_B)[c] \leq \frac{\epsilon}{8} \cdot |B|$, and thus the choice $f(c) = 0$ is valid.

Next, assume that $c = x + i N'$ for some $x \in X_0$ and $i \in [2r]$. We set~\smash{$f(c) := \thresh_{\frac{\epsilon}{2} \cdot |B|}(g(c))$} in an iteration when $|f'(x) + f'(x + N') - \sum_{i \in [2r]} g(x + i N')| \leq \frac{\epsilon}{2} \cdot |B|$. So suppose for contradiction that $|f(c) - (1_A \conv 1_B)[c]| > \epsilon \cdot |B|$ and thus $|g(c) - (1_A \conv 1_B)[c]| > \frac{\epsilon}{2} \cdot |B|$. As \cref{lem:3sum-approx-fft} guarantees that $g(c) \geq (1_A \conv 1_B)[c]$ for all $c$, we have
\begin{align*}
    \sum_{i \in [2r]} g(x + i N')
    &> \sum_{i \in [2r]} (1_A \conv 1_B)[x + i N'] + \tfrac{\epsilon}{2} \cdot |B| \\
    &= (1_{A'} \conv 1_{B'})[x] + (1_{A'} \conv 1_{B'})[x + N'] + \tfrac{\epsilon}{2} \cdot |B| \\
    &\geq \parens{f'(x) - \tfrac{\epsilon}{16} \cdot |B|} + \parens{f'(x + N') - \tfrac{\epsilon}{16} \cdot |B|} + \tfrac{\epsilon}{2} \cdot |B| \\
    &\geq f'(x) + f'(x + N') + \tfrac{\epsilon}{4} \cdot |B|,
\end{align*}
leading to a contradiction.

\medskip\par\noindent
\emph{Number of Iterations.}
Before we continue with the remaining correctness proof, let us first bound the number of (outer) iterations of the above algorithm by $\log N$. To this end, we show that in every iteration at least half of the elements $x \in X$ are removed. Focusing on any iteration, let us call an element $x \in X$ \emph{bad} if
\begin{equation*}
    \sum_{i \in [2r]} g(x + iN') > \sum_{i \in [2r]} (1_A \conv 1_B)[x + iN'] + \frac{\epsilon}{8} \cdot |B|,
\end{equation*}
and \emph{good} otherwise. Given that \cref{lem:3sum-approx-fft} (applied with parameter $\frac{\epsilon}{32r}$) guarantees
\begin{equation*}
    \sum_{c \in [2N]} |g(c) - (1_A \conv 1_B)[c]| \leq \frac{\epsilon}{32r} \cdot |B| \, |C| = \frac{\epsilon}{16} \cdot |B| \, |X|,
\end{equation*}
there can be at most $(\frac{\epsilon}{16} \cdot |B| \, |X|) / (\frac{\epsilon}{8} \cdot |B|) \leq \frac12 \cdot |X|$ bad elements. Moreover, for any good element~$x$ we have that
\begin{multline*}
    \abs*{f'(x) + f'(x + N') - \sum_{i \in [2r]} g(x + i N')} \\
    \leq \abs*{(1_{A'} \conv 1_{B'})[x] + (1_{A'} \conv 1_{B'})[x + N'] - \sum_{i \in [2r]} (1_A \conv 1_B)[x + i N']} + \frac{\epsilon}{16} + \frac{\epsilon}{16} + \frac{\epsilon}{8},
\end{multline*}
and thus $x$ is indeed removed in the current iteration.

\medskip\par\noindent
\emph{Correctness of $\norm{f}_0 \leq \Order(\epsilon^{-1} |A| \log^2 N)$.}
To prove the sparsity bound on $f$, observe that in any iteration we add total mass at most
\begin{equation*}
    \sum_{c \in [2N]} g(c) \leq \Order(|A| \, |B| \log N),
\end{equation*}
where the bound is due to \cref{lem:3sum-approx-fft}. Therefore, across all $\log N$ iterations we add mass at most $\sum_{c \in [2N]} f(c) \leq \Order(|A| \, |B| \log^2 N)$. However, recall that every nonzero entry in $f$ is at least $\frac{\epsilon}{2} \cdot |B|$ due to the thresholding; the bound on the sparsity of $f$ follows.

\medskip\par\noindent
\emph{Running Time.}
We finally analyze the running time of this algorithm. Ignore the cost of recursive calls for now and focus only on the above algorithm. As we showed before it runs in $\log N$ iterations, each of which is dominated by the call to \cref{lem:3sum-approx-fft} in time
\begin{align*}
    \widetilde\Order((r \epsilon^{-1} |A| + |B| + |C|) \cdot N^{\order(1)}) &= \widetilde\Order((r \epsilon^{-1} |A| + |B| + r |X|) \cdot N^{\order(1)}) \\
    &= \widetilde\Order((r \epsilon^{-1} |A| + |B|) \cdot N^{\order(1)}),
\end{align*}
using that $|X| \leq \Order(|A| / \epsilon \log^2 N)$ as established in the previous paragraph. The remaining steps also run in time $\Order(r |X|)$ and can thus be similarly bounded. To take the recursion into account, let $T(N, \epsilon)$ denote the running time of this algorithm (observing that the sizes of $A$ and $B$ do not chance in recursive calls). Then:
\begin{align*}
    T(N, \epsilon) &\leq T(\tfrac{N}{r}, \tfrac{\epsilon}{16}) + \widetilde\Order((r \epsilon^{-1} |A| + |B|) \cdot N^{\order(1)}),
\intertext{which solves to}
    T(N, \epsilon)
    &\leq \widetilde\Order(\log_r(N) \cdot r \cdot 16^{\log_r(N)} \cdot (\epsilon^{-1} |A| + |B|) \cdot N^{\order(1)}) \\
    &\leq \Order((\epsilon^{-1} |A| + |B|) \cdot N^{\order(1)})
\end{align*}
by setting $r = 2^{\sqrt{\log N}}$.
\end{proof}

An application of \cref{thm:introdeterministicapprox3sum} is approximately counting 4SUM solutions. In particular, we can efficiently approximate the additive energy of a set, which will be useful later for the derandomized algorithmic BSG theorem.

\begin{lemma}[Deterministic approximate 4SUM counting]
\label{lem:approx4sumcount}
   Given multisets $A,B,C,D\subseteq [N]$ of size at most $N$, we can deterministically approximate $|\{(a,b,c,d)\in A\times B\times C\times D: a+b=c+d\}|$ up to additive error $\eps (|A|+|C|) |B||D|$ in $O\big (\big (\eps^{-1}(|A|+|C|) + |B| + |D|\big )\cdot N^{o(1)}\big )$ time.
\end{lemma}
\begin{corollary}\label{cor:approxenergy}
   Given a set $A\subseteq [N]$, we can deterministically approximate its additive energy $\sfE(A)$ up to additive error $\eps |A|^3$ in $O(\eps^{-1}|A| N^{o(1)})$ time.
\end{corollary}
\begin{proof}[Proof of \cref{lem:approx4sumcount}]
   Using  \cref{thm:introdeterministicapprox3sum} we can compute an $O(\eps^{-1}|A|\log^2 N)$-sparse nonnegative vector $f$ such that $\|f- 1_A\conv 1_B\|_\infty\le \eps |B|$ in $O((\eps^{-1}|A|+|B|)N^{o(1)})$ time, and an $O(\eps^{-1}|C|\log^2 N)$-sparse nonnegative vector $g$ such that 
   $\|g- 1_C\conv 1_D\|_\infty\le \eps |D|$
in $O((\eps^{-1}|C|+|D|)N^{o(1)})$ time.
   Then we return $\sum_{i}f[i] g[i]$ as our approximation, which has additive error
   \begin{align*}
    & \Big \lvert    \sum_{i}(1_A\conv 1_B)[i](1_C\conv 1_D)[i] - \sum_{i}f[i] g[i] \Big \rvert \\
   &\le \underbrace{\Big \lvert      \sum_{i}(1_A\conv 1_B)[i](1_C\conv 1_D)[i] - \sum_{i}f[i] (1_C\conv 1_D)[i] \Big \rvert}_{=:e_1} + \underbrace{\Big \lvert \sum_{i}f[i] (1_C\conv 1_D)[i] - \sum_{i}f[i] g[i] \Big \rvert}_{=:e_2},
   \end{align*}
   where \begin{align*}
  e_1 &=   
  \Big \lvert      
   \sum_{i}\big ((1_A\conv 1_B)[i]-f[i]\big )\cdot (1_C\conv 1_D)[i]
   \Big \rvert\\
  & \le \|1_A\conv 1_B-f\|_\infty \cdot \sum_{i} (1_C\conv 1_D)[i] \\
  & \le  \eps |B|\cdot |C||D|,
   \end{align*}
   and
   \begin{align*}
  e_2 &=  \Big \lvert    \sum_{i}f[i]\cdot \big ( (1_C\conv 1_D)[i] -g[i]\big )\Big \rvert\\
  & \le \|1_C\conv 1_D-g\|_\infty\cdot \sum_{i} f[i] \\
  & \le \eps|D| \cdot \sum_{i: f[i]\neq 0} \Big( \big \lvert f[i] - (1_A\conv 1_B)[i]\big \rvert + (1_A\conv 1_B)[i]\Big)\\
  & \le \eps |D| \cdot \big ( \|f\|_0 \cdot \|f-1_A\conv 1_B\|_\infty + \sum_i (1_A\conv 1_B)[i] \big)\\
  & \le  \eps |D|\cdot \big ( \eps^{-1}|A|\log^2 N \cdot \eps |B| +  |A||B| \big )\\
  & \le  O(\eps |A||B||D|\log^2 N).
  \end{align*}
  So the total error is at most $e_1+e_2 = O(\eps |B||D|(|A|+|C|)\log^2 N)$. The error can be reduced to $\eps|B||D|(|A|+|C|)$ after decreasing $\eps$ by only an $O(\log^2 N)$ factor.  The total time complexity is still $O((\eps^{-1}|A|+|B|)N^{o(1)}) + O((\eps^{-1}|C|+|D|)N^{o(1)})$ as claimed.
\end{proof}

\section{Tool II: Deterministic 3SUM Counting for Small Doubling Sets}

The goal of this section is to prove the following theorem.

\begin{theorem}
    \label{thm:detsmalldoublethreesum}
   Given sets $A,B,C,S\subset [N]$, we can deterministically compute 
  \begin{enumerate}
    \item $(1_A\conv 1_B)[c]$ for all $c\in C$, and
        \item $(1_C\conv 1_{-B})[a]$ for all $a\in A$,
  \end{enumerate} 
   in time complexity \[O\Big(\frac{|A+S|\sqrt{|B||C|}}{\sqrt{|S|}}\cdot \polylog(N) 
 + \big (|A+S|+|C|\big)\cdot N^{o(1)} 
   \Big ).\]
   If randomization is allowed, then the $N^{o(1)}$ factor can be replaced by $\polylog(N)$.
\end{theorem}

Randomized versions of similar results were given by Abboud, Bringmann, and Fischer \cite{AbboudBF23}, and Jin and Xu \cite{JinX23}.
To deterministically implement the approaches of \cite{AbboudBF23,JinX23}, we follow the derandomization strategy of \cite{ChanL15,fischer3sum} and work modulo some $M$.
For our purposes, we pick $M \le |S|\cdot N^{o(1)}$ such that $S$ has many distinct remainders modulo $M$, as stated below.
\begin{lemma}
    \label{lem:detfindmodulus}
Given set $S\subseteq [N]$, we can deterministically find an integer $M\in \big [|S|/2, |S|\cdot 2^{O(\sqrt{\log |S|\log \log N})}\big ]$ in $|S|\cdot 2^{O(\sqrt{\log |S|\log \log N})}$ time such that the set $S\bmod M = \{\hat s\in \Z_M: \exists s\in S, s\bmod M = \hat s \}$ has size $|S\bmod M| \ge 0.9|S|$.
\end{lemma}
\begin{proof}
    We use the same technique as \cite{ChanL15}.
    Let $P$ be a parameter to be determined. We will iteratively pick primes $p_1,\dots,p_k \in [P,2P]$ and let the final modulus be $M=p_1p_2\cdots p_k$.  At iteration $i$ $(i\ge 1)$, the primes $p_1,\dots,p_{i-1}$ are already fixed, and we pick prime $p_i$ to minimize \[h(S,p_1p_2\cdots p_i):= \sum_{x,y\in S,x\neq y} [x\equiv y\pmod{p_1p_2\cdots p_i}],\] i.e., the number of collision pairs in $S$ modulo $p_1p_2\cdots p_i$. This can be done by enumerating all primes in $[P,2P]$, and computing $\sum_{x,y\in S}[x\equiv y\pmod{p_1p_2\cdots p_i}]$ in $\tilde O(|S|)$ time by bucketing the elements of $S$ modulo $p_1p_2\cdots p_i$ and summing up the bucket sizes squared. We terminate and return $M=p_1\cdots p_i$ once the number of collision pairs $h(S,M)\le 0.1|S|$, in which case we must have $|S\bmod M| \ge |S|-h(S,M)\ge 0.9|S|$  and $M\ge \frac{|S|^2}{h(S,M)+|S|} > |S|/2$ as desired.  The total time complexity is $\tilde O(kP|S|)$ where $k$ is the number of iterations before termination, and $M\le (2P)^k$. We now bound $k$.

    Consider any pair of $x,y\in S,x\neq y$ where $x\equiv y \pmod{p_1p_2\cdots p_{i-1}}$. By the prime number theorem, over a random prime $p_i \in [P,2P]$, the probability that $x\equiv y \pmod{p_1p_2\cdots p_{i-1}p_i}$ holds is at most $O(\frac{\log_P N}{P/\log P}) = O(\frac{\log N}{P})$. Hence $\Ex_{p_i\in [P,2P]}[h(S,p_1\cdots p_i)] \le \frac{c\log N}{P} h(S,p_1,\cdots,p_{i-1})$ by linearity of expectation (for some absolute constant $c>0$). By averaging, the minimizing $p_i\in [P,2P]$ that we choose satisfies $h(S,p_1\cdots p_i) \le  \frac{c\log N}{P} h(S,p_1\cdots p_{i-1})$. Since $h(S,1) = |S|(|S|-1)$, and the second to last iteration has $h(S,p_1\cdots p_{k-1})>0.1|S|$, we have
    \[ k\le 1 + \frac{\log \frac{|S|(|S|-1)}{0.1|S|}}{\log \frac{P}{c\log N}} < 1 + \frac{\log (10|S|)}{\log \frac{P}{c\log N}},\]
    and hence
    \[ M \le (2P)^k < 2P\cdot  (10|S|)^{\log(2P)/\log (\frac{P}{c\log N})} = 2P\cdot (10|S|)^{1 + \frac{\log(2c\log N)}{\log (P/(c\log N))}}.\]
    By choosing $P$ such that $\log P =\Theta(\sqrt{\log |S|\log \log N})$, we have $M \le |S|\cdot 2^{O(\sqrt{\log |S|\log \log N})} \le |S|\cdot N^{o(1)}$, and the time complexity is $\tilde O(kP|S|) \le |S|\cdot 2^{O(\sqrt{\log |S|\log \log N})} \le |S|\cdot N^{o(1)}$.
\end{proof}
The following lemma finds an efficient covering of $\Z_M$ using shifted copies of a set.
\begin{lemma}
    \label{lem:detmodmcover}
Given a set $\hat S\subseteq \Z_M$, we can deterministically find a set $\Delta \subseteq \Z_M$ in $\tilde O(M)$ time such that $\hat S+\Delta = \Z_M$ and $|\Delta|\le 
(2M\ln M)/|\hat S| $.
\end{lemma}
\begin{proof}
   Let $S_0=\hat S,\Delta_0 = \{0\}$. Repeat the following for $i=0,1,2,\dots$, until $S_i=\Z_M$: 
   \begin{itemize}
    \item Find $\delta_{i+1} \in \Z_M$ such that $|S_i \cup (\delta_{i+1}+S_i)|$ is maximized.
        \item Let $S_{i+1} := S_i\cup (\delta_{i+1}+S_i)$ and $\Delta_{i+1}:= \Delta_i + \{0,\delta_{i+1}\}$.
   \end{itemize}
   By induction, $S_i = \hat S + \Delta_i$ for all $i$. When the procedure terminates with $S_d=\Z_M$ at iteration $i=d$, we return $\Delta_d = \{0,\delta_1\}+\dots + \{0,\delta_d\}$.
   
   We first explain how to implement each iteration in $O(M\log M)$ time. Given $S_i\subseteq \Z_M$, we use FFT to compute the cyclic convolution $f = 1_{S_i} \conv_M 1_{-S_i}$ in $O(M\log M)$ time, where $f[\delta] =  |(S_i+\delta)\cap S_i|$ for all $\delta\in \Z_M$, and we pick $\delta_{i+1}=\delta$ that minimizes $f[\delta]$ and hence maximizes $|S_i\cup (\delta+S_i)| = 2|S_i| - f[\delta]$. Then, computing $S_{i+1},\Delta_{i+1}$ takes $O(M)$ time.
   
   Now we analyze the number of iterations $d$.
  Let $\sigma_i = \frac{M-|S_i|}{M} $ be the fraction of $\Z_M$ not covered by $S_i$.
   Observe that \[\Ex_{\delta\in \Z_M}\big \lvert\Z_M \setminus \big (S_i \cup (\delta+S_i)\big )\big \rvert = \sum_{x\in \Z_M \setminus S_i}\Pr_{\delta\in \Z_M}[x\notin (\delta+S_i)] = |\Z_M\setminus S_i|\cdot \sigma_i = \sigma_i^2 M,\] so  $\sigma_{i+1}\le \sigma_i^2$ by averaging. 
   Since $S_{d-1} \neq \Z_M, \sigma_{d-1}\ge \frac{1}{M}$.
   Hence, the number of iterations is \[d\le 1+\log_2\frac{\ln(1/\sigma_{d-1})}{\ln (1/\sigma_0)} 
 \le  \log_2 \frac{2\ln M}{-\ln  (1- |\hat S|/M)}  
 \le  \log_2 \frac{2M\ln M}{|\hat S|}, \] and $|\Delta|\le 2^d\le (2M\ln M)/|\hat S|$. The total time complexity is $O(dM\log M) = \tilde O(M)$.
\end{proof}

Now we prove \cref{thm:detsmalldoublethreesum}.

\begin{proof}[Proof of \cref{thm:detsmalldoublethreesum}]
    Apply \cref{lem:detfindmodulus} to set $S$ and in $|S|\cdot 2^{O(\sqrt{\log |S|\log \log N})} = |S|\cdot N^{o(1)}$ time obtain a modulus $M \in [|S|,|S|\cdot 2^{O(\sqrt{\log |S|\log \log N})}]$  such that $\hat S := S\bmod M \subseteq \Z_M$ has size $|\hat S|  \ge |S|/2$.
    Apply \cref{lem:detfindmodulus} to $\hat S$ and obtain set $\Delta \subseteq \Z_M$ in $\tilde O(M)$ time such that $\hat S+\Delta= \Z_M$ and $|\Delta| \le O(\frac{M\log M}{|\hat S|})\le O(\frac{M\log M}{|S|}) \le 2^{O(\sqrt{\log |S|\log \log N})}$.
Compute the sumset $A+S\subset \Z$ in deterministic $O(|A+S|\polylog N)$ time via \cref{thm:detnonnegsparsecovo}. Let $K = |A+S|/|S|$.
\begin{claim}
    \label{claim:partitionc}
   In $\tilde O(KM +|C|)$ time, we can deterministically compute a partition of $C$ into $p = O(\log M)$ parts, $C=\bigcup_{i\in [p]} C^{(i)}$, and shifts $\phi^{(0)},\dots,\phi^{(p-1)}\in \Z_M$, such that for every $i \in [p]$,
   \begin{equation}
       \label{eqn:forallcmultbound}
 ( 1_{A+S}\conv_M 1_{\Delta+\phi^{(i)}})[c\bmod M] \le O(K\log M) \text{ for all } c\in C^{(i)}. 
   \end{equation}
\end{claim}
\begin{proof}

Compute the cyclic convolution $1_{A+S}\conv_M 1_\Delta$ by brute force in $O(|A+S|\cdot |\Delta|) = O(KM\log M)$ time.
We have
\[
\Ex_{\hat c\in \Z_M}\big[(1_{A+S}\conv_M 1_\Delta)[\hat c] \big ] =|A+S||\Delta|/M. 
\]
By Markov's inequality, the set \[\hat C_{\text{good}}:= \{\hat c\in \Z_M : (1_{A+S}\conv_M 1_\Delta)[\hat c] \le 2|A+S||\Delta|/M\}\] has size $|\hat C_{\text{good}}|\ge M/2$.
Then, we apply \cref{lem:detmodmcover} to $\hat C_{\text{good}}\subseteq \Z_M$ and compute in $\tilde O(M)$ time a set of shifts $\Phi\subseteq \Z_M$, $|\Phi|\le O((M\log M)/|\hat C_{\text{good}}|) \le O(\log M)$, such that $\Phi + \hat C_{\text{good}} = \Z_M$. So for every $\hat c\in \Z_M$ we can pick a shift $\varphi_{\hat c}\in \Phi$ such that $\hat c\in \varphi_{\hat c}+\hat C_{\text{good}}$.
Now we define a partition $C = \bigcup_{\phi \in \Phi} C_\phi$ where $C_\phi := \{c\in C: \varphi_{c\bmod M}=\phi\}$. 
Then, for all $c\in C_\phi$, we have $(c\bmod M)-\phi \in \hat C_{\text{good}}$, and hence 
\[ ( 1_{A+S}\conv_M 1_{\Delta+\phi})[c\bmod M] = ( 1_{A+S}\conv_M 1_{\Delta})[(c\bmod M) - \phi]  \le 2|A+S||\Delta|/M \le O(K\log M).  \] 
Therefore, we can return $\{(C_\phi,\phi)\}_{\phi \in \Phi}$ as the desired $\{(C^{(i)}, \phi^{(i)})\}_{i\in [p]}$, where $p=|\Phi| = O(\log M)$.
\end{proof}

Let $\{(C^{(i)}, \phi^{(i)})\}_{i\in [p]}$ (where $p=O(\log M)$) be returned by 
\cref{claim:partitionc}.
In the following, we will solve the \#3SUM instance $(A,B,C^{(i)})$ separately for each $i\in [p]$, i.e., we compute the desired answers defined in the statement of \cref{thm:detsmalldoublethreesum} but with $C^{(i)}$ replacing $C$. Since $C$ is the disjoint union of all $C^{(i)}$, in the end we can easily merge these answers over all $i \in [p]$ to get the answers for the original \#3SUM instance $(A,B,C)$. This affects the total time complexity by at most a factor of $p=O(\log M)$.
From now on we focus on a particular pair of $C^{(i)}\subset \Z$ and $\phi^{(i)}\in \Z_M$.

Note that \[ (S + \Delta+\phi^{(i)})\bmod M = \hat S + \Delta + \phi^{(i)} = \Z_M + \phi^{(i)} = \Z_M.\]
Hence, for every $b\in \Z$, there exists $s_b \in S$ such that $(b - s_b) \bmod M \in \Delta + \phi^{(i)}$.
We pick such an $s_b$ for every $b\in B$, in total time $\tilde O(|S||\Delta+\phi^{(i)}| + |B|) = \tilde O(M+|B|)$ time.
Denote $d_b:= b - s_b$, which satisfies \[d_b\bmod M \in \Delta+\phi^{(i)}.\]
Note that any 3SUM solution $(a,b,c)\in A\times B\times C^{(i)}$ where $a+b=c$  must satisfy $d_b = (c-a)-s_b\in c- (A+S)$.
This motivates the definition of the set \[D_c:= \{ d\in c - (A+S): d\bmod M \in \Delta+\phi^{(i)}\}\]
for $c\in C^{(i)}$.
In this way, we can verify that the 3SUM solutions are decomposed into a disjoint union as follows  (where we naturally define $B_d:=\{b\in B:d_b=d\}, C_d:= \{c\in C^{(i)}: d\in D_c\}$),
\begin{equation}
    \label{eqn:decompose3sumsols}
 \{(a,b,c)\in A\times B\times C^{(i)}: a+b=c\} = \bigcup_{d\in \Z} \{(a,b,c)\in A\times B_d\times C_d: a+b=c\}.
\end{equation}
Note that $B_d\in d+S$ by the definition of $d_b$.

Before describing our final \#3SUM algorithm (which relies on the decomposition in \cref{eqn:decompose3sumsols}), we first need to give size upper bounds for the sets $C_d$, and show how to compute the sets $C_d$ efficiently.  
\begin{claim}
$\sum_{d\in \Z} |C_d| =  O(|C|K\log M)$.
\label{claim:cdsize}
\end{claim}
\begin{proof}
For $c\in C^{(i)}$,     by the definition of $D_c$, we have 
    \[ |D_c| = \{ x\in A+S: (c-x)\bmod M \in \Delta+\phi^{(i)}\} =  
 ( 1_{A+S}\conv_M 1_{\Delta+\phi^{(i)}})[c\bmod M]\underset{\text{by \cref{eqn:forallcmultbound}}}{\le} O(K\log M). \]
 Then, $\sum_{d\in \Z: d\bmod M\in \Delta+\phi^{(i)}} |C_d| = \sum_{c\in C^{(i)}} |D_c| \le |C^{(i)}| \cdot O(K\log M) = O(|C|K\log M)$.
\end{proof}

\begin{claim}
    We can compute all non-empty sets $C_d$ where $d\bmod M\in \Delta+\phi^{(i)}$,  in time 
    $\tilde O(|C||\Delta|+|A+S|)$ in addition to the output size.
\end{claim}
\begin{proof}
    Equivalently, we need to compute the sets $D_c$ for all $c\in C^{(i)}$.
To do this, by the definition of $D_c$, it suffices to list all triples in the following set,
\[ \{(\hat d,c,x) \in (\Delta+\phi^{(i)})\times C^{(i)} \times (A+S): \hat d = (c-x)\bmod M \}.\]
This can be viewed as listing all solutions in an instance of 3SUM modulo $M$.
It can be done by enumerating pairs $(\hat d,c)\in (\Delta+\phi^{(i)})\times C^{(i)}$ and searching for $x\in A+S$, in total time $\tilde O(|\Delta||C^{(i)}| + |A+S|) \le \tilde O(|\Delta||C|+|A+S|)$ in addition to the output size.
\end{proof}

Now the \#3SUM algorithm for $A,B,C^{(i)}$ works as follows based on the decomposition in \cref{eqn:decompose3sumsols}: for every $d\in \Z$ for which $B_d\neq \emptyset$ and $C_d \neq \emptyset$, we solve \#3SUM on $A,B_d,C_d$ in one of the two ways: 
\begin{enumerate}
\item Enumerate pairs in $B_d\times C_d$ and searching in $A$, in $\tilde O(|B_d||C_d|)$ time (after initially preprocessing $A$ into a binary search tree), or,
    \label{item:3sumbf}
\item Compute $1_{A}\conv 1_{B_d}$ using sparse convolution (\cref{thm:detnonnegsparsecovo}) and output $(1_{A}\conv 1_{B_d})[c]$ for all $c\in C_d$, and also compute $(1_{C_d}\conv 1_{-B_d})[a]$ for all $a\in A$ using \cref{lem:nickcount}, both in time complexity $O((|C_d|+ |A+B_d|)\polylog N) =  O((|C_d|+|A+S|)\polylog N)$ deterministically, where we used $B_d\subseteq d+S$.
    \label{item:3sumsparseconvo}
\end{enumerate}
    We choose \cref{item:3sumbf} if $|B_d|\le X$ for some parameter $X\ge 1$, otherwise choose \cref{item:3sumsparseconvo}. Then, the total time over all $d\in \Z$ is at most  (ignoring $\polylog(N)$ factors)
\begin{align*}
\sum_{d\in \Z: |B_d|\le X} |B_d||C_d| + \sum_{d\in \Z: |B_d|>X}(|C_d|+ |A+S|) &\le 
    X\cdot \sum_{d\in \Z}|C_d| + \frac{\sum_{d\in \Z}|B_d|}{X}\cdot |A+S|,
\end{align*}
which can be balanced to
\begin{align*}
     O\Big (\sqrt{\sum_{d\in \Z}|C_d| \sum_{d\in \Z}|B_d|\cdot |A+S|}\Big )
    & =  O\Big (\sqrt{(|C|K\log M)\cdot  |B|\cdot K|S|}\Big ), \tag{by \cref{claim:cdsize}}
\end{align*}
so the time complexity becomes
\[O\Big (K\sqrt{|B||C||S|} \polylog(N)\Big ).\]

To summarize, the total running time over all steps is (recall that $M\le |S|\cdot N^{o(1)}$ and $|\Delta| \le N^{o(1)}$)
\begin{align*}
&  |A+S|\polylog(N) + \tilde O(KM+ |C|) +\tilde O(|C||\Delta|) + K\sqrt{|B||C||S|}\polylog(N)  \\
& \le (|A+S|+|C|)\cdot N^{o(1)} +  K\sqrt{|B||C||S|}\polylog(N),
\end{align*}
as claimed.
\end{proof}
\section{Tool III: Almost Additive Hashing with Short Seed Length}
\label{sec:dethash}
In this section we design a small almost linear hash family $\caH\subset \{h\colon \{-N,\dots,N\}\to \F_p\}$ that satisfies certain independence properties. Such hash family (with large size) was used in the randomized 3SUM reductions of \cite{AbboudBF23} and \cite{JinX23}.
For our deterministic applications, we need to construct a hash  family of small size, so that we can afford to enumerate all functions in the family.

Our construction is inspired by \cite{AbboudBF23} and \cite{JinX23}, and additionally use the $\eps$-biased sets of \cite{AlonGHP92} to reduce the size of the family.
After introducing this  key ingredient in \cref{sec:epsbiased}, we present the definitions and properties of our hash family $\caH$ in \cref{subsec:hash}.
This hash family $\caH$ will be later used in \cref{sec:findsmallsubset},  \cref{sec:modtolowreduction}, and \cref{sec:trianglist}.

\subsection{Almost \texorpdfstring{$k$}{k}-wise independent distribution over \texorpdfstring{$\F_p^m$}{Fpm}}
\label{sec:epsbiased}
As a key ingredient, we will use the almost $k$-wise independent distribution with small sample space constructed by \cite{AlonGHP92}. For our purposes, we need to generalize the binary case of \cite{AlonGHP92} to the mod-$p$ case, following \cite{DBLP:conf/crypto/BierbrauerS00}.
The prime $p$ will be chosen small enough to fit in a machine word, so that addition and multiplication over $\F_p$ take constant time.

For vectors $a,b\in \F_p^m$, let $\langle a,b\rangle \in \F_p$ denote their inner product $\sum_{i=1}^m a_ib_i$.

We say a distribution over $\F_p^{k'}$ is \emph{$\delta$-almost uniform} if it is at most $\delta$ away from the uniform distribution over $\F_p^{k'}$ in $L_1$ distance.
Let $\mu$ be a distribution over $\F_p^m$. We say $\mu$ is \emph{$\delta$-almost $k$-wise independent}, if for every coordinate subset $S\subseteq [m]$ of size $|S|=k'\le k$, the marginal distribution of $\mu$ projected to the coordinates in $S$ is $\delta$-almost uniform.

The following lemma summarizes the construction borrowed from \cite{AlonGHP92,DBLP:conf/crypto/BierbrauerS00}, which will be used in the next section.
 \begin{lemma}
    \label{lem:dethash}
 We can deterministically construct $\caC \subset \F_p^m$ of size $|\caC|\le 4m^2p^{3k}$ in $\tilde O(m^3 p^{3k})$ time such that the following holds:
Suppose we sample $c\in \caC$ uniformly at random. 
\begin{enumerate}
   \item For any $k'\le k$ linearly independent vectors $x_1,x_2,\dots,x_{k'} \in \F_p^m$, the distribution of the $k'$-tuple $(\langle c,x_1\rangle ),\dots,\langle c,x_{k'}\rangle ) \in \F_p^{k'}$ is $\frac{1}{2p^{k-1}}$-almost uniform. 
      \label{item:dethashalmostuniform}
   \item For any $k'<k$ vectors $x_1,x_2,\dots,x_{k'} \in \F_p^m$, $\Pr[\langle c,x_1\rangle =\dots=\langle c,x_{k'}\rangle =0] \ge \frac{1}{2p^{k'}}$. 
      \label{item:dethashlb}
\end{enumerate}
 \end{lemma}
 
In the following we explain how \cref{lem:dethash} follows from known results in \cite{AlonGHP92, DBLP:conf/crypto/BierbrauerS00}. 
 
 For $\gamma \in \F_p^m$, define the Fourier coefficient $  \mu(\gamma) = \sum_{x\in \F_p^m}\mu(x) \omega_p^{-\langle \gamma,x\rangle}$, where $\omega_p = e^{2\pi i/p}$.
We say a distribution $\mu$ over $\F_p^m$ is \emph{$\eps$-biased} if $|  \mu(\gamma)| \le \eps$ for all $\gamma\in \F_p^m \setminus \{ \mathbf{0}\}$.

\begin{lemma}[e.g., {\cite[Corollary 1]{AlonGHP92}}]
   Every $\eps$-biased distribution over $\F_p^m$ is $p^{k/2}\eps$-almost $k$-wise independent.
   \label{lem:epsbiased-to-almostindep}
\end{lemma}

The following construction of $\eps$-biased distribution is a straightforward generalization of the powering construction from \cite{AlonGHP92}.
 \begin{lemma}[{\cite[Section 5]{AlonGHP92}}]
    \label{lem:aghppower}
    Let $\F_q = \F_{p^r}$ be a finite field  represented by a monic irreducible degree-$r$ polynomial over $\F_p$. For $a,b\in \F_q$, define their inner product $\langle a,b\rangle \in \F_p$ by viewing $a,b$ as vectors in $\F_p^r$.

     When $(x,y)$ is uniformly drawn from $\F_{q}\times \F_q$, the distribution of $(\langle 1,y\rangle,\langle x,y\rangle, \dots , \langle x^{m-1},y\rangle )\in \F_p^m$ is $\frac{m-1}{q}$-biased.
 \end{lemma}
 
\begin{lemma}
 For an $\eps$-biased distribution $\mu$ over $\F_p^m$, and an invertible matrix $M\in \F_p^{m\times m}$, the distribution of $Mx\in \F_p^m$ where $x\sim \mu$ is also $\eps$-biased.
 \label{lem:fourierinvert}
\end{lemma}
\begin{proof}
 The new distribution $M\mu$ has Fourier coefficients $\widehat{M\mu}(\gamma) =  \mu(M^{\top}\gamma)$ (where $\gamma \neq \mathbf{0}$), which have magnitude at most $\eps$ since $\mu$ is $\eps$-biased.
\end{proof}

\cref{lem:dethash} follows from putting the three lemmas together.
 \begin{proof}[Proof of \cref{lem:dethash}]
    We prove \cref{item:dethashalmostuniform} first.
Let  $\delta = \frac{1}{2p^{k-1}}$, and $\eps = \delta/p^{k/2} = \frac{1}{2p^{3k/2-1}}$.
Let $r := \lceil \log_p \frac{m}{\eps} \rceil \le  \log_p(2m) + 3k/2$, and $q: = p^r \ge m/\eps$.
 Construct the $\frac{m-1}{q}$-biased (and hence $\eps$-biased) set $\caC \subset \F_p^m$ from \cref{lem:aghppower}, which has size $|\caC| = q^2 =p^{2r} \le 4m^2 p^{3k}$, and can be computed in $|\caC|\cdot m\cdot \tilde O(r) = \tilde O(m^3 p^{3k})$ time (constructing the finite field $\F_{p^r}$ takes  $\tilde O(r^4 p^{1/2})$ time deterministically \cite{Shoup88}, which is not a bottleneck). Now we verify that $\caC$ satisfies the claimed property. Fixing $k'\le k$ linearly independent $x_1,\dots,x_{k'}\in \F_p^m$, there is an invertible $m\times m$ matrix $M$ such that the first $k'$ coordinates of $Mc$ for any $c\in \F_p^m$ are $\langle c, x_1\rangle, \dots, \langle c,x_{k'}\rangle$. By \cref{lem:fourierinvert},  the distribution of $Mc$ where $c$ is uniformly drawn from $\caC$ is also $\eps$-biased, and hence $\delta$-almost $k$-wise independent by \cref{lem:epsbiased-to-almostindep}. In particular, the first $k'$ coordinates of $Mc$, $\langle c, x_1\rangle, \dots, \langle c,x_{k'}\rangle$, are $\delta$-almost uniform. This concludes the proof of \cref{item:dethashalmostuniform}.
 
Now we derive \cref{item:dethashlb}  from \cref{item:dethashalmostuniform}. Pick any maximal linear independent subset of $x_1,\dots,x_{k'}$, and without loss of generality assume it contains $x_1,\dots,x_r$. Since each $x_j$ is a linear combination of $x_1,\dots,x_r$, we have 
\[
   \Pr[\langle c,x_1\rangle = \cdots = \langle c,x_{k'}\rangle = 0] = \Pr[\langle c,x_1\rangle = \cdots = \langle c,x_{r}\rangle = 0] \underset{\text{by \cref{item:dethashalmostuniform}}}{\ge} \tfrac{1}{p^r} -\delta \ge \tfrac{1}{p^{k'}} -\tfrac{1}{2p^{k-1}} \ge \tfrac{1}{2p^{k'}}
   \]
   as claimed.
 \end{proof}

 \subsection{Almost linear hashing with ``almost'' almost \texorpdfstring{$k$}{k}-wise independence}
 \label{subsec:hash}
Let $k\ge 2$ be a fixed constant integer.
We now construct our hash family $\caH \subset \{h\colon \{-N,\dots,N\} \to \F_p\}$.
In the following we assume the universe size $N$ is larger than some big constant (possibly dependent on $k$) in order to simplify some calculation. 

\begin{definition}[Hash family $\caH$]
Let $k\ge 2$ be a fixed constant integer.
Given universe size $N$ and a prime parameter $p \ge \log^{2k}N$,  
define a hash family $\caH \subset \{h\colon \{-N,\dots,N\} \to \F_p\}$ as follows: 

\begin{itemize}
    \item  
    Let 
    \[Q := 100\log N < p.\]
    Deterministically pick distinct primes $q_1,q_2,\dots,q_{m-1}\in [Q/2,Q]$ such that 
    \begin{equation}
    \label{eqn:largeproduct}
    q_1q_2\cdots q_{m-1} \in [N(\log  N)^{2k}, N(\log  N)^{3k}],
    \end{equation}
    which is possible since the product of all primes in the interval $[50\log  N, 100\log  N]$ is at least $N^2$ (by the Prime Number Theorem, e.g.~\cite{RosserS62}). In particular, 
\[ m = (1+o(1))\tfrac{\log  N}{\log  \log  N}.\]
    \item  Draw a random  $c\in  \F_p^{m}$ from \cref{lem:dethash} with parameters $p,m,k$. 
    Then, for $x\in \{-N,\dots,N\}$,
   define vector $w =  (x\bmod q_1 ,x\bmod q_2, \dots, x\bmod q_{m-1}, 1) \in \F_p^m$ (where each $x\bmod q_i$ is converted into $\{0,1,\dots,q_i-1\} \subset \F_p$), and let the hash value of $x$ be $h(x) := \langle c, w\rangle $.
\end{itemize}

By \cref{lem:dethash}, this hash family has size $|\caH| \le 4m^2 p^{3k}$ and can be computed in $\tilde O(m^3 p^{3k})$ time (the time for computing $q_1,\dots,q_{m-1}$ is $O(log N )=o(p)$ and is negligible), and each hash function $h\in \caH$ can be evaluated at any given $x$ in $O(m) = o(\log N)$ time.
\label{def:hashfamilyh}
\end{definition}

First, we show that the hash family $\caH$ is almost linear.

\begin{lemma}[Almost linearity, $\Delta_h$]
   \label{lem:almostlinearity}
   In \cref{def:hashfamilyh}, for every $h\in \caH$ (defined using $c\in \F_p^m$), we have
   $h(x+y)-h(x)-h(y) \in \Delta_h$ for all $x,y \in \{-N,\dots,N\}$, where $\Delta_h \subseteq \F_p$ 
   is defined as the iterative sumset,
   \[ \Delta_{h}:= \{0,-q_1c_1\} + \{0,-q_2c_2\} + \dots + \{0,-q_{m-1}c_{m-1}\} + c_m \]
with size $|\Delta_h|\le 2^{m-1} = N^{o(1)}$.
\end{lemma}
\begin{proof}
  Note that $q_i\le Q< p$ for all $i$.
  For any $x,y\in \Z$, \[h(x+y)-h(x)-h(y) = c_m + \sum_{i=1}^{m-1}c_i\cdot \big (((x+y)\bmod q_i) - (x\bmod q_i) - (y\bmod q_i)\big ).\]
  By $((x+y)\bmod q_i) - (x\bmod q_i) - (y\bmod q_i) \in \{0,-q_i\}$, the claim immediately follows.
\end{proof}

Using almost linearity, we now show that the hash values of Sidon 4-tuples under $\caH$ are highly correlated.
Recall that for a set $\Delta$, we define the iterative sumset $k\Delta = (k-1)\Delta + \Delta$.
\begin{lemma}[Hashing Sidon 4-tuples]
 Suppose integers $x,y,z,w\in [N]$ satisfy $x+y=z+w$. Then the hash family $\caH$ from \cref{def:hashfamilyh} (with $k\ge 3$) satisfies
   \[ \Pr_{h\in \caH}\big[\exists i\in \F_p: \{h(x),h(y),h(z),h(w)\} \subseteq  i + \Delta'_h\big ] \ge \frac{1}{2p^2},\]
   where $\Delta'_h:= 3(\{0\}\cup \Delta_h)\subseteq \F_p$, with size   $|\Delta'_h| < (|\Delta_h|+1)^3 \le 2^{3m} = N^{o(1)}$.
   \label{lem:hashquadlb}
\end{lemma}
\begin{proof}
  By the definition of $h$ and by \cref{item:dethashlb} of \cref{lem:dethash}, we have $h(x-y)=h(y-z)=0$ with at least $\frac{1}{2p^2}$ probability.
   Since $y-z=w-x$, in this case we also have $h(w-x)=0$. Then, by  \cref{lem:almostlinearity}, we have $h(y)-h(z) \in \Delta_h, h(x)-h(y)\in \Delta_h$, and $h(w)-h(x)\in \Delta_h$. So $h(y)\in h(z)+\Delta_h, h(x)\in h(z)+\Delta_h+\Delta_h$, and $h(w)\in h(z)+\Delta_h+\Delta_h+\Delta_h$. Then, for $i= h(z)$, we have  $\{h(x),h(y),h(z),h(w)\} \subseteq \{i\} \cup \{i+\Delta_h\}\cup \{i+2\Delta_h\} \cup \{i + 3\Delta_h\} = i + \Delta'_h$.
\end{proof}

To state the next property of $\caH$, we need the following definition (which also appeared in \cite{JinX23}):
\begin{definition}[$k$-term $\ell$-relation]
   \label{defn:lrelation}
   We say $k$ integers $a_1,\dots,a_k$ have a \emph{($k$-term) $\ell$-relation}, if there exist integers $\beta_1,\dots,\beta_k\in [-\ell,\ell]$ such that  $\sum_{j=1}^k\beta_j = \sum_{j=1}^k \beta_j a_j = 0$, and at least one of $\beta_1,\dots,\beta_k$ is non-zero.
\end{definition}
The following lemma states that the hash family $\caH$ is almost $k$-wise independent on integers \emph{without} $\ell$-relations for small $\ell$.\footnote{In contrast, on integers \emph{with} $\ell$-relations for small $\ell$, the almost $k$-wise independence must fail, as can be seen from \cref{lem:hashquadlb}.}

\begin{lemma}[``Almost'' almost $k$-wise independence]
   \label{lem:almostkwiseindep}
Let $\caH$ be the hash family from \cref{def:hashfamilyh}.
Let $k'\le k$, and integers $x_1,x_2,\dots,x_{k'} \in \{-N,\dots,N\}$ have no $2k^{2k+1}Q^{k'-1}$-relations. Then, 
\[ \Pr_{h\in \caH}\Big [h(x_1)=h(x_2)=\dots = h(x_{k'})\Big ] \le p^{-(k'-1)}\cdot 2(k+1)^2(2k+1)^{m-1} =p^{-(k'-1)} N^{o(1)} . \]
\end{lemma}

We actually deduce \cref{lem:almostkwiseindep} from the following lemma, which will also be useful in our applications.
\begin{lemma}
   \label{lem:deltahalmostkwiseindep}
Let $\caH$ be the hash family from \cref{def:hashfamilyh} (with fixed integer parameter $k$).
Let $k'\le k-1$, and integers $x_1,x_2,\dots,x_{k'} \in \{-N,\dots,N\}$ be such that $x_1,\dots,x_{k'},0$ have no  $(k'+1)$-term 
$k^{2k+1}Q^{k'-1}$-relations.
Then,
\[ \Pr_{h\in \caH}\Big [h(x_i)\in \bigcup_{0\le k_1,k_2\le k}(k_1\Delta_h-k_2\Delta_h) \text{ for all }1\le i \le k'\Big ] \le p^{-k'}\cdot 2(k+1)^2(2k+1)^{m-1}=p^{-k'} N^{o(1)}.  \]
\end{lemma}
\begin{proof}[Proof of \cref{lem:deltahalmostkwiseindep}]
Suppose $h\in \caH$ is specified by $c\in \F_p^m$.
By the definition of $\Delta_h$ in \cref{lem:almostlinearity}, we know 
\[ k_1\Delta_h- k_2\Delta_h = \Big \{ (k_1-k_2)c_m + \sum_{i=1}^{m-1} x_i q_i c_i\, :\, x_i \in \{-k_1,\dots,k_2\}\text{ for all }i \Big \} \subseteq \F_p. \]

   For $1\le j\le k'$, let $w_j:= (x_j \bmod q_1,\dots,x_j \bmod q_{m-1},1) \in \Z^m$, and let $\tilde w_j\in \F_p^m$ be $w_j$ with each coordinate viewed as an element in $\F_p$. 
   Recall from \cref{def:hashfamilyh} that $h(x_j) =  \langle c, \tilde w_j\rangle$.
Then, $h(x_j)\in k_1\Delta_h- k_2\Delta_h$ if and only if 
\begin{equation}
\label{eqn:temphashtozero}
  (1+k_2-k_1)c_m + \sum_{i=1}^{m-1} (\tilde w_j - x_iq_i)c_i =0 \in \F_p. 
\end{equation}
for some $0\le k_1,k_2\le k$ and $(x_1,\dots,x_{m-1})\in \{-k_1,\dots,k_2\}^{m-1}$.
Now we fix $k_1,k_2$ and $(x_1,\dots,x_{m-1})$, and later we will apply a union bound over all such possibilities (there are at most $(k+1)^2\cdot (2k+1)^{m-1}$ of them).

For each $1\le j\le k'$, define $u_j:= (w_j-x_1q_1, w_j-x_2q_2,\dots,w_j-x_{m-1}q_{m-1}, 1+k_2-k_1) \in \Z^m$, and let 
$\tilde u_j\in \F_p^m$ be $u_j$ with each coordinate viewed as an element in $\F_p$. Then \cref{eqn:temphashtozero} is equivalent to $\langle c, \tilde u_j\rangle = 0\in \F_p$.

   In the following we will prove that the vectors $\tilde u_1,\dots,\tilde u_{k'}$ are linearly independent. Once this is proved,  \cref{item:dethashalmostuniform} of \cref{lem:dethash}  would imply that the distribution of $(\langle c, \tilde u_1\rangle,\dots,\langle c, \tilde u_m\rangle)$ is $1/(2p^{k-1})$-almost uniform, and in particular they are all zeros with probability at most $p^{-k'} + 1/(2p^{k-1}) \le 2p^{-k'}$. So, by a union bound over all $\le (k+1)^2\cdot (2k+1)^{m-1}$ many possibilities of $k_1,k_2,(x_1,\dots,x_m)$, we get the claimed probability upper bound in \cref{lem:deltahalmostkwiseindep}.

To prove $\tilde u_1,\dots,\tilde u_{k'}$ are linearly independent, we suppose for contradiction that  they are linearly dependent.
   Define matrices $U:=\begin{pmatrix}
      | & | & & | \\
      u_1 & u_2 & \cdots & u_{k'} \\
      | & | & & |
      \end{pmatrix}  \in \Z^{m\times k'}$ and
      $\tilde U:=\begin{pmatrix}
      | & | & & | \\
      \tilde u_1 & \tilde u_2 & \cdots & \tilde u_{k'} \\
      | & | & & |
      \end{pmatrix}  \in \F_p^{m\times k'}.$
    Let $r<k'$ be the rank of $\tilde U$.
  We pick a subset $R$ of $r$ rows and a subset $C$ of $r$ columns such that the submatrix $\tilde U[R,C]$ has rank $r$. Since $r<k'$, we can pick another column $c\in \{1,2,\dots,k'\}\setminus C$.
  
  Now, we fix an arbitrary row $i\in \{1,2,\dots,m\}$. Observe that the $(r+1)\times (r+1)$ matrix
  $\tilde V:= \begin{pmatrix}
  \tilde U[R,C] & \tilde U[R,c]\\
  \tilde U[i,C] & \tilde U[i,c]
  \end{pmatrix}$
  must have rank $r$. This is because, if 
  $i\notin R$, then $\tilde V$ is a submatrix of $\tilde U$, and cannot have rank larger than $r$. Otherwise, $i\in R$, so the last row of $\tilde V$ repeats one of the previous $r$ rows. 
  Hence, $\det \tilde V = 0 \in \F_p$, which means the integer matrix 
  $ V:= \begin{pmatrix}
   U[R,C] &  U[R,c]\\
   U[i,C] &  U[i,c]
  \end{pmatrix}$
  has determinant congruent to $0$ modulo $p$. 
  Since each entry in $V$ is an integer bounded by $(k+1)Q$ (by the definition of $u_j$), $\det V$ is an integer bounded by $(r+1)! ((k+1)Q)^{r+1}\le (k+1)^{2k}Q^k = O((\log N)^k)$, which is smaller than 
  $p \ge \log^{2k}{N}$ assumed in \cref{def:hashfamilyh}. Hence, from $\det V\equiv 0\pmod{p}$ we have $\det V = 0$.
   By expanding $\det V$ along the last row, we have (for notational convenience, assume $C \cup\{c\} = \{1,2,\dots,r+1\}$ without loss of generality)
  \begin{equation}
     \label{eqn:relationj}
   \sum_{j\in C\cup \{c\}} \underbrace{(-1)^j\det U[R,C\cup \{c\} \setminus \{j\}]}_{=:\beta_j \in \Z} \cdot U[i,j]  = 0.
  \end{equation}
  Here, 
   \[|\beta_j| \le r!((k+1)Q)^r\le k^{2k}Q^{k'-1},\] and 
  \[\beta_c = (-1)^c \det U[R,C]\neq 0\] since $\tilde U[R,C]$ is full-rank and so is $U[R,C]$.
  We let $\beta_j =0$ for $j\notin C\cup \{c\}$.
  Since row $i$ is arbitrary, \cref{eqn:relationj} holds for all $1\le i\le m$.  
  For $i\le m-1$, recall $U[i,j] = w_j - x_iq_i  \equiv x_j \pmod{q_i}$, so \cref{eqn:relationj} implies 
  $ \sum_{1\le j \le k'} \beta_j x_j\equiv 0 \pmod{q_i}.$
  Since \[ \Big \lvert \sum_{1\le j\le k'}\beta_jx_j \Big \rvert \le k'\max_j |\beta_j| \cdot N \le k' k^{2k}Q^{k'-1}  N = O(N\log^{k'-1} N),\] 
 which by \cref{eqn:largeproduct} gives the upper bound $\lvert \sum_{1\le j\le k'}\beta_jx_j  \rvert{<} q_1q_2\cdots q_{m-1}$. Then,
 by Chinese Remainder Theorem with moduli $q_1,\dots,q_{m-1}$,  we have \[\sum_{1\le j \le k'} \beta_j x_j=0.\]
Set $\beta_{k'+1}=-\sum_{1\le j\le k'}\beta_j$, with $|\beta_{k'+1}|\le \sum_{1\le j\le k'}|\beta_j|\le  k^{2k+1}Q^{k'-1}$. Then, we have found a $(k'+1)$-term $k^{2k+1}Q^{k'-1}$-relation for $x_1,\dots,x_{k'},0$ with coefficients $\beta_1,\dots,\beta_{k'+1}$,  a contradiction to the assumption that such relations do not exist.
\end{proof}

\begin{proof}[Proof of \cref{lem:almostkwiseindep} using \cref{lem:deltahalmostkwiseindep}]
By almost linearity \cref{lem:almostlinearity} of $\caH$, we know $h(x_i)=h(x_j)$ implies $h(x_i-x_j)\in  -\Delta_h$. So $h(x_1)=\dots = h(x_{k'})$ implies $h(x_1-x_2),h(x_2-x_3),\dots,h(x_{k'-1}-x_{k'}) \in -\Delta_h$. The probability that this happens can be bounded using \cref{lem:deltahalmostkwiseindep} by at most 
$p^{-(k'-1)}\cdot 2(k+1)^2(2k+1)^{m-1}$, provided that we show $x_1-x_2,\dots,x_{k'-1}-x_{k'}, 0$ have no $k'$-term $k^{2k+1}Q^{k'-1}$-relations. To see this,  note that if there is such a relation, then it immediately implies a $k'$-term $2k^{2k+1}Q^{k'-1}$-relation on $x_1,x_2,\dots,x_{k'}$, contradicting our assumption.
\end{proof}

\section{Tool IV: Deterministic BSG Theorem in Subquadratic Time}
\label{sec:bsg}
In this section, we prove the following deterministic algorithmic version of the BSG theorem running in subquadratic time when $K$ is not too small.
    \begin{theorem}[Full version of \cref{thm:bsgintro}]
        \label{thm:derandbsg}
 Given set $A\subseteq [N], r\ge 2$ and $K\ge 1$ such that $\sfE(A)\ge |A|^3/K$, 
    we can deterministically find subsets $A',B' \subseteq A$ in $O((|A|^2/ K^{2r-4} + K^{13}|A|)N^{o(1)})$ time, such that 
    \begin{itemize}
        \item $|A'| \ge |A|/(64K)$,
        \item $|B'| \ge |A|/\big (K^{r+3}\cdot 2^{\frac{C\log N}{\sqrt{\log \log N}}}\big )$, and
            \item $|A'+B'|\le K^5|A|\cdot 2^{\frac{C\log N}{\sqrt{\log \log N}}}$,
    \end{itemize}
    for some absolute constant $C\ge 1$.
    \end{theorem}
    Note that Ruzsa triangle inequality implies $|A'+A'| \le \frac{|A'+B'|^2}{|B'|} \le K^{13+r}|A| \cdot 2^{\frac{3C\log N}{\sqrt{\log \log N}}}$.

\subsection{Finding a small high-energy subset}
\label{sec:findsmallsubset}
The goal of this subsection is to prove the following theorem about finding a subset of appropriate size and high additive energy. 
This theorem will be used later for proving our derandomized subquadratic-time BSG theorem.

\begin{theorem}[Finding a small high-energy subset]
 Given set $A \subseteq [N]$, size parameter $1\le R \le |A|$, and $1\le   K\le R^{1/3}$ such that $\sfE(A)\ge |A|^3/K$ holds, in $O(  K^{13} |A|N^{o(1)})$ time we can deterministically find a subset $A'\subseteq A$ 
 such that
  \[ \sfE(A') \ge  \frac{|A'|^3}{  K}\cdot 2^{-\frac{C\log  N}{\sqrt{\log \log  N}}},\]
    and 
   \[\frac{R}{  K^2} \le |A'|\le R\cdot 2^{\frac{C\log  N}{\sqrt{\log \log  N}}},\] 
   where $C\ge 1$ is some absolute constant.
\label{thm:find-subset-highenergy}
\end{theorem}
\cref{thm:find-subset-highenergy} follows from iterating the following lemma.
\begin{lemma}\label{lemma:find-subset-highenergy-lemma}
   Assume $N\ge N_0$ for some absolute constant $N_0>1$.
Given set $A \subseteq [N]$ and $1\le K\le |A|^{1/3}\cdot  2^{-\frac{3\log  N}{\sqrt{\log  \log  N}}}$ such that $\sfE(A)\ge |A|^3/K$ holds, in $O(K^{13} |A|N^{o(1)})$ time we can deterministically find a subset $A'\subseteq A$  such that
  \[ \sfE(A') \ge  \frac{|A'|^3}{  K}\cdot 2^{-\frac{700\log  N}{\log \log  N}}\]
  and 
   \[\frac{|A|}{  K^2}\cdot 2^{-\frac{3\log  N}{\sqrt{\log  \log  N}}} \le |A'| \le  |A|\cdot 2^{-\frac{\log  N}{\sqrt{\log  \log  N}}}.\] 
\end{lemma}
\begin{proof}[Proof of \cref{thm:find-subset-highenergy} using \cref{lemma:find-subset-highenergy-lemma}]
If $N=O(1)$ then $|A|=O(1)$, and we can simply return $A$ which satisfies the requirements for large enough $C\ge 1$. Hence, in the following we assume $N$ is large enough in order to simplify calculation.
   
   Let $A_0 = A, K_0=K$. In iteration $i$ ($i\ge 0$): 
   \begin{itemize}
      \item If $|A_i|\le R\cdot 2^{\frac{2109\log N}{\sqrt{\log \log N}}}$, then return $A_i$ as our answer $A'$.
         \item Otherwise, apply \cref{lemma:find-subset-highenergy-lemma} to $A_i$ and $K_i$, and let $A_{i+1}\subseteq A_i$ be the returned subset, and let $K_{i+1} = K_i\cdot 2^{\frac{700\log  N}{\log  \log  N}}$. 
   \end{itemize}
   We now verify the conditions required by \cref{lemma:find-subset-highenergy-lemma} are satisfied.
   By the energy bound in \cref{lemma:find-subset-highenergy-lemma}, we have $\sfE(A_{i+1})\ge |A_{i+1}|^3/K_{i+1}$.
    By the size bound in \cref{lemma:find-subset-highenergy-lemma}, $|A_{i+1}|/|A_i| \le 2^{-\frac{\log  N}{\sqrt{\log  \log  N}}}$, so the algorithm terminates in at most $\sqrt{\log \log  N}$ iterations before the size of $A_i$ drops below $1$.
     Hence, \[K_i= K  2^{\frac{700\log  N}{\log  \log  N}\cdot i} \le  K 2^{\frac{700\log  N}{\sqrt{ \log  \log  N}}} \le R^{1/3} 2^{\frac{700\log  N}{\sqrt{ \log  \log  N}}} < \big (|A_i|2^{\frac{-2109\log N}{\sqrt{\log \log N}}}\big )^{1/3} 2^{\frac{700\log  N}{\sqrt{\log   \log  N}}} = |A_i|^{1/3}2^{\frac{-3\log  N}{\sqrt{\log   \log  N}}},\] satisfying the conditions required by \cref{lemma:find-subset-highenergy-lemma}.

Now we show the returned subset $A_i\subseteq A$ satisfies the claimed properties. If the returned subset is $A_0=A$ then clearly all requirements are satisfied, so we assume $i\ge 1$.  
Then, $|A_{i-1}|>R\cdot 2^{\frac{2109\log N}{\sqrt{\log \log N}}}\ge |A_i|$, so by the size bound of \cref{lemma:find-subset-highenergy-lemma} we have $|A_{i}|\ge \frac{|A_{i-1}|}{K_{i-1}^2}\cdot 2^{-\frac{3\log  N}{\sqrt{\log  \log  N}}} > \frac{R\cdot 2^{\frac{2109\log N}{\sqrt{\log \log N}}}}{(K \cdot 2^{\frac{700\log  N}{\sqrt{\log  \log  N}}})^2} \cdot 2^{-\frac{3\log  N}{\sqrt{\log  \log  N}}} \ge  \frac{R}{K^2}$ as desired. The energy bound of \cref{lemma:find-subset-highenergy-lemma} gives $\sfE(A_i) \ge \frac{|A_i|^3}{K_{i-1}}\cdot 2^{-\frac{700\log  N}{\log  \log  N}} \ge \frac{|A_i|^3}{K\cdot 2^{\frac{700\log N}{\sqrt{\log \log N}}}}\cdot 2^{-\frac{700\log  N}{\log  \log  N}} \ge \frac{|A_i|^3}{K} \cdot 2^{-\frac{701\log  N}{\sqrt{\log  \log  N}}}$ as desired.

The total time complexity over the at most $\sqrt{\log \log N}$ iterations is still $O(K^{13}|A|N^{o(1)})$. This concludes the proof.
\end{proof}

In the rest of this subsection, we prove \cref{lemma:find-subset-highenergy-lemma}.  We use the hash family $\caH \subset \{h\colon \{-N,\dots,N\} \to \F_p\}$ from \cref{def:hashfamilyh} with $k:=4$, and prime parameter $p \ge  2^{\frac{3\log N}{\sqrt{\log \log N}}}$ to be determined later. 
Recall from \cref{def:hashfamilyh} that $m =\frac{(1+o(1))\log N}{\log \log N}$, and $Q = 100\log N$. We assume $N$ is large enough in order to simplify some calculation. 

For $h\in \caH$ and $i\in \F_p$, let $h^{-1}(i)= \{x\in \{-N,\dots,N\}: h(x)=i\}$, and define 
\begin{equation}
   \label{eqn:defngi}
G_i := \bigcup_{d\in \Delta'_h}h^{-1}(i+d),
\end{equation}
where $\Delta'_h$ is defined in \cref{lem:hashquadlb} with size $|\Delta_h'|\le 2^{3m}$. Note that 
\begin{equation}
   \label{eqn:sumgi}
  \sum_{i\in \F_p} |A\cap G_i| \le \sum_{i\in \F_p} \sum_{d\in \Delta_h'} |A\cap h^{-1}(i+d)| =\sum_{d\in \Delta_h'}\sum_{i\in \F_p}  |A\cap h^{-1}(i+d)| = |\Delta_h'||A| \le 2^{3m}|A|.
\end{equation}

\begin{lemma}
   \label{lem:expboundshash}
Let set $A\subseteq [N]$. 
   \begin{enumerate}
      \item 
      \label{item:expectedenergy}
Let $\sfE(A)$ be the additive energy of $A$. Then,
\begin{equation}
      \label{eqn:item:expectedenergy}
  \Ex_{h\in \caH}\Big [\sum_{i\in \F_p}\sfE(A\cap G_i)\Big ] \ge  \frac{\sfE(A)}{2p^2}.
\end{equation}
   \item
      \label{item:expectedsizemoment}
    For $2\le k'\le k$, 
    \begin{equation}
      \label{eqn:item:expectedsizemoment}
   \Ex_{h\in \caH}\Big [\sum_{i\in \F_p}|A\cap G_i|^{k'}\Big ] \le
   2^{10mk^3}\cdot \left ( \frac{|A|^{k'}}{p^{k'-1}} + Q^{(k'-1)^2} |A|^{k'-1}\right ).
    \end{equation}
   \end{enumerate}
\end{lemma}
\begin{proof}
   \begin{itemize}
      \item Proof of \cref{item:expectedenergy}:
 For any $(x,y,z,w)\in A^4$ such that $x+y=z+w$,    by \cref{lem:hashquadlb} and by definition of $G_i$,  we have  $\Pr_{h\in \caH} \big  [ \exists i\in \F_p: \{x,y,z,w\} \subseteq G_i\big ]\ge \frac{1}{2p^2}$. Then, by linearity of expectation, 
 \begin{align*}
  \Ex_{h\in \caH}\Big [\sum_{i\in \F_p}\sfE(A\cap G_i)\Big ] &= \sum_{(x,y,z,w)\in A^4: x+y=z+w} \sum_{i\in \F_p} \Pr_{h\in \caH}[\{x,y,z,w\} \subseteq G_i]\\
  &\ge \sum_{(x,y,z,w)\in A^4: x+y=z+w} \Pr_{h\in \caH}[\exists i\in \F_p: \{x,y,z,w\} \subseteq G_i]\\
& \ge \frac{\sfE(A)}{2p^2}.
 \end{align*}

      \item Proof of \cref{item:expectedsizemoment}:
         We first bound $\Ex_{h\in \caH} \big [\sum_{i\in \F_p}|A\cap h^{-1}(i)|^{k'}\big ]$. By linearity of expectation, 
         \begin{equation}
            \label{eqn:expectedmoment}
         \Ex_{h\in \caH} \Big [\sum_{i\in \F_p}|A\cap h^{-1}(i)|^{k'}\Big ] = \sum_{(x_1,\dots,x_{k'})\in A^{k'}}\Pr_{h\in \caH}[h(x_1)=\cdots = h(x_{k'})].
         \end{equation}
         For $k'$-tuples $(x_1,\dots,x_{k'})\in A^{k'}$ that have no $k'$-term
 $k^{2k+1}Q^{k'-1}$-relations,
          we use \cref{lem:almostkwiseindep} to bound $\Pr_{h\in \caH}[h(x_1)=\cdots = h(x_{k'})] \le
          p^{-(k'-1)}\cdot 2(k+1)^2(2k+1)^{m-1}$.
         To bound the number of the remaining $k'$-tuples $(x_1,\dots,x_{k'})\in A^{k'}$, we can first pick the coefficients $\beta_1,\dots,\beta_{k'}$ (which sum to zero and are not all zero) of the $k'$-term relation involving $x_1,\dots,x_{k'}$ in $(2k^{2k+1}Q^{k'-1}+1)^{k'-1} - 1$ ways, and assuming $\beta_j\neq 0$, we then pick $x_1,\dots,x_{j-1},x_{j+1},\dots,x_{k'}\in A$ in $|A|^{k'-1}$ ways, which uniquely determine $x_j$ via the relation $\sum_{j=1}^{k'}\beta_j x_j =0$. The number of ways is 
         $\big ((2k^{2k+1}Q^{k'-1}+1)^{k'-1} - 1\big )\cdot |A|^{k'-1} \le k^{3k^2}Q^{(k'-1)^2} |A|^{k'-1}$.
         Hence, \cref{eqn:expectedmoment} can be upper bounded by 
         \[   \frac{|A|^{k'}}{p^{k'-1}} 2(k+1)^2(2k+1)^{m-1} + k^{3k^2}Q^{(k'-1)^2} |A|^{k'-1}.\]

         Then,
        \begin{align*}
        \Ex_{h\in \caH} \left [ \sum_{i\in \F_p}|A\cap G_i|^{k'}\right ] &= \Ex_h \left [ \sum_{i\in \F_p}\Big (\sum_{d\in \Delta_h'}|A\cap h^{-1}(i+d)| \Big)^{k'}\right ] \\
        & \le  \Ex_{h\in \caH} \left [ \sum_{i\in \F_p}|\Delta_h'|^{k'-1}\sum_{d\in \Delta_h'}|A\cap h^{-1}(i+d)|^{k'}\right ] \tag{H\"older}\\
        & = \Ex_{h\in \caH} \left [ |\Delta_h'|^{k'}\sum_{i\in \F_p}|A\cap h^{-1}(i)|^{k'}\right ] \\
        & \le 2^{3mk'} \cdot \left (
            \frac{|A|^{k'}}{p^{k'-1}} 2(k+1)^2(2k+1)^{m-1} + k^{3k^2}Q^{(k'-1)^2} |A|^{k'-1}
         \right )\\
         & \le  2^{10mk^3}\cdot \left ( \frac{|A|^{k'}}{p^{k'-1}} + Q^{(k'-1)^2} |A|^{k'-1}\right ).
        \end{align*} 
   \end{itemize}
\end{proof}

Using \cref{lem:expboundshash}, we now show that there exists a function $h\in \caH$ such that one of the sets $A\cap G_i$ has appropriate size and high additive energy.

\begin{corollary}
   \label{cor:existsgoodset}
   Suppose set $A \subseteq [N]$ has size 
  \begin{equation}
     \label{eqn:existsgoodsetrequirement}
   |A|\ge Q^9p^3 
  \end{equation} 
    and additive energy $\sfE(A) \ge |A|^3/K$. Then, there exist a hash function $h\in \caH$ and $i\in \F_p$ such that 
  \begin{enumerate}
   \item $\sfE(A\cap G_i) \ge \frac{|A\cap G_i|^3}{K} \cdot \frac{1}{10\cdot 2^{10mk^3}}$.
      \item  $|A\cap G_i| \ge \frac{|A|}{p} \cdot \frac{1}{10\cdot 2^{10mk^3} K }$.
         \item $|A\cap G_i| \le \frac{|A|}{p} \cdot 10\cdot 2^{10mk^3} K$.
  \end{enumerate}
\end{corollary}
\begin{proof}
   Without loss of generality, assume $\sfE(A) = |A|^3/K$.
   By scaling and subtracting the inequalities in \cref{eqn:item:expectedsizemoment} (for $k'=2,3,4$) from  \cref{eqn:item:expectedenergy},  we have 
  \begin{align*}
   & \Ex_{h\in \caH} \sum_{i\in \F_p}\Big [\sfE(A\cap G_i)\cdot K - |A\cap G_i|^2 \cdot \tfrac{|A|}{ {2^{10mk^3}}\cdot 10p} - |A\cap G_i|^3 \cdot \tfrac{1}{ {2^{10mk^3}}\cdot 10} - |A\cap G_i|^4 \cdot \tfrac{p}{ {2^{10mk^3}}\cdot 10|A|} \Big ] \\ 
  & \ge  \tfrac{\sfE(A)}{2p^2}\cdot K  -\big (\tfrac{|A|^2}{p} + Q|A|\big ) \cdot \tfrac{|A|}{10p} - \big (\tfrac{|A|^3}{p^2} + Q^4|A|^2\big )\cdot \tfrac{1}{10} - \big (\tfrac{|A|^4}{p^3} + Q^{9}|A|^3\big )\cdot \tfrac{p}{10|A|}\\
  & = \frac{|A|^3}{p^2} \Big ( \frac{1}{2} -  \frac{1}{10} - \frac{1}{10}- \frac{1}{10}\Big ) -\frac{Q|A|^2}{10p} - \frac{Q^4|A|^2}{10} - \frac{Q^9|A|^2 p}{10}\\
  & \ge  \frac{|A|^3}{5p^2}  - \frac{Q^9|A|^2 p}{5}\\
   & >0.  \tag{by \cref{eqn:existsgoodsetrequirement}}
  \end{align*}
  Hence there exist $h\in \caH$ and $i\in \F_p$ such that 
   \[ \sfE(A\cap G_i)\cdot K >  |A\cap G_i|^2 \cdot \tfrac{|A|}{ {2^{10mk^3}}\cdot 10p} + |A\cap G_i|^3 \cdot \tfrac{1}{ {2^{10mk^3}}\cdot 10} + |A\cap G_i|^4 \cdot \tfrac{p}{ {2^{10mk^3}}\cdot 10|A|} .  \]
  The three terms on the RHS are all non-negative, so each of them is smaller than the LHS, giving us three inequalities. These three inequalities (combined with the bound $\sfE(A\cap G_i) \le |A\cap G_i|^3$) immediately give us the three claimed properties of $A\cap G_i$.
\end{proof}

Now we are ready to prove \cref{lemma:find-subset-highenergy-lemma}. The idea is to enumerate all $h\in \caH$ and all $i\in \F_p$, and find a good subset of $A$ as guaranteed by \cref{cor:existsgoodset}, using the deterministic algorithm for approximating additive energy from \cref{cor:approxenergy}.

\begin{proof}[Proof of \cref{lemma:find-subset-highenergy-lemma}]
We use the hash family $\caH \subset \{h\colon \{-N,\dots,N\} \to \F_p\}$ in \cref{def:hashfamilyh}. Recall we set $k=4$.
Now we pick an arbitrary prime (recall $m = \frac{(1+o(1))\log N}{\log \log N}$) \[p\in [  K\cdot 2^{2m \sqrt{\log  m}},2  K\cdot 2^{2m \sqrt{\log  m}}],\]
which satisfies the $p \ge \log^{2k} N$ requirement from \cref{def:hashfamilyh}, and 
$p\le O(K 2^{(2+o(1))\frac{\log  N }{\sqrt{\log  \log  N}}}) \le  O(  K N^{o(1)})$.

From the assumption 
$  K\le |A|^{1/3} 2^{-\frac{3\log  N}{\sqrt{\log  \log  N}}}$ in the lemma statement, we can also verify that the condition \cref{eqn:existsgoodsetrequirement} in \cref{cor:existsgoodset} of satisfied.

Now we describe the algorithm. We enumerate all the $O(m^2p^{12})$ hash functions in the family $\caH$ from \cref{def:hashfamilyh}. For each $h\in \caH$, compute the buckets $A \cap h^{-1}(i)$ for all $i\in \F_p$, compute sets $\Delta_h$ and $\Delta_h'$, and then compute the sets $A \cap G_i$ for all $i\in \F_p$ (\cref{eqn:defngi}). 
   Then, for each $i\in \F_p$ such that $A\cap G_i$ satisfies the size bounds of \cref{cor:existsgoodset}, that is,
   \begin{equation}
   |A\cap G_i| \in \big[\tfrac{|A|}{p}\cdot \tfrac{1}{10\cdot 2^{10mk^3} \cdot   K}, \tfrac{|A|}{p} \cdot 10\cdot 2^{10mk^3}\cdot   K\big ]\subset \big[\tfrac{|A|}{  K^2}\cdot 2^{-\frac{3\log  N}{\sqrt{\log  \log  N}}}, |A|\cdot 2^{-\frac{\log  N}{\sqrt{\log  \log  N}}}\big ],
   \label{eqn:sizeboundagi}
   \end{equation}
   we use \cref{cor:approxenergy} to deterministically approximate $\sfE(A\cap G_i)$ up to additive error $\frac{1}{3}\cdot \frac{|A\cap G_i|^3}{  K} \cdot \frac{1}{10\cdot 2^{10mk^3}}$ in $O(  K 2^{10mk^3} |A\cap G_i| N^{o(1)})$ time.
    If the approximate value of $\sfE(A\cap G_i)$ is at least $\frac{2}{3}\cdot \frac{|A\cap G_i|^3}{  K}\cdot \frac{1}{10\cdot 2^{10mk^3}}$, then we can terminate and return $A\cap G_i$ as the answer, which has additive energy
  \[ \sfE(A\cap G_i) \ge \tfrac{1}{3}\cdot \tfrac{|A\cap G_i|^3}{  K}\cdot \tfrac{1}{10\cdot 2^{10mk^3}} \ge \tfrac{|A\cap G_i|^3}{  K}\cdot 2^{-\frac{700\log  N}{\log \log  N}},\]
   as desired, and also satisfies the claimed size bound due to \cref{eqn:sizeboundagi}. 
   On the other hand, by \cref{cor:existsgoodset}, we know there must exist some $h\in \caH$ and $i\in \F_p$ that make our algorithm terminate, so our algorithm is correct.  Observe that the bottleneck of our algorithm is invoking \cref{cor:approxenergy} for all $h \in \caH$ and all $i\in \F_p$, which has total time complexity \[|\caH|\cdot \sum_{i\in \F_i}O(  K 2^{10mk^3} |A\cap G_i| N^{o(1)}) \underset{\text{by \cref{eqn:sumgi}}}{\le} O(m^2 p^{12})\cdot O( K 2^{10mk^3} 2^{3m} |A| N^{o(1)}) =O(  K^{13}|A| N^{o(1)}) \]
   as claimed.
\end{proof}

\subsection{Proof of \texorpdfstring{\cref{thm:derandbsg}}{Theorem~\ref*{thm:derandbsg}}}

We need the following lemma which is implicit in the work of Schoen \cite{schoen} on the BSG theorem.
\begin{lemma}[{\cite[Proof of Lemma 2.1]{schoen}}]
    Suppose set $ A\subset \Z$, $\sfE(A) =\kappa |A|^3$, and $c>0$.  Let $Q = \{x: (1_A\conv 1_{-A})[x] \ge \frac{\kappa |A|}{4}\}$.
    Then there exists $s\in Q$ such that $X = A\cap (A+s)$ satisfies $|X| \ge \frac{\kappa |A|}{3}$ and $T = \{(a,b)\in X\times X: (1_A\conv 1_{-A})[a-b] < c\kappa |A|\}$ has size $ |T| \le 16c |X|^2$.
    \label{lem:schoenoriginal}
    \end{lemma}
    The following corollary is an algorithmic version of \cref{lem:schoenoriginal} via deterministic approximate 3SUM counting.
    \begin{corollary}
    \label{cor:schoenalgo}
    Given set $A\subseteq [N]$ and parameters $K\ge 1, c\in (0,1)$ such that $\sfE(A)\ge |A|^3/K$, we can deterministically find a subset $X\subseteq A$ in $O(|A|^2K^2c^{-2} N^{o(1)})$ time such that  $|X|\ge \frac{|A|}{3K}$ and
$|\{(a,b)\in X\times X: (1_A\conv 1_{-A})[a-b] < \frac{c|A|}{3K}\}|\le 18c|X|^2$.
    \end{corollary}
    \begin{proof}
        Let $\sfE(A) = \kappa |A|^3$ and $\kappa \ge 1/K$.
        Use \cref{thm:introdeterministicapprox3sum} to compute a vector $p$, such that $\|p-1_A\conv 1_{-A}\|_\infty\le \frac{c}{3K}\cdot |A|$, in $O(Kc^{-1}|A| N^{o(1)})$ time. Let $\tilde P= \{x: p[x]\ge \frac{2c}{3K}\cdot |A|\}$, which  satisfies
    \[\underbrace{\{x: 1_{A}\conv 1_{-A}[x]\ge c\kappa |A|\}}_{=:P} \subseteq\{x: 1_{A}\conv 1_{-A}[x]\ge \tfrac{c|A|}{K}\} \subseteq  \tilde P \subseteq \{x: 1_{A}\conv 1_{-A}[x]\ge \tfrac{c|A|}{3K}\}.\] 
and has size $|\tilde P| \le |A|^2/\frac{c|A|}{3K} = O(Kc^{-1}|A|)$.

    Similarly,  we can deterministically compute a set $\tilde Q$ of size $|\tilde Q| = O(K|A|\log^2 N)$ in $O(K|A|N^{o(1)})$ time such that
    \[\underbrace{\{x: 1_{A}\conv 1_{-A}[x]\ge \tfrac{\kappa |A|}{4}\}}_{=:Q}  \subseteq  \tilde Q \subseteq \{x: 1_{A}\conv 1_{-A}[x]\ge \tfrac{|A|}{12K}\}.\] 
 We enumerate every $s\in \tilde Q$, and let $X = A\cap (A+s)$.
 Use \cref{lem:approx4sumcount} (applied to $(A,B,C,D):= (X,\{0\},X,\tilde P),\eps := \min\{1,\frac{c|X|}{2|\tilde P|}\}$) to compute $\tilde m \in |\{(a,b)\in X\times X: a-b \in \tilde P\}| \pm c |X|^2$ in time $O((\eps^{-1} |X| + |\tilde P|)N^{o(1)}) = O((|X|+c^{-1}|\tilde P|)N^{o(1)}) \le O(Kc^{-2}|A| N^{o(1)})$, and we return $X$ as the answer if $|X| \ge \frac{|A|}{3K}$ and $|X|^2 - \tilde m \le 17c|X|^2$. The total time complexity is $|\tilde Q| \cdot O(Kc^{-2}|A| N^{o(1)}) = O(K^2c^{-2}|A|^2 N^{o(1)})$ as claimed. It remains to show the correctness of the algorithm.
 
First note that by \cref{lem:schoenoriginal} there exists $s\in Q\subseteq \tilde Q$ such that $X=A\cap (A+s)$ satisfies $|X|\ge \frac{\kappa |A|}{3} \ge \frac{|A|}{3K}$ and $T = \{(a,b)\in X\times X: a-b \notin P\}$ has size $|T|\le 16c|X|^2$, which means $|X|^2 - \tilde m \le |\{(a,b)\in X\times X: a-b \notin \tilde P\}| +c|X|^2 \le|T| + c|X|^2  \le 17c|X|^2$. Hence, our algorithm will always return some subset. On the other hand, when our algorithm returns $X$, we must have $|X|\ge \frac{|A|}{3K}$, and
$|\{(a,b)\in X\times X: (1_A\conv 1_{-A})[a-b] < \frac{c|A|}{3K}\}|\le |\{(a,b)\in X\times X: a-b\notin \tilde P\}| < (|X|^2 - \tilde m) + c|X|^2 \le 18c|X|^2$ as claimed.
    \end{proof}

    Now we are ready to prove \cref{thm:derandbsg}.
    \begin{proof}[Proof of \cref{thm:derandbsg}]
        Use \cref{thm:introdeterministicapprox3sum} to deterministically compute a set $S$ such that
        \[ \{s:(1_A\conv 1_A)[s] \ge \tfrac{|A|}{2K}\} \subseteq S \subseteq \{s:(1_A\conv 1_A)[s] \ge \tfrac{|A|}{4K}\}\]
        in $O(K|A|N^{o(1)})$ time, where $|S|\le |A|^2/\frac{|A|}{4K} = 4K|A|$.
        We have
        \begin{equation}
            \sum_{s\in S}((1_A\conv 1_A)[s])^2 = \sfE(A) - \sum_{s\notin S}((1_A\conv 1_A)[s])^2  
             \ge \tfrac{|A|^3}{K}- \sum_{s\notin S}\tfrac{|A|}{2K}(1_A\conv 1_A)[s]  \ge \tfrac{|A|^3}{2K} \label{eqn:sumxaconvasquared}
        \end{equation}
        and hence
         \begin{equation}
             \label{eqn:sumxaconva}
         \sum_{s\in S}(1_A\conv 1_A)[s] \ge \sum_{s\in S}\tfrac{((1_A\conv 1_A)[s])^2}{|A|} \ge \tfrac{|A|^2}{2K}.
         \end{equation}
         
         Use \cref{thm:introdeterministicapprox3sum} to compute a subset $B_1 \subseteq A$ such that 
       \[ \{b\in A:   (1_S\conv 1_{-A})[b] \ge \tfrac{|A|}{16K}\}  \subseteq B_1 \subseteq \{b\in A:   (1_S\conv 1_{-A})[b] \ge \tfrac{|A|}{32K}\}\]
       in time $O((K|S|+|A|)N^{o(1)}) = O(K^2|A| N^{o(1)})$.
       Then, 
       \begin{align*}
       |B_1| \ge \sum_{b\in A: (1_S\conv 1_{-A})[b]\ge \frac{|A|}{16K}}\tfrac{(1_S\conv 1_{-A})[b]}{|A|} &\ge \tfrac{1}{|A|}\Big (\big (\sum_{b\in A}(1_S\conv 1_{-A}[b]\big ) - |A|\cdot \tfrac{|A|}{16K}\big )\Big )\\
       & = \tfrac{1}{|A|}\Big (\big (\sum_{s\in S}(1_A\conv 1_{A}[s]\big ) - |A|\cdot \tfrac{|A|}{16K}\big )\Big )\\
       & \ge \tfrac{7|A|}{16K}. \tag{by \cref{eqn:sumxaconva}}
       \end{align*}
       Now we show a lower bound on $\sfE(B_1)$. For any $b\in A\setminus B_1$, $|\{a\in A: a+b\in S\}| = (1_S\conv 1_{-A})[b] < \frac{|A|}{16K}$, so $|\{(a,c,d)\in A^3: a+b=c+d\in S\}| < \frac{|A|^2}{16K}$, and hence
       \begin{equation}
|\{(a,b,c,d)\in A\times (A\setminus B_1)\times A\times A: a+b=c+d\in S\}| < \tfrac{|A|^3}{16K}.\label{eqn:energyaminusb}
       \end{equation}
       One can get three other analogous inequalities by symmetry. Then,
       \begin{align*}
        \sfE(B_1) &\ge |\{(a,b,c,d)\in B_1^4: a+b=c+d\in S\}|\\
        & \ge |\{(a,b,c,d)\in A^4: a+b=c+d\in S\}| - 4\cdot \tfrac{|A|^3}{16K} \tag{by \cref{eqn:energyaminusb}}\\
        & \ge \tfrac{|A|^3}{4K} \tag{by \cref{eqn:sumxaconvasquared}}\\
        & \ge \tfrac{|B_1|^3}{4K}.
       \end{align*}

       Assume $|A|\ge 192K^{r+3}$; otherwise, return $A'=A$ with any $B'\subseteq A, |B'|=1$ as answer, which satisfies  all the claimed properties for large enough constant $C$.
 Apply \cref{thm:find-subset-highenergy} to $B_1$ with parameter $4K$ and size parameter $R :=  \frac{|A|}{3K^r} \le |B_1|$ (which satisfy the required condition $4K\le R^{1/3}$ due to $|A|\ge 192K^{r+3}$), and obtain $B_0\subseteq  B_1$ in $O(K^{13}|B_1|N^{o(1)})$ time such that \[\sfE(B_0) \ge \tfrac{|B_0|^3}{4K} \cdot 2^{-\frac{C\log N}{\sqrt{\log\log N}}},\text{ and } \tfrac{|A|}{48K^{r+2}} = \tfrac{R}{16K^2}\le |B_0| \le R\cdot  2^{\frac{C\log N}{\sqrt{\log \log N}}} = \tfrac{|A|}{3K^r}\cdot 2^{\frac{C\log N}{\sqrt{\log\log N}}}.\]

Then, apply \cref{cor:schoenalgo} to $B_0$ with $c := (2^{13}K)^{-1}$ to find $X\subseteq B_0$ such that \[|X| \ge \tfrac{|B_0|}{12K} 2^{-\frac{C\log N}{\sqrt{\log \log N}}}\] and 
    \[ \big \lvert \big \{(b,b')\in X\times X: (1_{B_0}\conv 1_{-B_0})[b-b'] < c\cdot \tfrac{|B_0|}{12K} 2^{-\frac{C\log N}{\sqrt{\log \log N}}}\big \}\big \rvert \le 18c |X|^2,\]
     in $O(|B_0|^2K^2 c^{-2} N^{o(1)}) = O(R^2K^4 N^{o(1)}) = O(|A|^2/K^{2r-4} \cdot N^{o(1)})$ time.
Then, compute \[B':= \Big \{b\in X \,:\,  \big \lvert \big \{b'\in X: (1_{B_0}\conv 1_{-B_0})[b-b']< c\cdot \tfrac{|B_0|}{12K} 2^{-\frac{C\log N}{\sqrt{\log \log N}}}\big \}\big \rvert \le 36c|X| \Big \},\]
directly in $\tilde O(|B_0|^2+|X|^2)= \tilde O(|B_0|^2)$ time (not a bottleneck), which satisfies \[|B'|\ge |X| - \tfrac{18c |X|^2}{36c|X|} = \tfrac{|X|}{2}\ge  \tfrac{|B_0|}{24K} 2^{-\frac{C\log N}{\sqrt{\log \log N}}} \ge \tfrac{|A|}{48 K^{r+2} 24K} 2^{-\frac{C\log N}{\sqrt{\log \log N}}}  \ge \tfrac{|A|}{K^{r+3}} 2^{-\frac{C'\log N}{\sqrt{\log \log N}}}.\] 
    
Since $X\subseteq B_0\subseteq B_1$, 
by the definition of $B_1$ we have $(1_S\conv 1_{-A})[x] \ge \tfrac{|A|}{32K}$ for all $x\in X$,
so 
\begin{equation}
\sum_{a\in A}(1_S\conv 1_{-X})[a] = \sum_{x\in X}(1_S\conv 1_{-A})[x] \ge |X|\cdot \tfrac{|A|}{32K}.
\label{eqn:sumsminusxa}
\end{equation}
Use \cref{thm:introdeterministicapprox3sum} to compute a subset $A'\subseteq A$ such that
\[ \{ a\in A: (1_{S}\conv 1_{-X})[a]\ge \tfrac{|X|}{64K} \} \subseteq A' \subseteq \{ a\in A: (1_{S}\conv 1_{-X})[a]\ge \tfrac{|X|}{128K} \}\]
in time complexity $O((K|S| + |X| + |A|)N^{o(1)}) = O(K^2|A|N^{o(1)})$, where
    \[|A'|\ge \sum_{a\in A: (1_S\conv 1_{-X})[a]\ge \frac{|X|}{64K}}\tfrac{(1_S\conv 1_{-X})[a]}{|X|} \ge \tfrac{1}{|X|}\Big (\big (\sum_{a\in A}(1_S\conv 1_{-X})[a]\big ) - |A|\cdot \tfrac{|X|}{64K}\Big ) \underset{\text{by \cref{eqn:sumsminusxa}}}{\ge} \tfrac{|A|}{64K}. \]

    We return $A',B'\subseteq A$ as the answer, and it remains to show $|A'+B'|$ is small. By the definitions of $A',B'$, for every $a\in A',b\in B'$, there exist at least $\frac{|X|}{128K} - 36c |X| \ge \frac{|X|}{2^9K}$ many $b'\in X$ such that $a+b'\in S$ and $(1_{B_0}\conv 1_{-B_0})[b-b'] \ge c\cdot \tfrac{|B_0|}{12K} 2^{-\frac{C\log N}{\sqrt{\log \log N}}}$, so 
the number of distinct ways to represent $a+ b$ as  $a+b=(a+b') + (b-b')  =(a+b')+ b_0-b_0'$, where  $a+b'\in S, b_0\in B_0, b_0'\in B_0$, is at least $\frac{|X|}{2^9 K} \cdot c\cdot \tfrac{|B_0|}{12K} 2^{-\frac{C\log N}{\sqrt{\log \log N}}}$. Therefore,
\[|A'+B'| \le \frac{|S|\cdot |B_0|\cdot |B_0|}{\frac{|X|}{2^9 K} \cdot c\cdot \tfrac{|B_0|}{12K} 2^{-\frac{C\log N}{\sqrt{\log \log N}}}}  = |S|\cdot \frac{|B_0|}{|X|} \cdot K^3\cdot (2^9\cdot 2^{13}\cdot 12)2^{\frac{C\log N}{\sqrt{\log \log N}}} \le |A| K^5  2^{\frac{C'\log N}{\sqrt{\log \log N}}},\] 
as claimed.
The total time complexity of this algorithm is 
$O(K^2|A|N^{o(1)}) + O(K^{13}|B_1|N^{o(1)})+ O(|A|^2/K^{2r-4} \cdot N^{o(1)}) =O((|A|^2/ K^{2r-4} + K^{13}|A|)N^{o(1)}) $.
    \end{proof}

\section{Tool V: Deterministic 3SUM Counting for Popular Sums}

In this section we prove \cref{thm:introdetpopsum}, restated below.
\thmpopular*
\begin{proof}
  Let parameters $K\ge 1$ be determined later. Initialize $\hat A:= A$. We perform the following loop for iterations $i=1,2,\dots$: 
  \begin{itemize}
   \item If $|\hat A| \le |A|/K^{1/3}$ (in particular, $\sfE(\hat A) \le |A|^3/K$), then exit the loop.
   \item Use \cref{cor:approxenergy} to estimate $\sfE(\hat A)$ up to additive error $\frac{|\hat A|^3}{2K}$, in $O(K|\hat A|N^{o(1)})$ time. If the approximate energy is $\le \frac{3|\hat A|^3}{2K}$, then we exit the loop (and we know $\sfE(\hat A) \le \frac{2|\hat A|^3}{K}$). 
      \item Now we know $\sfE(\hat A) \ge  \frac{|\hat A|^3}{K}$. Apply our deterministic constructive BSG theorem (\cref{thm:derandbsg} with $r=3$) to $\hat A$, and obtain subsets $A_i,B_i \subseteq  \hat A$ such that 
         \begin{equation}
      |A_i| \ge |\hat A|/(64K),\,\,
|B_i| \ge |\hat A|/\big (K^{6}\cdot N^{o(1)}\big ),\,\, |A_i+B_i|\le K^5|\hat A|\cdot N^{o(1)}.
\label{eqn:biproperty}
         \end{equation}
\item Let $\hat A \gets \hat A \setminus A_i$.
  \end{itemize}
  Since after each iteration  $|\hat A|$ drops to at most $(1-\frac{1}{64K})|\hat A|$, the number of iterations before $|\hat A|$ drops below $|A|/K^{1/3}$ is at most $O(K\log K)$. By \cref{thm:derandbsg}, each iteration takes time $O((|\hat A|^2/ K^{2} + K^{13}|\hat A|)N^{o(1)})$. So the total time for all iterations is at most $O((|A|^2/ K + K^{14}|A|)N^{o(1)})$. In the end we obtain a decomposition $A = A_1 \cup \dots \cup A_{g} \cup \hat A$, where $g=O(K\log K)$, $\sfE(\hat A) \le \frac{2|A|^3}{K}$, along with sets $B_i$ satisfying \cref{eqn:biproperty}.
  
  Use \cref{thm:introdeterministicapprox3sum} to deterministically compute a set $\tilde S$ such that 
  \[ \{s: (1_{A}\conv 1_B)[s] \ge |A|/k\} \subseteq \tilde S \subseteq \{s: (1_{A}\conv 1_B)[s] \ge |A|/(2k)\}\]
  in $O((|A|+k|B|)N^{o(1)})$ time, where $|\tilde S|\le \frac{|A||B|}{|A|/2k} = 2k|B|$. Our goal is to exactly compute $(1_A\conv 1_B)[s]$ for all $s\in \tilde S$.
  
 For every $1\le i\le g $, use \cref{thm:introdeterministicapprox3sum} to deterministically approximate $(1_{A_i}\conv 1_B)[s]$ for all $s\in \tilde S$ up to additive error $\frac{|A_i|}{12k}$, in time $\tilde O((k|B|+|A_i|)N^{o(1)})$, and total time $O((Kk|B| + |A|)N^{o(1)})$ for all $1\le i\le g$. Then we define a partition $\tilde S = S_1\cup S_2 \cdots \cup S_g \cup \hat S$ as follows: for each $s\in \tilde S$, pick an arbitrary $1\le i\le g$ such that the approximate value of $(1_{A_i} \conv 1_B)[s]$ is at least $\frac{|A_i|}{6k}$ (so that the real value satisfies 
 $(1_{A_i} \conv 1_B)[s] \ge \frac{|A_i|}{12k}$), and put $s$ in $S_i$. If no such $i$ exists, then put $s$ in $\hat S$. 
 Note that every $s\in \hat S$ satisfies 
 \begin{equation}
 (1_{\hat A}\conv 1_B)[s] = (1_{A}\conv 1_B)[s] - \sum_{i=1}^g (1_{A_i}\conv 1_B)[s] \ge  \frac{|A|}{2k} - \sum_{i=1}^g(\frac{|A_i|}{6k} + \frac{|A_i|}{12k}) \ge \frac{|A|}{4k}.
 \label{eqn:smultlowerbound}
 \end{equation}
 We have the following claims:
 \begin{claim}
   \label{claim7.2}
For all $1\le i\le g$, $|S_i +  B_i| \le O(kK^6|B|N^{o(1)})$.
 \end{claim}
 \begin{proof}
  For any $s\in S_i, b\in B_i$, there are at least $(1_{A_i}\conv 1_B)[s] \ge \frac{|A_i|}{12k}$ many distinct representations of $s+b= (s-a) + (a+b)$ such that $a\in A_i, s-a\in B$, $a+b\in A_i+B_i$. Hence, $(1_{B}\conv 1_{A_i+B_i})[s+b] \ge \frac{|A_i|}{12k}$. Therefore, $|S_i+B_i| \le \frac{|B|\cdot |A_i+B_i|}{|A_i|/(12k)} \le 12k|B| \cdot 64K^6 N^{o(1)}$, where the last step uses \cref{eqn:biproperty}.
 \end{proof}
 We use \cref{thm:detsmalldoublethreesum} (applied to $(A,B,C,S):= (S_i,-B,A,B_i)$) to compute $(1_{A}\conv 1_{B})[s]$ for all $s\in S_i$, in 
   time complexity 
  \begin{align*}
 &   O\Big(\frac{|S_i+B_i|\sqrt{|-B||A|}}{\sqrt{|B_i|}}\cdot N^{o(1)} 
   \Big )\\
    & \le O\Big(\frac{kK^6|B|\sqrt{|B||A|}}{\sqrt{\frac{|A|}{K^{1/3}}/K^{6}}}\cdot N^{o(1)} 
   \Big ) \tag{by \cref{eqn:biproperty} and \cref{claim7.2}}\\
   & \le O(|B|^{1.5}k K^{55/6} N^{o(1)}).
  \end{align*} 
   Summing over all $1\le i\le g$, the total time is $O(|B|^{1.5}k K^{61/6} N^{o(1)})$.

\begin{claim}
  $|\hat S| \le O\left (k^2\sqrt{\frac{|B|^3}{K|A|}}\right )$. 
\end{claim}
\begin{proof}
  By \cref{eqn:smultlowerbound}, 
  \[  |\hat S|\cdot \left (\frac{|A|}{4k}\right )^2 \le \sum_{s\in \hat S} ((1_{\hat A}\conv 1_{B})[s])^2 \le \sfE(\hat A,B) \le \sqrt{\sfE(\hat A) \sfE(B)} \le \sqrt{\frac{2|A|^3}{K} \cdot |B|^3},  \]
  so $|\hat S| \le O\left (k^2\sqrt{\frac{|B|^3}{K|A|}}\right )$.
\end{proof}
We use brute force to compute $(1_A\conv 1_B)[s]$ for all $s\in \hat S$, in $\tilde O(|B| + |\hat S||A|) = \tilde O (k^2 \sqrt{|A||B|^3/K})$ time.

Therefore, we have computed $(1_A\conv 1_B)[s]$ for all $s\in \tilde S = S_1\cup S_2 \cdots \cup S_g \cup \hat S$, in total time 
\[O\Big ( \Big ((|A|^2/ K + K^{14}|A|)  +  (|B|^{1.5}k K^{61/6} ) + k^2 \sqrt{|A||B|^3/K} \Big )\cdot N^{o(1)}\Big )\] 
By setting $K= |A|^{3/64}$, the total time becomes
\[ O\Big ((|A|^{2-3/64} + k^2|B|^{1.5}|A|^{61/128})\cdot N^{o(1)}\Big ) \le O(k^2\cdot |A|^{2-3/128} N^{o(1)}).\]
\end{proof}
\section{Deterministic Approximate Text-to-Pattern Hamming Distances}
\label{sec:ham}
We use our deterministic approximate 3SUM-counting algorithm (\cref{thm:introdeterministicapprox3sum}) to show an $O(nm^{o(1)}/\eps)$-time deterministic algorithm for $(1+\eps)$-approximate Text-to-Pattern Hamming Distances, proving \cref{thm:dethdmain}.

\begin{proof}[Proof of \cref{thm:dethdmain}]
    Let $T[1\dd n]$ and $P[1\dd m]$ denote the input text and pattern strings over an alphabet $\Sigma$.  By dividing $T$ into $O(n/m)$ chunks, we can  assume $n=O(m)$ without loss of generality.

    We start by describing a very simple $\eps m$-additive approximation algorithm. For each symbol $s\in \Sigma$, let $A_s:= \{j: T[j]=s\}$ and $B_s:=\{k: P[k]=s\}$. Then the Hamming distance between $T[i+1\dd i+m]$ and $P[1\dd m]$ can be expressed as 
    \[ |P| - \sum_{s\in \Sigma} (1_{A_s} \conv 1_{-B_s})[i].\]
    For each $s\in \Sigma$, use \cref{thm:introdeterministicapprox3sum} to compute in $\tilde O((\eps^{-1}|A_s|+|B_s|)\cdot m^{o(1)})$ time a sparse vector $f_s$ such that $\|f_s - 1_{A_s} \conv 1_{-B_s}\|_\infty\le \eps |B_s|$. The total time complexity is $\tilde O((\eps^{-1}\sum_{s\in \Sigma}|A_s|+\sum_{s\in \Sigma}|B_s|)\cdot m^{o(1)}) = \tilde O(\eps^{-1}m^{1+o(1)})$.
    Then, sum up the vectors $f_s$ over $s\in \Sigma$, and output $|P| - \sum_{s\in \Sigma}f_s[i]$ as our approximation for the Hamming distance at shift $i$, which clearly has additive error at most $\sum_{s\in \Sigma} \big \lvert \big (f_s - 1_{A_s} \conv 1_{-B_s}\big )[i]\big \rvert\le \sum_{s\in \Sigma}\eps |B_s| = \eps m$.

    This algorithm already achieves $(1+\eps)$-approximation for distances that are $\Theta(m)$, but in general we still need to handle smaller distances $\Theta(k)$ and achieve better additive error $O(\eps k)$ for them. 
    Similarly to \cite[Lemma 3.5]{focs23}, we handle this using the technique of Gawrychowski and Uzna\'{n}ski \cite{GawrychowskiU18}, which reduced $k$-bounded Text-to-Pattern Hamming Distances to the special case where both text and pattern strings have run-length encodings bounded by $O(k)$ (i.e., they are concatenations of $O(k)$ blocks of identical characters). This reduction runs in near-linear time, and preserves additive approximation.
    In the following we show how to approximate $\sum_{s\in \Sigma}(1_{A_s}\conv 1_{-B_s})$ with $O(\eps k)$ coordinate-wise additive error in this special case in $\tilde O(\eps^{-1}m^{1+o(1)})$ time, from which the theorem follows.

    Now each $A_s, B_s$  is a union of blocks (intervals), and there are $O(k)$ blocks in total over all symbols $s\in \Sigma$. Since these $O(k)$ blocks have total length $O(m)$, we can assume each block has length $O(m/k)$ (by possibly breaking up longer blocks).
     Then, we use the same strategy as \cite[Lemma 3.5]{focs23} and assume each block is dyadic (i.e., it takes the form $\{z\cdot 2^u, z\cdot 2^u+1,\dots,(z+1)\cdot 2^u-1\}$ for some integers $z,u\ge 0$), by breaking up each original block into at most $2\log (m/k) +O(1)$ blocks; the total number of blocks becomes $O(k\log m)$.
     For $s\in \Sigma$ and $1\le 2^u \le O(m/k)$, define set $A^{u}_s$ as the left endpoints of $2^u$-length blocks in $A_s$, and define $B^{u}_s$ similarly.  In particular, $1_{A_s} = \sum_u (1_{A^u_s} \conv 1_{\{0,1,\dots,2^{u}-1\}})$. We can decompose $(1_{A_s}\conv 1_{-B_s})$ into the sum of contribution from pairs of blocks, and obtain
     \begin{equation}
      \sum_{s\in \Sigma}(1_{A_s}\conv 1_{-B_s}) = \sum_{s\in \Sigma}\sum_{u,v}1_{A_s^u} \conv 1_{-B_s^v} \conv \underbrace{1_{\{0,\dots,2^u-1\}} \conv 1_{-\{0,\dots,2^v-1\}}}_{=: L_{u,v}}.
      \label{eqn:apxequation}
     \end{equation}
     For each $s,u,v$, we use \cref{thm:introdeterministicapprox3sum} to compute $f_{s,u,v}$ such that $\|f_{s,u,v} - 1_{A_s^u}\conv 1_{-B_s^v}\|_\infty\le \frac{\eps k}{m} |B_s^v|$, in $O\big (\big ((\eps^{-1}m/k)|A_s^u|+|B_s^v|\big )\cdot m^{o(1)}\big )$ time. The total time complexity over all $s,u,v$ is $O(\eps^{-1}m^{1+o(1)})$ since $\sum_{u,s}|A_s^u|,\sum_{v,s}|B_s^v|\le O(k\log m)$ and $u,v\le O(\log m)$.
      When $u\ge v$ (the $u<v$ case is symmetric),
     $L_{u,v}$ is a multiset with multiplicity at most $2^v$ and supported on an interval of length $<2\cdot 2^u$,
 and   $A_s^u$ and $B_s^v$ only contain multiples of $2^v$, so for any $i$ the absolute difference between $(1_{A_s^u} \conv 1_{-B_s^v} \conv L_{u,v})[i]$ and $(f_{s,u,v} \conv L_{u,v})[i]$ (WLOG  assuming $f_{s,u,v}$ is only supported on multiples of $2^v$) is at most
 \begin{align*}
  \sum_{j} \big \lvert (f_{s,u,v}-1_{A_s^u}\conv 1_{-B_s^v}) [j]\big \rvert \cdot L_{u,v}[i-j]&\le \frac{2\cdot 2^u}{2^v} \|f_{s,u,v} - 1_{A_s^u}\conv 1_{-B_s^v}\|_\infty  \|L_{u,v}\|_\infty \\ & \le \frac{2\cdot 2^u}{2^v} \cdot \frac{\eps k}{m} |B_s^v| \cdot 2^v \\ & = O(\eps |B_s^v|).
 \end{align*}
 Hence, if we use $f_{s,u,v}$ to replace $1_{A_s^u}\conv 1_{-B_s^v}$ in \cref{eqn:apxequation} and compute the summation using FFT in $\tilde O(m)$ time, we can obtain an approximation of  $\sum_{s\in \Sigma}(1_{A_s}\conv 1_{-B_s})$ with coordinate-wise additive error $\sum_{s,u,v} O(\eps |B_s^v|) \le O(\eps k\log^2 m)$. This error can be made $O(\eps k)$ by suitably decreasing $\eps$. The total time complexity is $ O(\eps^{-1}m^{1+o(1)})$.
\end{proof}

\section{Deterministic 3SUM Reductions}
In this section, we prove the deterministic fine-grained reduction from 3SUM to 3SUM with small additive energy \cref{thm:smalldoub3sumhardnessmain}.
The overall framework is the same as previous randomized reductions in \cite{AbboudBF23,JinX23}, but here we need to make everything deterministic, by
combining the derandomized versions of the BSG theorem, the small-doubling 3SUM algorithm, and the derandomized almost additive hashing from previous sections.

\subsection{From high-energy to moderate-energy}
Throughout this section, we consider the following version of 3SUM: given $A\subset \{-n^3,\dots,n^3\}$ of size $|A|=n$, determine if there exists $(a,b,c)\in A\times A\times A$ such that $a+b+c=0$.\footnote{The case where $A$ comes from a larger universe can be deterministically reduced to this version by \cite[Theorem 1.2]{fischer3sum}.} The additive energy of the 3SUM instance is $n^2< \sfE(A)< n^3$.

\begin{lemma}
\label{lem:hightomod}
   There is a deterministic fine-grained reduction from 3SUM to 3SUM with additive energy at most $n^{2.95}$.
\end{lemma}
\begin{proof}
    The randomized version of this reduction was already given by \cite{AbboudBF23} and \cite{JinX23} using a randomized algorithmic version of the BSG theorem. Our derandomized version follows the same idea, except that we need to use the deterministic BSG theorem and deterministic energy approximation algorithm.
    
    Let $A_0$ denote the 3SUM instance, $|A_0|=n$. Let $K = n^{0.05}$.
    Maintain a subset $A\subseteq A_0$, which initially equals $A_0$.
      We iteratively proceed as follows. In each iteration:
      \begin{enumerate}
       \item If $|A|\le n^{0.9}$, then solve 3SUM on $A$ in $O(|A|^2)= O(n^{1.8})$ time, and return the 3SUM answer.
        \item Use \cref{cor:approxenergy} to deterministically approximate $\sfE(A)$ up to additive error $\frac{|A|^3}{3K}$, in $ O(K|A|n^{o(1)})$ time. If the approximate value is at most $\frac{2|A|^3}{3K}$, then we know $\sfE(A)\le |A|^3/K \le |A|^{2.95}$, and we can use the moderate-energy 3SUM solver to solve $A$, and return the 3SUM answer.
            \item Otherwise, we have $\sfE(A)\ge \frac{|A|^3}{3K}$.
                We apply the deterministic BSG algorithm (\cref{thm:derandbsg} with $r=3$) to $A$ with parameter $3K$, and in  $O(|A|^2/K^2 + K^{13}|A|)\cdot n^{o(1)}$ time  we can obtain $A'\subseteq A$ of size $|A'|\ge |A|/(192K)$ and another set $B'$ such that $|B'|\ge |A|/(K^6n^{o(1)}), |A'+B'|\le K^5 |A| n^{o(1)}$.
                We then use \cref{thm:detsmalldoublethreesum} (applied to $(A,B,C,S):= (A',A,-A,B')$) to decide whether $(1_{-A}\conv 1_{-A})[a]>0$ for some $a\in A'$ (i.e., there exists $(a,b,c)\in A'\times A\times A$ such that $-b-c=a$), in time complexity 
                \begin{align*}
& O\Big(\frac{|A'+B'|\sqrt{|A||A|}}{\sqrt{|B'|}}\cdot \polylog(n) 
 + \big (|A'+B'|+|A|\big)\cdot n^{o(1)} \Big ) \\
 &= O((\frac{K^5|A|^2}{\sqrt{|A|/K^6}} + K^5|A|)n^{o(1)})\\
 & = O(K^8 |A|^{1.5}n^{o(1)}).
                \end{align*}
                If a 3SUM solution is found, then we return and terminate. Otherwise, we know $A$ has no 3SUM solutions that involve numbers in $A'$, so we can remove $A'$ and let $A\gets A\setminus A'$, and proceed to the next iteration.
      \end{enumerate}
      Since $|A'|\ge |A|/(192K)$, each iteration decreases $|A|$ to at most $(1-1/(192K))|A|$, so the number of iteration before $|A|$ drops below $n^{0.9}$ is at most $O(K\log n)$. The time complexity in each iteration is 
      \[ O(K|A|n^{o(1)}) + O((|A|^2/K^2 + K^{13}|A|)n^{o(1)}) + O(K^8|A|^{1.5}n^{o(1)}),\]
      so over all $O(K\log n)$ iterations, the total time is at most
      \[ n^{o(1)}\cdot O(K^{14}n + K^{9}n^{1.5} + n^2/K ) \le O(n^{1.95+o(1)}). \]
      So the time complexity of this reduction is truly subquadratic.
\end{proof}

\subsection{From moderate-energy to low-energy}
\label{sec:modtolowreduction}

In this section, we use our almost additive hash family $\caH$ from \cref{def:hashfamilyh} and the deterministic energy approximation algorithm (\cref{cor:approxenergy}) to show the following fine-grained reduction:

\begin{lemma}
\label{lem:modtolow}
   Let $\alpha \in (0,1)$ be a fixed constant. 
There is a deterministic fine-grained reduction from 3SUM with additive energy at most $n^{2+\alpha}$ to 3SUM with additive energy at most $n^{2+\beta}$, where $\beta = \alpha-\alpha(1-\alpha)/30 \in (0, \alpha)$.
\end{lemma}
\begin{proof}
   Suppose all 3SUM instances of additive energy at most $n^{2+\beta}$  can be solved in $n^{2-\delta}$ time for some $\delta>0$.
   
Let $A\subset [N]$ be the input 3SUM instance with $\sfE(A)\le |A|^{2+\alpha}$, with $|A|=n$, $N=\poly(n)$.
We will perform the standard 3SUM self-reduction using the hash family $\caH \subset \{h\colon \{-N,\dots,N\} \to \F_p\}$ from \cref{lem:dethash}, with parameter $k=3$, and prime $2\le p<|A|$ to be determined later.

By almost linearity 
 of $\caH$ (\cref{lem:almostlinearity}), any 3SUM solution $a+b+c=0$ satisfies $h(0)-h(a)-h(b+c) \in \Delta_h$ and $h(b+c) - h(b)-h(c) \in \Delta_h$, so $h(a)+h(b)+h(c) \in h(a)+ h(b+c)-\Delta_h \in h(0)-2\Delta_h$.
For $h\in \caH$ and $i\in \F_p$, let $h^{-1}(i)= \{x\in \{-N,\dots,N\}: h(x)=i\}$.
Define $\Delta'_h = h(0)-2\Delta_h$ which has size $|\Delta'_h|\le |\Delta_h|^2 = n^{o(1)}$.
So in order to solve the original 3SUM instance, we only need to solve the smaller 3-partite 3SUM instances $\big (h^{-1}(i)\cap A, h^{-1}(j)\cap A, h^{-1}(k)\cap A\big )$, where $i+j+k\in \Delta'_h$. (There are $p\cdot p\cdot |\Delta'_h| = p^2n^{o(1)}$ such instances.)

We need to reduce these 3-partite 3SUM  instances back to the one-partite version that we are working with, as follows:
 suppose $X,Y,Z\subseteq \{-N,\dots,N\}$ are the three input sets of the 3-partite 3SUM. we define set $W = (X+30N) \cup (Y + 100N) \cup (Z - 130N)$ as the one-partite 3SUM instance. Note that any original 3SUM solution $(x,y,z)\in X\times Y\times Z$, $x+y+z=0$ implies a solution $(x+30N)+(y+100N)+(z-130N)=0$ in the new instance. One can verify that, conversely, any 3SUM solution in $W$ correspond to an original solution, due to the way we shifted the sets. Furthermore, we can similarly verify that any Sidon 4-tuple in $W$ can be divided into two pairs, $(a,b),(c,d)$, where $a-b=c-d$, and the two numbers in each pair come from the same part ($X+30N$, $Y+100N$, or $Z-130N$).  
 So it actually corresponds to a (shifted) Sidon 4-tuple from the original instance in $X,Y,Z,$ or between $X,Y$, between $Y,Z$, or between $Z,X$.
 This means the new one-partite 3SUM instance $W$ has additive energy \[ \sfE(W) = \sfE(X) + \sfE(Y)+\sfE(Z) + \sfE(X,Y)+\sfE(Y,Z)+\sfE(Z,X).\]

Our goal is to deterministically pick a $h\in \caH$ to perform the 3SUM self-reduction mentioned above, so that the produced instances have low total additive energy.
To do this, we try every $h\in \caH$, and check whether $h$ is good.
Recall from \cref{def:hashfamilyh} that $|\caH| \le 4m^2 p^{12}$.

We first need to analyze the expected total energy of the produced instances under a random $h\in \caH$.
To do this, we consider each individual Sidon 4-tuple in $A$, $a+b=c+d$, and analyze its expected contribution to the total energy of all produced instances.
Given the discussion above, the contribution of this 4-tuple to the total energy (under a specific $h\in \caH$) is determined as follows:
\begin{itemize}
    \item 
If $h(a)=h(b)=h(c)=h(d)=i^*$ for some $i^*\in \F_p$, then they are hashed to the same bucket $i^*$, and it corresponds to a Sidon 4-tuple in every instance  
$\big (h^{-1}(i^*)\cap A, h^{-1}(j)\cap A, h^{-1}(k)\cap A\big )$, where $i^*+j+k\in \Delta'_h$.  There are $O(p\cdot |\Delta'|)$ such choices of $j,k$, so this Sidon 4-tuple contributes $O(p\cdot |\Delta'|)$ to the total energy.
\item 
If $h(a),h(b),h(c),h(d)$ are not distinct but not all equal, then by a similar reasoning, this Sidon 4-tuple can contribute at most $O(|\Delta'|)$ to the total energy.
\item 
If $h(a),h(b),h(c),h(d)$ are distinct,  then they are hashed to four different buckets, and hence do not contribute to the total energy.
\end{itemize}

Now we can analyze the expected contribution of the Sidon 4-tuple $a+b=c+d$ to the total energy, by distinguishing several cases:
\begin{itemize}
      \item if $a=b=c=d$, then clearly $h(a)=h(b)=h(c)=h(d)$. 
      There are only $|A|$ such 4-tuples.
         \item Otherwise, there are at least two distinct values among $a,b,c,d$. Say they are $a,b$. Note that $a\neq b$ cannot have any $2$-term $\ell$-relation as defined in \cref{defn:lrelation}, so by \cref{lem:almostkwiseindep} we know $h(a)=h(b)$ happens with at most $n^{o(1)}/p$ probability.
            Now we look at more subcases.
            \begin{itemize}
      \item If $a,b,c$ has no $O(Q^2)$-relation, then by \cref{lem:almostkwiseindep}, $h(a)=h(b)=h(c)$ happens with probability at most $n^{o(1)}/p^2$. 
      The number of such 4-tuples are trivially upper-bounded by   $\sfE(A)$.
   \item Otherwise, $a,b,c$ has a $O(Q^2)$-relation. Then there are at most $(O(Q^2))^2-1 = n^{o(1)}$ ways to pick the coefficients $\beta_1+\beta_2+\beta_3=0$ of this relation, and there are $|A|^2$ possibilities for $(a,b)\in A^2$, which then uniquely determine $c$ by the relation $\beta_1a+\beta_2b+\beta_3c=0$, and then uniquely determine $d$ by $a+b=c+d$.
   Hence, in total there are at most $n^{o(1)}|A|^2$ such 4-tuples. 
            \end{itemize}
\end{itemize}
Hence, the expected number of tuples $(a,b,c,d)\in A^4$ such that $h(a)=h(b)=h(c)=h(d)$ can be bounded by $|A| + \sfE(A)\cdot n^{o(1)}/p^2 + n^{o(1)}|A|^2 \cdot n^{o(1)}/p$.
The expected number of 
$(a,b,c,d)\in A^4$ such that 
$h(a),h(b),h(c),h(d)$ are not distinct can be bounded by $|A| + \sfE(A)\cdot n^{o(1)}/p$.
Together, their contribution to the total energy of all instances is in expectation
\begin{align}
& \big (|A| + \sfE(A)\cdot n^{o(1)}/p^2 + n^{o(1)}|A|^2 \cdot n^{o(1)}/p\big ) \cdot O(p\cdot |\Delta'|)   + \big (|A| + \sfE(A)\cdot n^{o(1)}/p\big ) \cdot O(|\Delta'|)\nonumber  \\
& = O( (|A|p + \sfE(A)/p + |A|^2)\cdot n^{o(1)}).
\label{eqn:energybound}
\end{align}

We can also use a similar (and simpler) argument  to bound the expected total squared size of the buckets, by
\begin{equation}
    \Ex_{h\in \caH}\Big [\sum_i |h^{-1}(i)\cap A|^2\Big ] \le  O(|A|^2\cdot n^{o(1)}/p).
\label{eqn:squaresizebound}
\end{equation}

By Markov's inequality and a union bound, we know there is an $h \in\caH$ such that both upper bounds \cref{eqn:energybound,eqn:squaresizebound} hold (up to a constant factor) without taking expectation. The goal is to find such $h$ (or any other good enough $h$), by enumerating $h\in \caH$ and checking whether $h$ is good deterministically.

Let $\eps>0$ be some constant to be determined later.
For each $h\in \caH$, do the following steps:
\begin{itemize}
    \item 
    If the bound in \cref{eqn:squaresizebound} is violated for $h$ then we terminate this procedure and start to check the next hash function.
    Otherwise, use $h$ to perform the 3SUM self-reduction described above and obtain $p^2|\Delta'|=p^2n^{o(1)}$ smaller 3SUM instances.  
    \item 
    All instances of size $< |A|^{1-\eps}/p$ are marked as \emph{light}.
    All buckets of size $|h^{-1}(i)\cap A| \ge \frac{1}{3}|A|^{1+\eps}/p$ are marked as \emph{heavy}. 
    All instances that include at least one heavy buckets are marked as \emph{heavy}.
    \item 
    Let $\mu:=|A|^{\alpha - 1-4\eps}$.
    For each non-heavy instance (which can have size $\le |A|^{1+\eps}/p$), use \cref{cor:approxenergy} to approximate its additive energy up to additive error
    $\mu\cdot \big (|A|^{1+\eps}/p\big )^3 = |A|^{2+\alpha-\eps }/p^3$.   If the approximate energy is greater than  $0.5|A|^{2+\alpha + 4\eps }/p^3$, then mark this instance as \emph{high-energy}.
    
    The total time of this step over all $p^2n^{o(1)}$ instances is $p^2n^{o(1)}\cdot O(\mu^{-1}\cdot |A|^{1+\eps}/p \cdot n^{o(1)}) = p|A|^{2-\alpha+5\eps+o(1)}$.
    \item If the sum of the approximate energy (from the previous step) over all non-heavy instances is greater than $|A|^{2+\alpha+\eps}/p$, then $h$ is bad, and we terminate this procedure and start to check the next hash function.
    Otherwise, we declare $h$ as good, and we will not try other hash functions any more.

    The following steps are only run for the good $h$.
    \item 
    We solve all light, and high-energy instances by brute force in quadratic time. For each item $a\in A$ landing in a heavy bucket, use $\tilde O(|A|)$ time to find 3SUM solutions involving it.
    \item For each remaining instance that is non-heavy, it has size $s\in [|A|^{1-\eps}/p,|A|^{1+\eps}/p]$, and has additive energy at most  $0.5|A|^{2+\alpha+4\eps}/p^3 + \mu \big (|A|^{1+\eps}/p\big )^3 \le |A|^{2+\alpha+4\eps}/p^3$. 
    We will set parameters 
    later so that  
    \[
        |A|^{2+\alpha + 4\eps}/p^3 \le (|A|^{1-\eps}/p)^{2+\beta}
    \]
    holds,
    which is equivalent to 
   \begin{equation}
      |A|^{\alpha+(6+\beta)\eps - \beta} \le p^{1-\beta}. 
        \label{eqn:tohold}
   \end{equation} 
 Then this instance has additive energy at most $s^{2+\beta}$, and hence can be solved in $s^{2-\delta}$ time by assumption.
\end{itemize}

To show the correctness of this reduction, we need to show at least one $h$ is good. Note that when estimating the total additive energy of all $p^2n^{o(1)}$ instances, the total additive error incurred is 
\[ \mu\cdot \big (|A|^{1+\eps}/p\big )^3 \cdot p^2n^{o(1)}  = |A|^{2+\alpha-\eps +o(1)}/p. \]
By \cref{eqn:energybound} we know there is some $h$ such that the total energy is at most \[O((|A|^2+|A|^{2+\alpha}/p))\cdot n^{o(1)} \le O( |A|^{2+\alpha+o(1)}/p).\]
For such $h$, the approximate total energy is at most $O( |A|^{2+\alpha+o(1)}/p) + |A|^{2+\alpha-\eps +o(1)}/p \le  |A|^{2+\alpha + \eps }/p$, so our algorithm will declare this $h$ as good.
    On the other hand, if we declare some $h$ as good, then we know the total additive energy of all non-heavy instances is at most $|A|^{2+\alpha+\eps}/p + |A|^{2+\alpha-\eps +o(1)}/p \le 2|A|^{2+\alpha+\eps}/p$.

Now it remains to analyze the time complexity of the whole algorithm.
\begin{itemize}
    \item The energy estimation step takes $p|A|^{2-\alpha+5\eps+o(1)}$ time for each $h\in \caH$. We need to enumerate all $h\in \caH$ (where $|\caH| \le n^{o(1)}p^{12})$, so the total time is 
    \begin{equation}
    \label{eqn:tosubquad1}
    p^{13}|A|^{2+5\eps-\alpha+o(1)}.
    \end{equation}
    \item  By \cref{eqn:squaresizebound}, the total size of heavy buckets $i$ (that is, with size $|h^{-1}(i)\cap A|\ge \frac{1}{3}|A|^{1+\eps}/p$) is at most $O(|A|^2\cdot n^{o(1)}/p)/ \big (\frac{1}{3}|A|^{1+\eps}/p\big ) \le O(|A|^{1-\eps+o(1)})$. For each of them we spend $\tilde O(|A|)$ time to find 3SUM solutions involving it, which is $O(|A|^{2-\eps+o(1)})$ total time.
    \item The total time for solving (at most $p^2 n^{o(1)}$ many) light instances is $p^2 n^{o(1)} \cdot \big (|A|^{1-\eps}/p\big )^2 = |A|^{2-2\eps+o(1)}$.
    \item Since the total additive energy of all non-heavy instances is $\le 2|A|^{2+\alpha+\eps}/p$, the number of instances with that are marked as high-energy (which should have energy at least $0.5|A|^{2+\alpha+4\eps}/p^3 - \mu\cdot \big (|A|^{1+\eps}/p\big )^3 \ge 0.4|A|^{2+\alpha+4\eps}/p^3$) is at most 
    $\frac{2|A|^{2+\alpha+\eps}/p}{0.4|A|^{2+\alpha+4\eps}/p^3}= 5p^2/|A|^{3\eps}$.
    The total time for solving these high-energy non-heavy instances is  at most $\big (5p^2/|A|^{3\eps}\big )\cdot O((|A|^{1+\eps}/p)^2) \le O(|A|^{2-\eps+o(1)})$.
    \item The remaining instances are solved by the assumed subquadratic-time algorithm on lower energy  instances. The total time is 
    \begin{equation}
     p^2n^{o(1)} \cdot (|A|^{1+\eps}/p)^{2-\delta} 
    \label{eqn:tosubquad2}
    \end{equation}
\end{itemize}

The remaining goal is to set parameters $\beta\in (0,\alpha),p$, and $\eps>0$, so that \cref{eqn:tohold} holds, and \cref{eqn:tosubquad1,eqn:tosubquad2} are truly-subquadratic in $|A|$. This is achieved by setting 
\begin{align*}
\beta &= \alpha - \alpha(1-\alpha)/30,\\
p &= \Theta(|A|^{\alpha/15}),
\end{align*}
and $\eps>0$ small enough (in particular, $\eps:= \min\{ \alpha(1-\alpha)/420, \delta/100\}$).
To verify, note that \cref{eqn:tohold} holds because
\[ |A|^{\alpha+(6+\beta)\eps-\beta} = |A|^{\alpha(1-\alpha)/30+(6+\beta)\eps} \le |A|^{\alpha(1-\alpha)/20} = \Theta(p^{0.75(1-\alpha)}),\]
whereas $p^{1-\beta} = p^{(1-\alpha)(1+\alpha/30)} \ge p^{0.75(1-\alpha)}$.
\cref{eqn:tosubquad1} can be upper-bounded by 
\[ |A|^{13\alpha/15}\cdot |A|^{2+5\eps-\alpha+o(1)} \le |A|^{2+5\eps - 2\alpha/15 + o(1)} \le |A|^{2-\alpha/15+o(1)}.   \]
And \cref{eqn:tosubquad2} can be upper-bounded by $|A|^{(1+\eps)(2-\delta)+o(1)}p^{\delta} \le  |A|^{2+2\eps -\delta +\delta\alpha/15+o(1)}\le |A|^{2-0.9\delta+o(1)}$.
\end{proof}

Now we can prove  \cref{thm:smalldoub3sumhardnessmain}.
\begin{proof}[Proof of \cref{thm:smalldoub3sumhardnessmain}]
For any given constant $\delta>0$, we deterministically reduce from 3SUM to 3SUM with additive energy at most $n^{2+\delta}$ as follows: first apply \cref{lem:hightomod} to reduce from 3SUM to 3SUM with additive energy at most $n^{2.95}$. Then, starting with $\alpha_0:=0.95$, repeatedly apply \cref{lem:modtolow} to reduce from 3SUM with additive energy at most $\alpha_{i}$ to 3SUM with additive energy at most $\alpha_{i+1}= \alpha_i-\alpha_i(1-\alpha_i)/30$.
One can check that the sequences $\{\alpha_i\}$ converges to $0$. Hence, there is some step where we arrive at 3SUM instances of additive energy at most $n^{2+\delta}$ as desired.

Recall that we initially start from 3SUM with universe $O(n^3)$ (which is allowed by \cite[Theorem 1.2]{fischer3sum}). There are only finitely many steps of reductions. In each step, one can verify that  \cref{lem:modtolow} only blows up the universe size by a fixed polynomial (strictly speaking, the universe size does not change, but the instance size shrinks by a fixed polynomial factor).
Hence, in the end we arrive at 3SUM instances whose universe size is still $n^{O(1)}$ (where the $O(1)$ depends on $\delta$).
\end{proof}

\section{Reduction to Triangle Listing}
\label{sec:trianglist}
In this section, we prove the following theorem.
\begin{theorem}
\label{thm:to-trianglelisting}
    Given a 3SUM instance on a size $n$ set $A \subset [n^{O(1)}]$ with additive energy $\sfE(A) \le n^{2+\delta}$, a small parameter $\eps > 0$, and a parameter $P = n^{k \cdot \eps} \le n$ for some integer $k$, there is a deterministic reduction that reduces the 3SUM instance to $n^{o(1)}$  Triangle Listing instances on graphs with 
    \begin{itemize}
        \item $P^{2} \cdot n^{O(\eps)}$ nodes;
        \item $O(n^{1+\eps} / P)$ maximum degree;
        \item $n^{2+\delta+o(1)}/P + n^{4+o(1)}/P^4$ $4$-cycles, 
    \end{itemize}
    and we are required to report $n^{3+o(1)} / P^2$ triangles for each graph. 
    Furthermore, the running time of the reduction is $O(n^{2-\eps} + n^{1+O(\eps)} \cdot P)$.
\end{theorem}

\subsection{Warm-up: randomized reduction}
\label{sec:tri-listing-randomized}

Before we prove \cref{thm:to-trianglelisting}, we first show a randomized version of the reduction and we later derandomize it. The reduction is similar to \cite{AbboudBF23}'s reduction from low-energy 3SUM to Triangle Listing with few $4$-cycles (the difference is that they work over the group $\F_p^d$ for some prime $p$ and integer $d$, while we work over integers), and also shares similar high-level structure with \cite{JinX23}'s reduction from strong 3SUM to Triangle Detection. 

Let $p = \Theta(P)$ be a prime. 
Let $\mathcal{H}\colon  \{-N, \ldots, N\} \rightarrow \F_p$, for $N = n^{O(1)}$ be the
hash family constructed in \cref{def:hashfamilyh}, which is almost linear ``almost'' almost $k$-wise independent by \cref{lem:almostkwiseindep}. Note that in this case, two integers $x_1 \ne x_2$ cannot have an $n^{o(1)}$-relation, so the hash family is actually almost $2$-wise independent. 

Let $f, g, h$ be three hash functions sampled from $\mathcal{H}$.

For any $a \ne b \in A$, the probability $f(a) = f(b)$ is $O(1/ p)$. Therefore, 
\[
\Ex\left[\sum_{a, b \in A} [f(a) = f(b)]\right] = O(n^2 / p). 
\]
Therefore, the number of $a \in A$ that is in the same hash bucket (under hash function $f$) with $\ge n^{1+\eps} / p$ numbers in $A$ is $O(n^{1-\eps})$. For these numbers, we can detect whether they are in any 3SUM solution in $O(n^{2-\eps})$ total time, using the naive algorithm for 3SUM. Afterwards, we can remove these numbers from $A$, so that all hash bucket (under hash function $f$) contains at most $n^{1+\eps}/p$ numbers. 

Similarly, after $\tO(n^{2-\eps})$ time pre-processing we can ensure that any number in $A$ is mapped to the same bucket with $\le n^{1+\eps}/p$ other numbers under either $g$ or $h$.

We will construct $n^{o(1)}$ graphs, one for every $(\delta_{f, 1}, \delta_{f, 2}, \delta_{f, 3}, \delta_{g, 1}, \delta_{g, 2}, \delta_{g, 3}, \delta_{h, 1}, \delta_{h, 2}, \delta_{h, 3}) \in \Delta_f^3 \times \Delta_g^3 \times \Delta_h^3$. The constructed graph will be a tripartite graph, where the nodes of the three parts are labeled by 
\[
\{0 \} \times \F_p \times \F_p, \quad \F_p \times \{0\} \times \F_p, \quad \F_p \times \F_p \times \{0\}
\]
respectively. Next we describe the edges of the graph. We add an edge from a vertex $(x_1, y_1, z_1) \in \{0 \} \times \F_p \times \F_p$ to $(x_2, y_2, z_2) \in \F_p \times \{0\} \times \F_p$ if there is some $a \in A$ such that $x_2 -x_1 = f(a) + \delta_{f, 1}$, $y_2 -y_1 = g(a) + \delta_{g, 1}$ and $z_2 - z_1 = h(a) + \delta_{h, 1}$. Edges between other parts are similarly added, except that we replace $(\delta_{f, 1}, \delta_{g, 1}, \delta_{h, 1})$ with $(\delta_{f, 2}, \delta_{g, 2}, \delta_{h, 2})$ or $(\delta_{f, 3}, \delta_{g, 3}, \delta_{h, 3})$ depending on which part it is. More precisely:
\begin{itemize}
    \item For any $a \in A$, we add an edge with label $a$ (the label is purely for convenience in the analysis, and does not technically appear in the Triangle Listing instances) between $(0, y_1, z_1) \in \{0 \} \times \F_p \times \F_p$ and $(x_2, 0, z_2) \in \F_p \times \{0\} \times \F_p$ if $x_2 = f(a) + \delta_{f, 1}$, $-y_1 = g(a) + \delta_{g, 1}$ and $z_2 - z_1 = h(a) + \delta_{h, 1}$. 
    \item For any $a \in A$, we add an edge  with label $a$ between $(x_2, 0, z_2) \in \F_p \times \{0\} \times \F_p$ and $(x_3, y_3, 0) \in \F_p \times \F_q \times \{0\}$ if $x_3 - x_2 = f(a) + \delta_{f, 2}$, $y_3 = g(a) + \delta_{g, 2}$ and $- z_2 = h(a) + \delta_{h, 2}$. 
    \item For any $a \in A$, we add an edge  with label $a$ between $(x_3, y_3, 0) \in \F_p \times \F_p \times \{0\}$ and $(0, y_1, z_1) \in \{0\} \times \F_p \times \F_p$ if $- x_3 = f(a) + \delta_{f, 3}$, $y_1 - y_3 = g(a) + \delta_{g, 3}$ and $z_1 = h(a) + \delta_{h, 3}$. 
\end{itemize}
We only keep one copy of all duplicated edges. 

The construction of the graph is efficient. For instance, to find all edges between the first two parts, we enumerate $a \in A$ and $z \in \F_p$, and then add edges between $(0, -g(a) - \delta_{g, 1}, z) \in \{0\} \times \F_p \times \F_p$ and $(f(a) + \delta_{f, 1}, 0, z + h(a) + \delta_{h, 1}) \in  \F_p \times \{0\} \times \F_p$. The running time is essentially linear in the number of edges. 

\paragraph{Number of Vertices. } Clearly, the number of vertices in the constructed graph is $O(p^2)$. 

\paragraph{Maximum Degree. } Consider a vertex $(0, y, z) \in \{0\} \times \F_p \times \F_p$ in the first node part. In order for it to connect to some node in the second node part, there must exist some $a \in A$ where $-y = g(a) + \delta_{g, 1}$. By the pre-processing, the number of such $a$ is at most $n^{1+\eps} / p$. Similarly, in order for $(0, y, z)$ to connect to some node in the third node part, we need to have some $a \in A$ where $z = h(a) + \delta_{h, 3}$, so the number of such $a$ again is at most $n^{1+\eps} / p$. Thus, the degree of  $(0, y, z)$ is $O(n^{1+\eps} / p)$. By symmetry, the maximum degree of the graph is $O(n^{1+\eps} / p)$. 

\paragraph{Number of Triangles. } 
First of all, if $a, b, c \in A$ have $a+b+c = 0$, then $f(a) + f(b) + f(c) \in f(a+b+c)-2\Delta_f \subseteq -3\Delta_f$ (note that as $f(0 + 0) - f(0) - f(0) \in \Delta_f$, we must have $f(0) \in -\Delta_f$). Thus, there must exist $\delta_{f, 1}, \delta_{f, 2}, \delta_{f, 3} \in \Delta_f$, so that $(f(a) + \delta_{f, 1}) + (f(b) + \delta_{f, 2}) + (f(c) + \delta_{f, 3}) = 0$. Similarly, must exist $\delta_{g, 1}, \delta_{g, 2}, \delta_{g, 3} \in \Delta_g$, so that $(g(a) + \delta_{g, 1}) + (g(b) + \delta_{g, 2}) + (g(c) + \delta_{g, 3}) = 0$, and $\delta_{h, 1}, \delta_{h, 2}, \delta_{h, 3} \in \Delta_h$, so that $(h(a) + \delta_{h, 1}) + (h(b) + \delta_{h, 2}) + (h(c) + \delta_{h, 3}) = 0$. Then the following three vertices must form a triangle in the graph constructed with parameters $(\delta_{f, 1}, \delta_{f, 2}, \delta_{f, 3}, \delta_{g, 1}, \delta_{g, 2}, \delta_{g, 3}, \delta_{h, 1}, \delta_{h, 2}, \delta_{h, 3})$:
\[
(0, -(g(a) + \delta_{g, 1}), h(c) + \delta_{h, 3}), \quad (f(a) + \delta_{f, 1}, 0, -(h(b) + \delta_{h, 2})), \quad (-(f(c) + \delta_{f, 3}), g(b) + \delta_{g, 2}, 0).
\]

For any triangle in the constructed graph, it is not difficult to verify that the edge labels $a, b, c$ of the three edges of the triangle satisfy 
\begin{equation}
\label{eq:triangle-bound-random}
\begin{gathered}
(f(a) + \delta_{f, 1}) + (f(b) + \delta_{f, 2}) + (f(c) + \delta_{f, 3}) = 0,\\
(g(a) + \delta_{g, 1}) + (g(b) + \delta_{g, 2}) + (g(c) + \delta_{g, 3}) = 0,\\
(h(a) + \delta_{h, 1}) + (h(b) + \delta_{h, 2}) + (h(c) + \delta_{h, 3}) = 0.
\end{gathered}
\end{equation}
Thus, triples $(a, b, c)$ satisfying \cref{eq:triangle-bound-random} are in one-to-one correspondence with triangles in the graph. 

\cref{eq:triangle-bound-random} implies that 
\[
f(a + b + c) \in 2\Delta_f - 3\Delta_f, \quad g(a + b + c) \in 2\Delta_g - 3\Delta_g, \quad h(a + b + c) \in 2\Delta_h - 3\Delta_h. 
\]
If $a + b + c \ne 0$, then by \cref{lem:deltahalmostkwiseindep}, each of these equations holds with probability $n^{o(1)} / p$, so the probability that they all hold is $n^{o(1)} / p^3$. Furthermore, each triple $(a, b, c)$ can contribute at most $1$ triangle in each constructed graph. Thus, the expected number of triangles in each graph with edge labels $a, b, c$ for $a + b + c \ne 0$ is $n^{3+o(1)} / p^3$, and it suffices to list more triangles than this number before we find an actual 3SUM solution if an actual 3SUM solution exists.

\paragraph{Number of Four-Cycles. } By symmetry, it suffices to consider the following two types of four-cycles:
\begin{itemize}
    \item The $4$-cycle (with edge labels) is  $(0, y_1, z_1) \overset{a_1}{-} (x_2, 0, z_2) \overset{a_2}{-} (x_3, y_3, 0) \overset{a_3}{-} (x_4, 0, z_4) \overset{a_4}{-} (0, y_1, z_1)$. Such a $4$-cycle exists only if the following set of equations hold:
    \begin{alignat*}{8}
        & x_2 \in f(a_1) + \Delta_f, \quad && x_3 - x_2 \in f(a_2) + \Delta_f, \quad && x_3 - x_4 \in f(a_3) + \Delta_f, \quad && x_4 \in f(a_4) + \Delta_f, \\
        & -y_1 \in g(a_1) + \Delta_g,\quad && y_3 \in g(a_2) + \Delta_g,\quad && y_3 \in g(a_3) + \Delta_g,\quad && -y_1 \in g(a_4) + \Delta_g, \\
        & z_2 - z_1 \in h(a_1) + \Delta_h, \quad && -z_2 \in h(a_2) + \Delta_h, \quad && -z_4 \in h(a_3) + \Delta_h, \quad && z_4 - z_1 \in h(a_4) +\Delta_h.
    \end{alignat*}
    These equations imply 
    \begin{gather*}
        f(a_1) + f(a_2) - f(a_3) - f(a_4) \in 2\Delta_f - 2\Delta_f,\\
        g(a_1) - g(a_4) \in \Delta_g - \Delta_g, \quad g(a_2) - g(a_3) \in \Delta_g - \Delta_g\\
         h(a_1) + h(a_2) - h(a_3) - h(a_4) \in 2\Delta_h - 2\Delta_h,
    \end{gather*}
    which further imply 
    \begin{equation}
    \label{eq:C4-bound-randomized-1}
    \begin{gathered}
        f(a_1 + a_2 - a_3 - a_4) \in 3\Delta_f - 4\Delta_f,\\
        g(a_1 - a_4) \in \Delta_g - 2\Delta_g, \quad g(a_2- a_3) \in \Delta_g - 2\Delta_g\\
         h(a_1 + a_2 - a_3 - a_4)  \in 3\Delta_h - 4\Delta_h. 
    \end{gathered}
    \end{equation}
    Furthermore, each tuple $(a_1, a_2, a_3, a_4)$ satisfying \cref{eq:C4-bound-randomized-1} can contribute at most $n^{o(1)}$ such $4$-cycles. Notice that if $a_2 = a_3$ or $a_1 = a_4$, then the above walk from $(0, y_1, z_1)$ back to $(0, y_1, z_1)$ actually uses duplicated edges, so it is not a $4$-cycle. Thus, we can assume $a_2 \ne a_3$ and $a_1 \ne a_4$. We consider the following types of tuples $(a_1, a_2, a_3, a_4)$ satisfying \cref{eq:C4-bound-randomized-1}:
    \begin{itemize}
        \item $a_1 + a_2 - a_3 - a_4 = 0$. In this case, the probability that $g(a_1 - a_4) \in \Delta_g - 2\Delta_g$ is $n^{o(1)}/p$, so the expected number of such tuples satisfying \cref{eq:C4-bound-randomized-1} is bounded by $\sfE(A) \cdot n^{o(1)} / p = n^{2+\delta+o(1)} / p$. 
        \item $a_1 + a_2 - a_3 - a_4 = 0$. Fix any tuple $(a_1, a_2, a_3, a_4)$ with $a_1 + a_2 - a_3 - a_4 \ne 0$. By \cref{lem:deltahalmostkwiseindep}, the probability that $f(a_1 + a_2 - a_3 - a_4) \in 3\Delta_f - 4\Delta_f$ is $n^{o(1)}/p$, and the probability that $h(a_1 + a_2 - a_3 - a_4)  \in 3\Delta_h - 4\Delta_h$ is $n^{o(1)} / p$. So the probability that the tuple satisfies these two conditions is $n^{o(1)}/p^2$. We then consider two subcases:
        \begin{itemize}
            \item $a_1 - a_4, a_2 - a_3, 0$ have a $3$-term $n^{o(1)}$-relation. In this case, given $a_1, a_2, a_3$, there can only be $n^{o(1)}$ choices for $a_4$, so the number of such tuples is $n^{3+o(1)}$. For each such tuple, the probability that $g(a_1 - a_4) \in \Delta_g - 2\Delta_g$ is $n^{o(1)} / p$. Combined with the probability that $f(a_1 + a_2 - a_3 - a_4) \in 3\Delta_f - 4\Delta_f$  and $h(a_1 + a_2 - a_3 - a_4)  \in 3\Delta_h - 4\Delta_h$ both hold, each such tuple satisfies \cref{eq:C4-bound-randomized-1} with probability $n^{o(1)} / p^3$, so the expected number of such tuples satisfying \cref{eq:C4-bound-randomized-1} is $n^{3+o(1)} / p^3$. 
            \item $a_1 - a_4, a_2 - a_3, 0$ have no $3$-term $n^{o(1)}$-relation. In this case, the probability that $g(a_1 - a_4) \in \Delta_g - 2\Delta_g, g(a_2- a_3) \in \Delta_g - 2\Delta_g$ both hold is $n^{o(1)} / p^2$ by \cref{lem:deltahalmostkwiseindep}. Thus, the expected number of such tuples  satisfying \cref{eq:C4-bound-randomized-1} is $n^{4+o(1)} / p^4$. 
        \end{itemize}
    \end{itemize}
         \item The $4$-cycle (with edge labels) is  $(0, y_1, z_1) \overset{a_1}{-} (x_2, 0, z_2) \overset{a_2}{-} (0, y_3, z_3) \overset{a_3}{-} (x_4, 0, z_4) \overset{a_4}{-} (0, y_1, z_1)$. Such a $4$-cycle exists only if the following set of equations hold:
    \begin{alignat*}{8}
        & x_2 \in f(a_1) + \Delta_f, \quad && x_2 \in f(a_2) + \Delta_f, \quad && x_4 \in f(a_3) +\Delta_f, \quad && x_4 \in f(a_4) + \Delta_f, \\
        & -y_1 \in g(a_1) + \Delta_g,\quad && -y_3 \in g(a_2) + \Delta_g,\quad && -y_3 \in g(a_3) + \Delta_g,\quad && -y_1 \in g(a_4) + \Delta_g, \\
        & z_2 - z_1 \in h(a_1) + \Delta_h, \quad && z_2 - z_3 \in h(a_2) + \Delta_h, \quad && z_4 - z_3 \in h(a_3) + \Delta_h, \quad && z_4 - z_1 \in h(a_4) + \Delta_h.
    \end{alignat*}
        These equations imply
    \begin{gather*}
    f(a_1) - f(a_2) \in \Delta_f - \Delta_f, \quad f(a_4) - f(a_3) \in \Delta_f - \Delta_f,\\
    g(a_1) - g(a_4) \in \Delta_g - \Delta_g, \quad g(a_2) - g(a_3) \in \Delta_g - \Delta_g \\
    h(a_1) + h(a_3) - h(a_2) - h(a_4) \in 2\Delta_h - 2 \Delta_h, 
    \end{gather*}
    which further imply 
    \begin{equation}
    \label{eq:C4-bound-randomized-2}
    \begin{gathered}
    f(a_1 - a_2) \in \Delta_f - 2\Delta_f, \quad f(a_4 - a_3) \in \Delta_f - 2 \Delta_f,\\
    g(a_1 - a_4) \in \Delta_g - 2 \Delta_g, \quad g(a_2 - a_3) \in \Delta_g - 2 \Delta_g \\
    h(a_1 + a_3 - a_2 - a_4) \in 3\Delta_h - 4 \Delta_h. 
    \end{gathered}
    \end{equation}
    Additionally, for every tuple $(a_1, a_2, a_3, a_4) \in A^4$ satisfying \cref{eq:C4-bound-randomized-2}, there can be at most $n^{o(1)} p$ $4$-cycles of this type in the graph with labels $a_1, a_2, a_3, a_4$ (the $p$ factor comes from the $p$ choices of $z_1$). Similar as before, we can assume $a_1 \ne a_2, a_2 \ne a_3, a_3 \ne a_4, a_4 \ne a_1$, as otherwise the $4$-cycle degenerates to a $4$-walk. We consider the following types of tuples satisfying \cref{eq:C4-bound-randomized-2} (and the number of corresponding $4$-cycles is a factor of $n^{o(1)} \cdot p$ larger):
    \begin{itemize}
        \item $a_1 - a_4 = a_2 - a_3$. The number of such tuples is $\sfE(A)$. For each such tuple, the probability that $f(a_1 - a_2) \in \Delta_f - 2\Delta_f$ is $n^{o(1)} / p$, and the probability that $g(a_1 - a_4) \in \Delta_g - 2 \Delta_g$ is $n^{o(1)} / p$. These two events are independent, so the probability that each tuple satisfies \cref{eq:C4-bound-randomized-2} is $n^{o(1)} / p^2$, and thus the expected number of  $(a_1, a_2, a_3, a_4)$ with $a_1 - a_4 = a_2 - a_3$ satisfying \cref{eq:C4-bound-randomized-2} is $\sfE(A) \cdot n^{o(1)} / p^2 = n^{2+\delta+o(1)} / p^2$. 
        \item $a_1 - a_4 \ne a_2 - a_3$. In this case, we consider four subcases:
        \begin{itemize}
        \item $a_1 - a_2, a_4 - a_3, 0$ have a $3$-term $n^{o(1)}$-relation and $a_1 - a_4, a_2 - a_3, 0$ have a $3$-term $n^{o(1)}$-relation. In this case, there are rational numbers $C$ and $D$ whose denominator and numerator are both $n^{o(1)}$ such that $a_1 - a_2 = C(a_4 - a_3)$ and $a_1 - a_4 = D(a_2 - a_3)$. Rearranging the equations, we get that $a_1 - a_2 = C(a_4 - a_3)$ and $a_1 - D a_2 = a_4 - D a_3$. We know $D \ne 1$ as $a_1 - a_4 \ne a_2 - a_3$, so we can uniquely solve $a_1, a_2$ given $a_3, a_4, C, D$. Therefore, the number of $(a_1, a_2, a_3, a_4) \in A^4$ falling into this case is $n^{2+o(1)}$. For each such $(a_1, a_2, a_3, a_4)$, the probability they satisfy \cref{eq:C4-bound-randomized-2} is $n^{o(1)} / p^3$, so the expected number of tuples satisfying \cref{eq:C4-bound-randomized-2} this case is $n^{2+o(1)} / p^3$. 
        \item $a_1 - a_2, a_4 - a_3, 0$ have a $3$-term $n^{o(1)}$-relation and $a_1 - a_4, a_2 - a_3, 0$ have no $3$-term $n^{o(1)}$-relation. Here, the number of such tuples is $n^{3+o(1)}$, as given $a_1, a_2, a_3$ and the coefficients for the $3$-term $n^{o(1)}$-relation, we can uniquely determine $a_4$. For each such tuple, the probability that $g(a_1 - a_4) \in \Delta_g - 2 \Delta_g$ and  $g(a_2 - a_3) \in \Delta_g - 2 \Delta_g$ is $n^{o(1)} / p^2$, as $a_1 - a_4, a_2 - a_3, 0$ have no $3$-term $n^{o(1)}$-relation. The probability that $f(a_1 - a_2) \in \Delta_f - 2\Delta_f$ and the probability that $h(a_1 + a_3 - a_2 - a_4) \in 3\Delta_h - 4 \Delta_h$ are both $n^{o(1)} / p$. The previous three events are independent, so the expected number of tuples satisfying \cref{eq:C4-bound-randomized-2} this case is $n^{3+o(1)} / p^4$. 
        \item $a_1 - a_2, a_4 - a_3, 0$ have no $3$-term $n^{o(1)}$-relation and $a_1 - a_4, a_2 - a_3, 0$ have a $3$-term $n^{o(1)}$-relation. This case is symmetric to the previous case, and the expected number is $n^{3+o(1)} / p^4$. 
        \item Both $a_1 - a_2, a_4 - a_3, 0$ and $a_1 - a_4, a_2 - a_3, 0$ have no $3$-term $n^{o(1)}$-relations. The number of such tuples is $O(n^4)$, and each tuple satisfies \cref{eq:C4-bound-randomized-2} with probability $n^{o(1)} / p^5$ by \cref{lem:deltahalmostkwiseindep}. Thus, the expected number of tuples satisfying \cref{eq:C4-bound-randomized-2} this case is $n^{4+o(1)} / p^5$. 
        \end{itemize}
    \end{itemize}

\end{itemize}
Overall, the expected number of $4$-cycles is 
\[
n^{2+\delta+o(1)} / p + n^{4+o(1)} / p^4 + n^{2+o(1)} / p^2 + n^{3+o(1)} / p^4 + n^{4+o(1)} / p^4 = n^{2+\delta +o(1)} / p + n^{4+o(1)} / p^4. 
\]
\subsection{Derandomization}

The randomness of the above construction comes from the random choices of the hash functions $f, g, h$. By \cref{def:hashfamilyh}, the size of the hash family is $p^{O(1)}$, so it will be expensive to enumerate over all triples of hash functions when $p$ is big. Similar to \cite{fischer3sum}, the idea for derandomization is to pick hash functions from smaller families (say when $p = n^\eps$, the hash family will have size $n^{O(\eps)}$), and taking their direct product to get the final hash function.

We need the following lemma.

\begin{lemma}
\label{lem:tri-listing-small-prime}
Let $A \subset [n^{O(1)}]$ be a size-$n$ set without 3SUM solutions, and $p = \Theta(n^\eps)$ be a prime. 
Let $F, G, H: \{-N, \ldots, N\} \rightarrow \F_p^k$ for some constant $k$ be hash functions, where for any $x, y \in \{-N, \ldots, N\}$, $F(x + y) - F(x) - F(y) \in \Delta_F, G(x + y) - G(x) - G(y) \in \Delta_G, H(x + y) - H(x) - H(y) \in \Delta_H$ for $|\Delta_F|, |\Delta_G|, |\Delta_H| = n^{o(1)}$. Additionally, evaluating $F, G, H$ takes $n^{o(1)}$ time. For any $f, g, h \in \mathcal{H}$ from the family in \cref{def:hashfamilyh} with parameter $p$, define $F \times f: \{-N, \ldots, N\} \rightarrow \F_p^{k+1}$ by $(F \times f)(x) = (F(x), f(x))$, and define $\Delta_{F \times f} = \Delta_F \times \Delta_f$, and similarly define $G \times g, \Delta_{G \times g}$ and $H \times h, \Delta_{H \times h}$. Then there exist $f, g, h \in \mathcal{H}$, so that the following properties hold:
\begin{enumerate}[itemsep=\medskipamount]
    \item \label{item:lem:tri-listing-small-prime:item1}
    It holds that:
    \[
    \sum_{a \ne b \in A} [(F \times f)(a) = (F \times f)(b)] \le O(1/p) \cdot \sum_{a \ne b \in A} [F(a) = F(b)].  
    \]
    The same holds when replacing $(F, f)$ with $(G, g)$ or $(H, h)$. 
    \item \label{item:lem:tri-listing-small-prime:item2}
    It holds that:
    \[
    \sum_{\substack{a, b, c \in A}} [(F \times f)(a)+(F \times f)(b)+(F \times f)(c) \in - 3\Delta_{F \times f}] \le \frac{n^{o(1)}}{p} \cdot  \sum_{\substack{a, b, c \in A}} [F(a) + F(b) +  F(c) \in - 3\Delta_{F}]. 
    \]
    The same holds when replacing $(F, f)$ with $(G, g)$ or $(H, h)$.
    \item \label{item:lem:tri-listing-small-prime:item3} 
    Let $\mathcal{C}_{f_0, g_0, h_0}(a_1, a_2, a_3, a_4)$ be the indicator variable for the conditions
    \begin{equation*}
    \begin{gathered}
        f_0(a_1) + f_0(a_2) - f_0(a_3) - f_0(a_4) \in 2\Delta_{f_0} - 2\Delta_{f_0},\\
        g_0(a_1) - g_0(a_4) \in \Delta_{g_0} - \Delta_{g_0}, \quad g_0(a_2) - g_0(a_3) \in \Delta_{g_0} - \Delta_{g_0}\\
         h_0(a_1) + h_0(a_2) - h_0(a_3) - h_0(a_4) \in 2\Delta_{h_0} - 2\Delta_{h_0}. 
    \end{gathered}
    \end{equation*}
    Then 
    \[
    \sum_{\substack{a_1, a_2, a_3, a_4 \in A^4 \\ a_2 \ne a_3, a_1 \ne a_4 \\ a_1 - a_4 = a_3 - a_2}}  \mathcal{C}_{F \times f, G \times g, H \times h}(a_1, a_2, a_3, a_4) \le \frac{n^{o(1)}}{p^2} \sum_{\substack{a_1, a_2, a_3, a_4 \in A^4 \\ a_2 \ne a_3, a_1 \ne a_4 \\ a_1 - a_4 = a_3 - a_2}}  \mathcal{C}_{F, G, H}(a_1, a_2, a_3, a_4).
    \]
    The same holds for any permutations of $(F, f), (G, g), (H, h)$. 
    \item \label{item:lem:tri-listing-small-prime:item4}
    It holds that:
    \begin{align*}
    & \sum_{\substack{a_1, a_2, a_3, a_4 \in A^4 \\ a_2 \ne a_3, a_1 \ne a_4 \\ a_1 - a_4 \ne a_3 - a_2 \\ 
    a_1 - a_4, a_2 - a_3, 0 \text{ have } n^{o(1)}\text{-relation}}} \mathcal{C}_{F \times f, G \times g, H \times h}(a_1, a_2, a_3, a_4) \\
     \le \frac{n^{o(1)}}{p^3} \cdot {} & \sum_{\substack{a_1, a_2, a_3, a_4 \in A^4 \\ a_2 \ne a_3, a_1 \ne a_4 \\ a_1 - a_4 \ne a_3 - a_2 \\ 
    a_1 - a_4, a_2 - a_3, 0 \text{ have } n^{o(1)}\text{-relation}}} \mathcal{C}_{F, G, H}(a_1, a_2, a_3, a_4).
    \end{align*}
    The same holds for any permutation of $(F, f), (G, g), (H, h)$. 
    \item \label{item:lem:tri-listing-small-prime:item5}
    It holds that:
    \begin{align*}
    & \sum_{\substack{a_1, a_2, a_3, a_4 \in A^4 \\ a_2 \ne a_3, a_1 \ne a_4 \\ a_1 - a_4 \ne a_3 - a_2 \\ 
    a_1 - a_4, a_2 - a_3, 0 \text{ do not have } n^{o(1)}\text{-relation}}} \mathcal{C}_{F \times f, G \times g, H \times h}(a_1, a_2, a_3, a_4)\\
     \le \frac{n^{o(1)}}{p^4} \cdot {} & \sum_{\substack{a_1, a_2, a_3, a_4 \in A^4 \\ a_2 \ne a_3, a_1 \ne a_4 \\ a_1 - a_4 \ne a_3 - a_2 \\ 
    a_1 - a_4, a_2 - a_3, 0 \text{ do not have } n^{o(1)}\text{-relation}}} \mathcal{C}_{F, G, H}(a_1, a_2, a_3, a_4).
    \end{align*}
    The same holds for any permutation of $(F, f), (G, g), (H, h)$. 
    \item \label{item:lem:tri-listing-small-prime:item6} 
    Let $\mathcal{D}_{f_0, g_0, h_0}(a_1, a_2, a_3, a_4)$ be an indicator variable for the conditions
    \begin{equation*}
    \begin{gathered}
            f_0(a_1) - f_0(a_2) \in \Delta_{f_0} - \Delta_{f_0}, \quad f_0(a_4) - f_0(a_3) \in \Delta_{f_0} - \Delta_{f_0},\\
    g_0(a_1) - g_0(a_4) \in \Delta_{g_0} - \Delta_{g_0}, \quad g_0(a_2) - g_0(a_3) \in \Delta_{g_0} - \Delta_{g_0} \\
    h_0(a_1) + h_0(a_3) - h_0(a_2) - h_0(a_4) \in 2\Delta_{h_0} - 2 \Delta_{h_0}. 
    \end{gathered}
    \end{equation*}
    Then 
    \begin{align*}
    & \sum_{\substack{a_1, a_2, a_3, a_4 \in A^4 \\ a_1 \ne a_2, a_2 \ne a_3, a_3 \ne a_4, a_4 \ne a_1 \\ a_1 - a_4 = a_2 - a_3}}  \mathcal{D}_{F \times f, G \times g, H \times h}(a_1, a_2, a_3, a_4) \\
    \le \frac{n^{o(1)}}{p^2} \cdot {} & \sum_{\substack{a_1, a_2, a_3, a_4 \in A^4 \\ a_1 \ne a_2, a_2 \ne a_3, a_3 \ne a_4, a_4 \ne a_1 \\ a_1 - a_4 = a_2 - a_3}}  \mathcal{D}_{F, G, H}(a_1, a_2, a_3, a_4).
    \end{align*}
    The same holds for any permutation of $(F, f), (G, g), (H, h)$. 
    \item \label{item:lem:tri-listing-small-prime:item7}
    It holds that:
    \begin{align*}
    &\sum_{\substack{a_1, a_2, a_3, a_4 \in A^4 \\ a_1 \ne a_2, a_2 \ne a_3, a_3 \ne a_4, a_4 \ne a_1 \\ a_1 - a_4 \ne a_2 - a_3 \\ \text{both } (a_1 - a_2, a_4 - a_3, 0) \text{ and } (a_1 - a_4, a_2 - a_3, 0) \\ \text{have 3-term } n^{o(1)}\text{-relations}}}  \mathcal{D}_{F \times f, G \times g, H \times h}(a_1, a_2, a_3, a_4) \\ \le \frac{n^{o(1)}}{p^3} \cdot {} &\sum_{\substack{a_1, a_2, a_3, a_4 \in A^4 \\ a_1 \ne a_2, a_2 \ne a_3, a_3 \ne a_4, a_4 \ne a_1 \\ a_1 - a_4 \ne a_2 - a_3 \\ \text{both } (a_1 - a_2, a_4 - a_3, 0) \text{ and } (a_1 - a_4, a_2 - a_3, 0) \\ \text{have 3-term } n^{o(1)}\text{-relations}}} \mathcal{D}_{F, G, H}(a_1, a_2, a_3, a_4).
    \end{align*}
    The same holds for any permutation of $(F, f), (G, g), (H, h)$. 
    \item \label{item:lem:tri-listing-small-prime:item8}
    It holds that:
    \begin{align*}
    &\sum_{\substack{a_1, a_2, a_3, a_4 \in A^4 \\ a_1 \ne a_2, a_2 \ne a_3, a_3 \ne a_4, a_4 \ne a_1 \\ a_1 - a_4 \ne a_2 - a_3 \\ \text{exactly one of } (a_1 - a_2, a_4 - a_3, 0) \text{ and } (a_1 - a_4, a_2 - a_3, 0) \\ \text{has 3-term } n^{o(1)}\text{-relation}}}  \mathcal{D}_{F \times f, G \times g, H \times h}(a_1, a_2, a_3, a_4) \\ \le \frac{n^{o(1)}}{p^4} \cdot {} &\sum_{\substack{a_1, a_2, a_3, a_4 \in A^4 \\ a_1 \ne a_2, a_2 \ne a_3, a_3 \ne a_4, a_4 \ne a_1 \\ a_1 - a_4 \ne a_2 - a_3 \\ \text{exactly one of } (a_1 - a_2, a_4 - a_3, 0) \text{ and } (a_1 - a_4, a_2 - a_3, 0) \\ \text{has 3-term } n^{o(1)}\text{-relation}}} \mathcal{D}_{F, G, H}(a_1, a_2, a_3, a_4).
    \end{align*}
    The same holds for any permutation of $(F, f), (G, g), (H, h)$. 
    \item \label{item:lem:tri-listing-small-prime:item9}
    It holds that:
    \begin{align*}
    &\sum_{\substack{a_1, a_2, a_3, a_4 \in A^4 \\ a_1 \ne a_2, a_2 \ne a_3, a_3 \ne a_4, a_4 \ne a_1 \\ a_1 - a_4 \ne a_2 - a_3 \\ \text{neither } (a_1 - a_2, a_4 - a_3, 0) \text{ nor } (a_1 - a_4, a_2 - a_3, 0) \\ \text{has 3-term } n^{o(1)}\text{-relation}}}  \mathcal{D}_{F \times f, G \times g, H \times h}(a_1, a_2, a_3, a_4) \\ \le \frac{n^{o(1)}}{p^5} \cdot {} &\sum_{\substack{a_1, a_2, a_3, a_4 \in A^4 \\ a_1 \ne a_2, a_2 \ne a_3, a_3 \ne a_4, a_4 \ne a_1 \\ a_1 - a_4 \ne a_2 - a_3 \\ \text{neither } (a_1 - a_2, a_4 - a_3, 0) \text{ nor } (a_1 - a_4, a_2 - a_3, 0) \\ \text{has 3-term } n^{o(1)}\text{-relation}}} \mathcal{D}_{F, G, H}(a_1, a_2, a_3, a_4).
    \end{align*}
    The same holds for any permutation of $(F, f), (G, g), (H, h)$. 
\end{enumerate}
Additionally, we can find such $f, g, h$ in time
\[
n^{O(\eps)} \cdot \left(n + p^k + \sum_{W \in \{F, G, H\}}\sum_{a \ne b \in A} [W(a) - W(b) \in \Delta_W - \Delta_W] \right). 
\]
\end{lemma}
\begin{proof}[Proof Sketch]
For random $f, g, h \in \mathcal{H}$, one can show that the expectation of the left hand side is upper bounded by the right hand side, for each of the inequalities, by  \cref{lem:deltahalmostkwiseindep} (similar to \cref{sec:tri-listing-randomized}). Therefore, by union bound, there exist $f, g, h$ where the left hand side is upper bounded by $100$ times its expectation, for all the inequalities.

The majority of the proof is dedicated to computing the right hand sides given $F, G, H$ (and also similarly the left hand sides given $F, G, H, f, g, h$) deterministically. Then it suffices to try all possible choices of $f, g, h$ and choose one that satisfy the inequalities. 

\begin{itemize}
    \item Computing \cref{item:lem:tri-listing-small-prime:item1}. Clearly, we can compute the right hand side in $n^{1+o(1)}$ time. 
    \item Computing \cref{item:lem:tri-listing-small-prime:item2}. Here, we can embed $\F_p^k$ to $\Z$ (i.e., view all elements in $\F_p^k$ as a number in base $3p$), and then use FFT to compute the right hand side in $\tO((3p)^k) = \tO(p^k)$ time. 
    \item Computing \cref{item:lem:tri-listing-small-prime:item3,item:lem:tri-listing-small-prime:item4,item:lem:tri-listing-small-prime:item5}. For every pair of integers $C, D$ bounded by $n^{o(1)}$ in absolute values, we aim to compute 
    \[\sum_{\substack{a_1, a_2, a_3, a_4 \in A^4 \\ a_2 \ne a_3, a_1 \ne a_4 \\ 
    C(a_1 - a_4) + D(a_2 - a_3) = 0}} \mathcal{C}_{F, G, H}(a_1, a_2, a_3, a_4).\]
    Clearly, the above values are sufficient for computing \cref{item:lem:tri-listing-small-prime:item3,item:lem:tri-listing-small-prime:item4,item:lem:tri-listing-small-prime:item5}. 

    Recall $\mathcal{C}_{F, G, H}(a_1, a_2, a_3, a_4)$ is the indicator variable for the following conditions:
    \begin{equation*}
    \begin{gathered}
        F(a_1) + F(a_2) - F(a_3) - F(a_4) \in 2\Delta_{F} - 2\Delta_{F},\\
        G(a_1) - G(a_4) \in \Delta_{G} - \Delta_{G}, \quad G(a_2) - G(a_3) \in \Delta_{G} - \Delta_{G}\\
         H(a_1) + H(a_2) - H(a_3) - H(a_4) \in 2\Delta_{H} - 2\Delta_{H}. 
    \end{gathered}
    \end{equation*}

    First, we enumerate $a_1 \ne a_4 \in A$ where $G(a_1) - G(a_4) \in \Delta_G - \Delta_G$ (after an $n^{1+o(1)}$ preprocessing, we can perform this enumeration in $O(n^{1+o(1)} + \sum_{a \ne b \in A} [G(a) - G(b) \in \Delta_G - \Delta_G])$) time, and then insert
    \[
    \left(F(a_1) - F(a_4), H(a_1) - H(a_4), C(a_1 - a_4)\right)
    \]
    to a multiset. Then we enumerate $a_2 \ne a_3 \in A$ where $G(a_2) - G(a_3) \in \Delta_G - \Delta_G$, and we add the number of points in the multiset inside 
    \[
    \left(F(a_3) - F(a_2) + 2\Delta_F - 2\Delta_F, H(a_3) - H(a_2) + 2\Delta_H - 2\Delta_H, -D(a_2 - a_3)\right)
    \]
    to the final answer. 
    The running time of this algorithm is 
    \[
    n^{1+o(1)} + n^{o(1)} \cdot \sum_{a \ne b \in A} [G(a) - G(b) \in \Delta_G - \Delta_G]. 
    \]

    \item Computing \cref{item:lem:tri-listing-small-prime:item6,item:lem:tri-listing-small-prime:item7,item:lem:tri-listing-small-prime:item8,item:lem:tri-listing-small-prime:item9}. For every integers $C_1, D_1, C_2, D_2$ bounded by $n^{o(1)}$ in absolute values, we aim to compute 
    \[\sum_{\substack{a_1, a_2, a_3, a_4 \in A^4 \\ a_1 \ne a_2, a_2 \ne a_3, a_3 \ne a_4, a_4 \ne a_1 \\ C_1(a_1 - a_2) = D_1(a_4 - a_3)\\
    C_2 (a_1 - a_4) = D_2 (a_2 - a_3)}} \mathcal{D}_{F, G, H}(a_1, a_2, a_3, a_4).
    \]
    Clearly, we can compute \cref{item:lem:tri-listing-small-prime:item6,item:lem:tri-listing-small-prime:item7,item:lem:tri-listing-small-prime:item8,item:lem:tri-listing-small-prime:item9} once we compute the above values for all $C_1, D_1, C_2, D_2$. Recall $\mathcal{D}_{F, G, H}(a_1, a_2, a_3, a_4)$ is the indicator variable for the following conditions:
    \begin{equation*}
    \begin{gathered}
            F(a_1) - F(a_2) \in \Delta_{F} - \Delta_{F}, \quad F(a_4) - F(a_3) \in \Delta_{F} - \Delta_{F},\\
    G(a_1) - G(a_4) \in \Delta_{G} - \Delta_{G}, \quad G(a_2) - G(a_3) \in \Delta_{G} - \Delta_{G} \\
    H(a_1) + H(a_3) - H(a_2) - H(a_4) \in 2\Delta_{H} - 2 \Delta_{H}. 
    \end{gathered}
    \end{equation*}
    First, we enumerate $a_1 \ne a_4 \in A$ where $G(a_1) - G(a_4) \in \Delta_G - \Delta_G$. Similar to the previous case, this can be done in $O(n^{1+o(1)} + \sum_{a \ne b \in A} [G(a) - G(b) \in \Delta_G - \Delta_G])$ time. Then we add 
    \[
    \left(F(a_1), F(a_4), H(a_1) - H(a_4), a_1, a_4\right)
    \]
    to a multiset. 
    Then we enumerate $a_2 \ne a_3 \in A$ where $G(a_2) - G(a_3) \in \Delta_G - \Delta_G$, query how many points in the multiset are in
    \[
    \left(F(a_2) + \Delta_F - \Delta_F, F(a_3) + \Delta_F - \Delta_F, H(a_2) - H(a_3) + 2\Delta_H - 2\Delta_H, \Z \setminus\{a_2\}, \Z \setminus \{a_3\} \right),
    \]
    and add this count to the final answer. 
    
    Using data structure, each insertion and query can be performed in $n^{o(1)}$ time, so the overall running time is \[
    n^{1+o(1)} + n^{o(1)} \cdot \sum_{a \ne b \in A} [G(a) - G(b) \in \Delta_G - \Delta_G]. 
    \]
\end{itemize}
\end{proof}

Now the idea is to apply \cref{lem:tri-listing-small-prime} $O(1/\eps)$ times to find the final hash functions. However, the running time \cref{lem:tri-listing-small-prime} has a dependency on $\sum_{W \in \{F, G, H\}}\sum_{a \ne b \in A} [W(a) - W(b) \in \Delta_W - \Delta_W])$ which can be as large as $O(n^2)$. Therefore, we need to perform the next lemma several times first. 

\begin{lemma}
\label{lem:tri-listing-small-prime-preprocessing}
Let $A \subset [n^{O(1)}]$ be a size-$n$ set, and $p = \Theta(n^\eps)$ be a prime. 
Let $F: \{-N, \ldots, N\} \rightarrow \F_p^k$ for some constant $k$ be a hash function, where for any $x, y \in \{-N, \ldots, N\}$, $F(x + y) - F(x) - F(y) \in \Delta_F$ and $|\Delta_F| = n^{o(1)}$. Additionally, evaluating $F$ takes $n^{o(1)}$ time. For any $f \in \mathcal{H}$ from the family in \cref{def:hashfamilyh} with parameter $p$, define $F \times f: \{-N, \ldots, N\} \rightarrow \F_p^{k+1}$ by $(F \times f)(x) = (F(x), f(x))$, and define $\Delta_{F \times f} = \Delta_F \times \Delta_f$. Then there exists $f\in \mathcal{H}$, so that
\[
\sum_{a \ne b \in A} [(F \times f)(a) - (F \times f)(b) \in \Delta_{F \times f} - \Delta_{F \times f}] \le \frac{n^{o(1)}}{p} \sum_{a \ne b \in A} [F(a) - F(b) \in \Delta_{F} - \Delta_{F}]. 
\]
Additionally, we can find $f$ in $n^{1+O(\eps)}$ time. 
\end{lemma}
We omit the proof of \ref{lem:tri-listing-small-prime-preprocessing} as it is similar to (and simpler than) the proof of \cref{lem:tri-listing-small-prime}. 

Finally, we are ready to prove \cref{thm:to-trianglelisting}. 

\begin{proof}[Proof of \cref{thm:to-trianglelisting}]
    We first call  \cref{lem:tri-listing-small-prime-preprocessing} $\ell = O(1)$ times iteratively, and reset $F := F \times f$ after every iteration, to obtain $F$, so that 
    \[
    \sum_{a \ne b \in A} [(F \times f)(a) - (F \times f)(b) \in \Delta_{F \times f} - \Delta_{F \times f}] \le n^{2 - C \eps}
    \]
    for a sufficiently large constant $C$. The running time is $n^{1+O(\eps)}$. 

    Then we set $G, H = F$, and then call  \cref{lem:tri-listing-small-prime} $k = \log_n P /\eps$ times iteratively (after each call, set $F := F \times f, G := G \times g, H := H \times h$). The running time will be $n^{1+O(\eps)} + n^{O(\eps) + 2 - C\eps}$. As $C$ is sufficiently large, the running time can be made $n^{1+O(\eps)} + n^{2 - \eps}$. 

    One caveat is that \cref{lem:tri-listing-small-prime} requires that $A$ to not contain 3SUM solutions. However, we can run  \cref{lem:tri-listing-small-prime} on $A$ regardless, and if at any point, the lemma fails, we can immediately conclude that $A$ has a 3SUM solution. Thus, even if $A$ contains 3SUM solutions, we will end up getting desired hash functions $F, G, H$. 

    The remainder of the proof is largely similar to the proof in \cref{sec:tri-listing-randomized}. The key differences are that the nodes of the graph is 
    \[
\{0 \} \times \F_{p}^{\ell+k} \times \F_{p}^{\ell+k}, \quad \F_{p}^{\ell+k} \times \{0\} \times \F_{p}^{\ell+k}, \quad \F_{p}^{\ell+k} \times \F_{p}^{\ell+k} \times \{0\}
\]
where $p^{\ell + k} = P \cdot n^{O(\eps)}$, 
and the hash functions to use if $(F, G, H)$ instead of $(f, g, h)$. The number of vertices becomes $O((P \cdot n^{O(\eps)})^2) = P^2 \cdot n^{O(\eps)}$. By \cref{lem:tri-listing-small-prime} \cref{item:lem:tri-listing-small-prime:item1}, we have 
\[
 \sum_{a \ne b \in A} [F(a) = F(b)] = O(1/p^k) \cdot n^2 = O(n^2 / P).
\]
The number of elements that are in bucket (under $F$) with $\ge n^{1+\eps} / P$ elements is $O(n^{1-\eps})$, so we can test whether any of them is in a 3SUM solution and then remove them in $O(n^{2-\eps})$ time. We can similarly do so for buckets under $G$ and $H$ as well. Thus, then the maximum degree of the graph is bounded by $O(n^{1+\eps} / P)$, similar to \cref{sec:tri-listing-randomized}. The running time for constructing the graphs is essentially the total number of edges of all graphs, which is $P^2 \cdot n^{O(\eps)} \cdot n^{1+\eps} / P = n^{1+O(\eps)} \cdot P$. 

The bound on the number of triangles and the number of $4$-cycles similarly follow from a combination of the bounds in \cref{lem:tri-listing-small-prime} and the analysis in  \cref{sec:tri-listing-randomized}. We omit the details for conciseness. 
\end{proof}

Next we use \cref{thm:to-trianglelisting} to prove \cref{thm:4cyclehardnessmain}, which we  recall below:
\fourcycle*

\begin{proof}
For the sake of contradiction, suppose such an $O(n^{2-\delta} + t)$ time or $O(m^{4/3-\delta} + t)$ time algorithm for $4$-Cycle Listing exists for some $\delta > 0$, and we will use it to design a deterministic truly subquadratic time algorithm for 3SUM, violating the deterministic 3SUM hypothesis. 

First, by \cref{thm:smalldoub3sumhardnessmain}, it suffices to solve 3SUM instances with additive energy $\le n^{2.01}$. Then we run \cref{thm:to-trianglelisting} with $\eps = \delta / C$ for a sufficiently large $C$ with $C/\delta$ being an integer and $k = 1+C/\delta$ (so that $P = n^{1/2 + \delta/C}$) in $O(n^{2-\delta/C} + n^{3/2+O(\delta / C)})$ time to further reduce to $n^{o(1)}$ Triangle Listing instances on graphs with 
    \begin{itemize}
        \item $n^{1+O(\delta / C)}$ nodes;
        \item $n^{0.5}$ maximum degree (thus, $n^{1.5+O(\delta / C)}$ edges); 
        \item $n^{1.51 - \delta / C + o(1)} + n^{2 - 4\delta / C +o(1)}$  $4$-cycles, 
    \end{itemize}
    and for each graph we need to list $n^{2-2\delta/C+o(1)}$ triangles. The running time becomes $O(n^{2-\delta/C})$ for sufficiently large $C$, which is truly subquadratic. 

    For each produced Triangle Listing instance on some graph $G$, we use the following standard  reduction \cite{DBLP:conf/stoc/AbboudBKZ22, AbboudBF23, JinX23} to reduce it to a $4$-Cycle Listing instance. We create four node sets $V_1, V_2, V_3, V_4$, each being a copy of $V(G)$. Between $V_1$ and $V_2$, we add a matching. Between $V_2$ and $V_3$, $V_3$ and $V_4$, and $V_4$ and $V_1$, we add copies of $E(G)$. Observe that $4$-cycles in the new graph that use one node from each of $V_1, V_2, V_3, V_4$ correspond to triangles in $G$, while all other $4$-cycles correspond to $4$-cycles in $G$. Furthermore, each $4$-cycle in $G$ only adds $O(1)$ $4$-cycles in the new graph. Therefore, it suffices to list $n^{1.51 - \delta / C + o(1)} + n^{2 - 4\delta / C +o(1)} + n^{2-2\delta / C + o(1)} = n^{2-2 \delta /C+o(1)}$ $4$-cycles (for sufficiently large $C$) in the new graph. 
    
    The bounds on number of nodes and edges in the new graph are only a constant factor larger than those in $G$. Thus, we can apply the assumed $4$-Cycle Listing algorithm on the new graph to obtain a running time 
    \[
    \left(n^{1+O(\delta / C)}\right)^{2-\delta} + n^{2-2 \delta /C+o(1)}
    \]
    or 
    \[
    \left(n^{1.5+O(\delta / C)}\right)^{4/3-\delta} + n^{2-2 \delta /C+o(1)},
    \]
    either of which achieves a truly subquadratic running time for sufficiently large $C$. 
\end{proof}

\section{\texorpdfstring{$k$}{k}-Mismatch Constellation}

Recall that in the Constellation with mismatches problem,
   we are given integer sets $A,B\subseteq [N]$ with $|A|, |B| \le n$ and a parameter $1\le k \le  |B|$, and the task is to compute the set \[\{c\in \Z:  |(c+B)\setminus A| < k\}.\]
   As in previous works on this question, we impose the condition
$k\le (1-\eta)|B|$ for a fixed constant $\eta\in (0,1)$, so that the output size is at most $\frac{|A||B|}{|B|-k} \le \eta^{-1}|A| = O(n)$. 
Here we let $\eta = 0.7$. The case where $\eta$ is a smaller (but fixed) constant in $(0,0.7)$ (in particular, $k = \Theta(|B|)$) can be solved via different methods: using \cite[Lemma 8.4]{ChanWX23} (randomized) or our \cref{thm:introdetpopsum} (deterministic), we can solve it in $n^{2-\Omega(1)} = k^{1-\Omega(1)}\cdot n$ time.

Our algorithm for Constellation with $k$ mismatches first reduces the original problem to the following partial convolution problem  with promise (\cref{prob:reducedprob}). This reduction is simple but randomized, so the algorithm obtained this way will be randomized as well (we stress that our result is new even allowing randomization).
Later in \cref{sec:constellationderand} we will describe how to derandomize this algorithm.
\begin{prob}
    \label{prob:reducedprob}
   Given sets $A,B,C\subseteq [N]$,  and parameter $k\le (1-\eta)|B|$ for some fixed constant $\eta\in (0,1)$, with the promise that 
   \begin{equation}
       \label{eqn:promisec}
    |(c+B)\setminus A| \le k
   \end{equation}
    for all $c\in C$, the task is to compute $(1_A \conv 1_{-B})[c]$ for all $c\in C$.
\end{prob}
\begin{proposition}[Simple reduction via subsampling]
    \label{prop:simpleconstellationreduction}
   There is a randomized reduction from  Constellation with $k_0$ mismatches (where $k_0 \le (1-\eta_0)|B|$) to \cref{prob:reducedprob} with the same input sets $A,B$, and parameter $k=(1+\eta_0/20)k_0 \le (1-0.95\eta_0)|B|$, and $|C| \le 2\eta_0^{-1}|A|$.  The reduction  takes time $\tilde O(|A|+|B|)$.
\end{proposition}
\begin{proof}
    Denote $n:= |A|+|B|$.
First, assume $k_0\ge  100\eta_0^{-2}\log n$, since otherwise the algorithm of \cite{CardozeS98,nickconstellation} already solves Constellation with $k_0$ mismatches in $\tilde O(nk_0) \le \tilde O(n)$ time.

Subsample $B'\subseteq B$ by keeping each $b\in B$ independently with probability $p=\frac{100\log n}{\eta_0^2 k_0}$.
By Chernoff bound and $|B|>k_0 \ge 100\eta_0^{-2}\log n$, we have  $|B'| \in (1\pm \eta_0/20 ) p|B|$ with $1-\exp(-\Omega(\eta_0^2 p|B|) ) = 1-1/\poly(n)$ probability.
Let $k' = (1+\eta_0/40)pk_0$.
Then, for every $c\in \Z$, the following holds with $1-1/\poly(n)$ probability by Chernoff bound:
\begin{enumerate}
    \item If $|(c+B)\setminus A| \le k_0$, then $|(c+B')\setminus A| \le k'$.
        \label{item1}
        \item 
If $|(c+B)\setminus A| > (1+\eta_0/20) k_0$, then $|(c+B')\setminus A| > k'$.
\label{item2}
\end{enumerate}
By a union bound, the properties above hold for all $c\in A-B$, with $1-1/\poly(n)$ probability.
Hence we can run the algorithm of \cite{CardozeS98,nickconstellation} to solve the Constellation  with $k'$ mismatches instance $(A,B')$ in $\tilde O(nk') = \tilde O(n)$ time (note $k'<(1-0.1\eta_0)|B'|$ still holds),  and obtain $C = \{c\in \Z: |(c+B')\setminus A|\le k'\}$. By \cref{item1} above,  $C$ is a set of candidates which contains all possible answers in the original Constellation with $k_0$ mismatches instance $(A,B)$  (with $1-1/\poly(n)$ probability).
Note that by \cref{item2} above, $|C| \le \frac{|A||B|}{|B|-(1+\eta_0/20) k_0} \le 2\eta_0^{-1}|A|$.
Then, it remains to compute $(1_A\conv 1_{-B})[c]$ for all $c\in C$, and output all $c\in C$ satisfying $(1_{A}\conv 1_{-B})[c] > |B|-k_0$.  This is the task of \cref{prob:reducedprob} and satisfies the promise of \cref{prob:reducedprob} for $k:= (1+\eta_0/20)k_0$ (by \cref{item2} above). This finishes the description of the reduction.
\end{proof}

Fischer \cite{nickconstellation} already gave an algorithm solving \cref{prob:reducedprob} which is efficient when $C$ is small:
\begin{lemma}[follows from \cite{nickconstellation}]
\cref{prob:reducedprob} has a deterministic algorithm in $\tilde O(|A|+|B|+k|C|)$ time.
\label{lem:smallc}
\end{lemma}
\begin{proof}
Note that \[|B+C|\le |A| + \sum_{c\in C}|(c+B)\setminus A|  \underset{\text{\cref{eqn:promisec}}}{\le}  |A|+k|C|.\] 
Using \cref{lem:nickcount}, we can compute $(1_A\conv 1_{-B})[c]$ 
for all $c\in C$ in $\tilde O(|A|+|B+C|) \le \tilde O(|A|+|B|+k|C|)$ time.
\end{proof}

\subsection{Algorithm for large \texorpdfstring{$B$}{B}}
\label{sec:largeb}
The goal of this section is to prove the randomized version of \cref{thm:constelargebmain} ($k$-mismatch Constellation for small $B$) as a warm-up. To do this, we will solve \cref{prob:reducedprob}. The algorithm for \cref{prob:reducedprob} itself will be deterministic, but it only implies a randomized version of \cref{thm:constelargebmain} due to the randomness in the reduction to \cref{prob:reducedprob} in \cref{prop:simpleconstellationreduction}.  
We will modify this randomized algorithm into a  deterministic one (hence fully proving \cref{thm:constelargebmain}) in the next subsection, using a standard scaling technique.

In \cref{prob:reducedprob}, we denote $L = |A|/|B|$, and we think of $L$ as small.

Since we originally started with Constellation with $k\le 0.3|B|$ mismatches (as stated in \cref{thm:constelargebmain}), after the reduction of \cref{prop:simpleconstellationreduction} we still have \[k\le 0.4|B|.\]
 
 Let $1\le R \le k$ be a parameter to be determined later. If $|C|\le |A|/R$, then we can use \cref{lem:smallc} and solve \cref{prob:reducedprob} in $\tilde O(|A|+|B| +k\cdot |C|) \le \tilde O(k|A|/R)$ time deterministically, which is faster than $\tilde O(k|A|)$.
Hence,  in the following, we assume 
\begin{equation}
    \label{eqn:clb}
|C|\ge |A|/R.
\end{equation}

\begin{lemma}
$|C-C|\le 5|A|L$.
\label{lem:cminusc}
\end{lemma}
\begin{proof}
 Consider the bipartite graph $G = \{(b,c)\in B\times C: b+c\in A\}$. The degree of a node $c\in C$ in $G$ is 
 $\ge |B|-k$
  by \cref{eqn:promisec}.
 For every $(c,c')\in C\times C$, since both $c,c'$ have degree at least  $|B|-k \ge 0.6|B|$, they have at least $2\cdot (|B|-k)-|B| \ge 0.2|B|$ common neighbors in $G$.
 Hence, $c-c'$ can be represented as $(c+b)-(c'+b)$ where $c+b,c'+b\in A$, for at least $0.2|B|$ many distinct $b$. Therefore, $(1_A\conv 1_A) [c-c']\ge 0.2|B|$ for all $c-c'\in C-C$, and  $|C-C|\le \frac{|A|\cdot |A|}{0.2|B|} = 5|A|L$.
\end{proof}

We use \cref{thm:introdeterministicapprox3sum} to 
 deterministically compute a vector $f$ such that $\|f-1_A\conv 1_{-C}\|_\infty \le 0.25|C|$, in $O((|A|+|B|+|C|)N^{o(1)})$ time, and let $B_{\mathrm{bad}}:= \{b\in B: f[b] \le 0.5|C|\}$. Then, we know
\begin{itemize}
    \item For all $b\in B_{\mathrm{bad}}$, $(1_A\conv 1_{-C})[b] \le 0.75|C|$.
        \item For all $b\in B\setminus B_{\mathrm{bad}}$, $(1_A\conv 1_{-C})[b] \ge 0.25 |C|$.
\end{itemize}
Note that by \cref{eqn:promisec}, $\sum_{b\in B} (1_A\conv 1_{-C})[b]  = \sum_{c\in C} (1_A\conv 1_{-B})[c] \ge |C|\cdot (|B|-k)$. Hence,
\begin{equation}
    \label{eqn:bbadub}
 |B_{\mathrm{bad}}|\le  \frac{\sum_{b\in B}(|C|- (1_A\conv 1_{-C})[b])}{|C|-0.75|C|} \le \frac{|B||C| - |C|\cdot (|B|-k)}{0.25|C|} =  4k.
\end{equation}

Let $B' = B\setminus B_{\mathrm{bad}}$.

\begin{lemma}
    \label{lem:bprimeplusc}
   $|B'+C| \le 20|A|RL$.
\end{lemma}
\begin{proof}
 Consider the bipartite graph $G = \{(b,c)\in B\times C: b+c\in A\}$. For every $b\in B'$, the degree of $b$ in $G$ is $(1_A\conv 1_{-C})[b]\ge 0.25|C|$.
Hence, for every $b\in B', c\in C$, $b+c$ can be represented as $(b+c')+(c-c')$ where $(b+c')\in A, (c-c')\in (C-C)$  by choosing $c'\in C$ from the neighbors of  $b$ in $G$, which has $\ge  0.25|C|$ different possibilities. Therefore, $(1_A\conv 1_{C-C})[b+c]\ge 0.25|C|$ for all $b+c\in B'+C$, and hence \[|B'+C| \le \frac{ |A| \cdot |C-C|}{0.25|C|} \underset{\text{\cref{eqn:clb}}}{\le} 4R |C-C| \underset{\text{\cref{lem:cminusc}}}{\le} 20|A| R L.  \]
\end{proof}

Then, the algorithm decomposes the answers into two parts,
\[(1_A \conv 1_{-B})[c] = (1_A \conv 1_{-B'})[c] + (1_A \conv 1_{-B_{\mathrm{bad}}})[c],\]
which are separately computed as follows for all $c\in C$:
\begin{itemize}
    \item 
Use \cref{lem:nickcount} to compute $(1_{A}\conv 1_{-B'})[c]$ for all $c\in C$, in $\tilde O(|B'+C|) \le  \tilde O(|A| RL)$ time by \cref{lem:bprimeplusc} (this running time is dominated by the running time for the next part). 

\item Use \cref{thm:detsmalldoublethreesum} (applied to $(A,B,C,S):=(C,B_{\mathrm{bad}},A,-C)$ ) to compute $(1_{A}\conv 1_{ B_{-\mathrm{bad}}})[c]$ for all $c\in C$, in  
    \begin{align*}
 &        \tilde O\Big (\frac{|C-C|\sqrt{|A||B_{\mathrm{bad}}|}}{\sqrt{|-C|}}\Big ) + (|C-C|+|A|)N^{o(1)} 
 \\
 & \le  O\Big (\frac{|A|L\sqrt{|A|\cdot k}}{\sqrt{|A|/R}}N^{o(1)}\Big )  \tag{by \cref{lem:cminusc}, \cref{eqn:clb}, \cref{eqn:bbadub}}\\
 & =  O(|A| L\sqrt{kR} N^{o(1)}). 
    \end{align*}
\end{itemize}
        
Combined with the running time $\tilde O(|A|k/R)$ for the $|C|<|A|/R$ case,  the total time for solving \cref{prob:reducedprob} becomes
\begin{align*}
& \tilde O(|A|k/R) +  O(|A| L\sqrt{kR}N^{o(1)}) \\
 & \le  O(|A|k^{2/3}L^{2/3} N^{o(1)}) 
\end{align*}
by choosing $R=(k/L^2)^{1/3}\ge 1$ (assuming $L\le \sqrt{k}$; otherwise the $\tilde O(|A|k)$-time algorithm by \cite{nickconstellation} already satisfies the claimed time bound).

\subsection{Derandomization \texorpdfstring{$B$}{B} by scaling}
\label{sec:constellationderand}

In this section we will prove (the deterministic version of) \cref{thm:constelargebmain}, so we are not allowed to use the randomized reduction to \cref{prob:reducedprob} any more. The main idea is basically the same as in the previous section, and we only need to 
apply the scaling trick (e.g., \cite{BringmannN21,BringmannFN22}) on top of it, so we will be brief in the proof. 
We remark that Fischer's deterministic algorithm for $k$-mismatch Constellation \cite{nickconstellation} also used scaling, but for a multiset (i.e., weighted) version of the problem.
Here we avoid dealing with multisets in order to simplify some of the arguments.

\begin{proof}[Proof of \cref{thm:constelargebmain}]
   Let $A,B\subseteq [N]$ be the input instance of the $k$-mismatch Constellation problem, where $k\le 0.3|B|$. 

   At the very beginning, we use \cref{lem:detfindmodulus} to find a modulus $M \in [|B|/2, |B|\cdot N^{o(1)}]$ such that the set $B\bmod M$ has size $|B\bmod M| \ge 0.9|B|$, in $O(|B|\cdot N^{o(1)})$ time.
   
   We perform iterations $0\le i \le \lceil \log_2 (2N/M)\rceil$, where in iteration $i$ we define modulus $M_i = M\cdot 2^i$, and define \emph{sets} $A_i:= (A\bmod M_i) + \{0,M_i\} \subseteq \Z \cap [0,2M_i), B_i := B\bmod M_i \subseteq \Z \cap [0,M_i)$.  Note that $|B_i|\ge |B \bmod M| \ge 0.9|B|$, so $k\le 0.3|B| < 0.4|B_i|$ for all $i$.
We have the following simple observation:
\begin{observation}
For all $0\le c < M_{i+1}$, $|(c+B_{i+1} )\setminus A_{i+1}| \ge |((c\bmod M_i) + B_i)\setminus A_i|$.
\end{observation}
\begin{proof}
   For any element $b_i \in B_i$ such that $(c\bmod M_i) + b_i \in ((c\bmod M_i)+B_i)\setminus A_i$, let $b_{i+1} \in B_{i+1}$ be such that $b_{i+1}\in b_i + \{0,M_i\}$. If $(c+b_{i+1})\in A_{i+1}$, then we would have $(c+b_i)\bmod M_i = (c+b_{i+1})\bmod M_{i}  \in (A_{i+1})\bmod M_i = A\bmod M_i$, so that $((c\bmod M_i)+b_{i}) \in (A\bmod M_i)+\{0,M_i\} = A_i$, a contradiction.  Hence, $c+b_{i+1}\in (c+B_{i+1})\setminus A_{i+1}$. This implies the desired size inequality.
\end{proof}

We will maintain the following invariant: at the end of iteration $i$ ($i\ge 0$), we have computed the set 
\[C_i:= \{c\in [0,M_i): (1_{A_i} \conv 1_{-B_i})[c] \ge |B_i|-k\},\]
which has size $|C_i| \le \frac{|A_i||B_i|}{|B_i|-k} = O(|A_i|) = O(|A|)$, and satisfies $|B_i+C_i| \le |A_i|+k|C_i|$.

This invariant at the end of iteration $0$ can be easily satisfied: we first use FFT to compute $1_{A_0} \conv 1_{-B_0}$ in $\tilde O(M_0) = O(|B|N^{o(1)})$ time, and let $C_0:= \{c\in [0,M_0): (1_{A_0} \conv 1_{-B_0})[c] \ge |B_0|-k\}$.

Now we need to satisfy the invariant at the end of iteration $i$ ($i\ge 1$).  We already have $C_{i-1}$ from the previous iteration. We run the algorithm from \cref{sec:largeb} (in a whitebox manner) on $A_{i-1},B_{i-1},C_{i-1}$, $k$, which satisfy the promise of \cref{prob:reducedprob} that $k\le 0.4|B_{i-1}|$ and $(c+B_{i-1})\setminus A_{i-1}|\le k$ for all $c\in C_{i-1}$ (due to the invariant).
Now compute $C'_{i} = C_{i-1} + \{0,M_{i-1}\}$. We know $C_i \subseteq C'_i$. 
Recall that the algorithm from \cref{sec:largeb} has two cases:
\begin{itemize}
    \item Case $|C_{i-1}|<|A_{i-1}|/R$.
        
    In this case we have $|B_i+ C_i'| \le |B_{i-1}+\{0,M_{i-1}\} + C_{i-1}+\{0,M_{i-1}\}| \le 3|B_{i-1}+C_{i-1}| \le 3(|A_i| + k|C_{i-1}|) = O(|A|k/R)$. So we can compute $(1_{A_i}\conv 1_{-B_{i}})[c]$ for all $c\in C'_i$ in $\tilde O(|A|k/R)$ time using \cref{lem:nickcount} (similarly to \cref{lem:smallc}).
    
    \item Case $|C_{i-1}|<|A_{i-1}|/R$.
        
    In this case, the algorithm from \cref{sec:largeb} computes the decomposition $B_{i-1} =B_{\mathrm{bad}}\cup  B'$, such that $|B'+C_{i-1}|= O(|A|RL)$ and $|C_{i-1}-C_{i-1}| = O(|A|L)$, where $L = |A_{i-1}|/|B_{i-1}| \le |A|/(0.9|B|)$. 
    
    We use this to obtain a decomposition $B_i = B_{\mathrm{bad},i}\cup  B'_i$, such that $B_{\mathrm{bad},i} \subseteq B_{\mathrm{bad}} + \{0,M_{i-1}\}$ and $B'_i \subseteq B' + \{0,M_{i-1}\}$. Then we still have $|B'_i + C_{i}| \le 3|B'+C_{i-1}| = O(|A|RL)$ and $|C_i-C_i| \le 3|C_{i-1}-C_{i-1}|= O(|A|L)$, and $|B_{\mathrm{bad},i}| \le 2|B_{\mathrm{bad}}|$. So we can use the same procedure described at the end of \cref{sec:constellationderand} to compute 
$(1_{A_i} \conv 1_{-B_i})[c] = (1_{A_i} \conv 1_{-B_i'})[c] + (1_{A_i} \conv 1_{-B_{\mathrm{bad},i}})[c]$ for all $c\in C_i'$, and the running time is the same as before up to a constant factor.
\end{itemize}
In both cases we have computed $(1_{A_i} \conv 1_{-B_i})[c]$ for all $c\in C'_i$. Since $C_i \subseteq C'_i$, we use this information to obtain the desired set $C_i$.

When the last iteration $i^*$ finishes, since $M_{i^*}\ge  2N$,  the set $C_{i^*}$ gives the answer to the original $k$-mismatch Constellation  instance.
\end{proof}

\subsection{Combinatorial lower bound}

In this section, we prove \cref{thm:BMMLowerBound}, which we recall below: 
\BMMLowerBound*

\begin{proof}
For any $0 < \alpha \le 1$, it was shown in \cite{GawrychowskiU18} that, finding all shifts that result in $\le k$ mismatches in a Text-to-Pattern Hamming Distances instance with text $T$ and pattern $P$ where $|T| = n$, $|P| = \Theta(n^\alpha), k = \Theta(n^\alpha)$ requires $n^{1+\alpha/2-o(1)}$ time for combinatorial algorithms, under the BMM hypothesis. 

We can then further reduce this instance to a $k$-Mismatch Constellation instance with $\Theta(n)$-size integer sets $A, B$. Without loss of generality, we can assume all characters in the Text-to-Pattern Hamming Distances instance are integers in $[n]$. Let $M = 10n$ be a sufficiently large integer. Then we let 
\[
A := \{1, \ldots, 2n\} \cup \{i + M \cdot T[i]: 1 \le i \le |T|\}
\]
and 
\[
B := \{1, \ldots, n\} \cup \{i + M \cdot P[i]: 1 \le i \le |P|\}. 
\]
For any $0 \le c \le |T| - |P|$, consider $ \left| (c + B) \setminus A\right|$. The $c + \{1, \ldots, n\}$ part of $c+B$ will lie in $A$, so they will not contribute towards $ \left| (c + B) \setminus A\right|$. Hence, 
\[
    \left| (c + B) \setminus A\right| = \left| \left\{ 1 \le i \le |P|: i+ c + M \cdot P[i] \not \in A\right\}\right|.
\]
For every $1 \le i \le |P|$, the only number in $A$ that can possibly be equal to $i + c + M \cdot P[i]$ is $i + c + M \cdot T[i+c]$, by considering the numbers modulo $M$ (also, $i + c + M \cdot P[i] \ge M$, so numbers in $\{1, \ldots, 2n\}$  cannot be equal to it). Therefore, 
\begin{align*}
    \left| (c + B) \setminus A\right| &= \left| \left\{ 1 \le i \le |P|: i+ c + M \cdot P[i]  = i + c + M \cdot T[i+c]\right\}\right|\\
    & = \left| \left\{ 1 \le i \le |P|: P[i] =  T[i+c]\right\}\right|, 
\end{align*}
which is exactly the number of mismatches between $T$ and $P$ with shift $c$. Hence, if we can solve $k$-Mismatch Constellation between $A$ and $B$, we can determine the set of shifts $c$ for $P$ so that its number of mismatches with $T$ is bounded by $k$. Therefore, by the lower bound from \cite{GawrychowskiU18}, $k$-Mismatch Constellation requires $n^{1+\alpha /2  - o(1)} = \sqrt{k} n^{1-o(1)}$ time for combinatorial algorithms, under the BMM hypothesis.
\end{proof}

\subsection{Applications to \texorpdfstring{$k$}{k}-Mismatch String Matching with Wildcards}
\label{subsec:wildcard}
 Given a text string $T[0\dd n) \in \Sigma^{n}$ and a pattern string $P[0\dd m) \in (\Sigma \cup \{\diamondsuit\})^{m}$, our task is to find all occurrences of $P$ inside $T$ allowing up to $k$ Hamming mismatches (where wildcard $\diamondsuit$ can match any single character).
We first present a simple reduction which yields an $\tilde O(k^{2/3}n)$-time algorithm for this task, and then describe how to improve it to $\tilde O(k^{1/2}n)$ after some small modifications.

Assume $m\le n \le 2m$ via the standard trick of dividing $T$ into $\lfloor n/m\rfloor$ (overlapping) chunks each of length $\le 2m$. 

\begin{proof}[Proof of \cref{cor:wildcardbasic}]
View the alphabet $\Sigma$ as integers $\{1,2,\dots,|\Sigma|\}$. Construct two point sets $A, B\subset \Z^2$ by $A:= \{(i,T[i]) : 0\le i< n\}$ and $B:= \{(i,P[i]) : 0\le i< m, P[i]\neq \diamondsuit\}$.  The key observation is that, for all $c\in \{0,1,\dots,n-m\}$, $T[c\dd c+m)$ is a $k$-mismatch occurrence of $P[0\dd m)$ if and only if $\big \lvert \big (B+(c,0)\big ) \setminus A\big \rvert \le k$. (For every Hamming mismatch $P[i]\neq T[i+c]$, the point $(c+i,P[i])\in B+(c,0)$ is not in $A$, contributing one to the size of $(B+(c,0)\big ) \setminus A$. The wildcards do not make any contribution.)

 Hence, we can apply our $k$-mismatch Constellation algorithm (\cref{thm:constelargebmain}) to $A,B$ and obtain the set $C' = \{ c'\in \Z^2: |(c'+B)\setminus A| < k\}$, and return $\{c\in \{0,1,\dots,n-m\}: (c,0) \in C'\}$ as the answer.\footnote{To invoke \cref{thm:constelargebmain} we need to first flatten each 2-dimensional point $(x,y)$ into $(yM+x)$ where $M$ is a large enough integer ($M = 10n$ suffices). }
 Assuming $|B|\ge \Omega(m)$, the (randomized) time complexity of \cref{thm:constelargebmain}  is $\tilde O(|A|\cdot k^{2/3}\cdot (|A|/|B|)^{2/3}) = \tilde O(n\cdot k^{2/3}\cdot (n/m)^{2/3}) = \tilde O(n\cdot k^{2/3})$ as claimed.
 
Now we justify the $|B|\ge \Omega(m)$ assumption, as well as the $k< 0.3|B|$ assumption required by  \cref{thm:constelargebmain}. Note that $|B|$ equals the number of non-wildcard characters in the pattern. We add three non-wildcard characters $\spadesuit_1\spadesuit_2\spadesuit_3$ right after every character in the input pattern and  every character in the input text. (The new text and pattern have lengths $4n$ and $4m$ respectively.) Then, when the shift is $4\cdot i$, all the $\spadesuit_{(\cdot)}$ symbols are matched, and the number of Hamming mismatches is the same as in the original instance at shift $i$. Now $|B|\ge 3m$ and $k\le m < 0.3|B|$ hold.
\end{proof}

Now we describe how to modify the proof of \cref{thm:constelargebmain} to improve the time bound to $\tilde O(k^{1/2}n)$.
\begin{proof}[Proof sketch of \cref{thm:wildcardbetter}]
   Based on $P,T$, construct the same point sets $A,B$ as in the first paragraph of the previous proof of \cref{cor:wildcardbasic}. Now we run the $k$-Mismatch Constellation algorithm described in \cref{sec:largeb} on input sets $A,B \subset \Z^2$ (for simplicity, here we directly work with $\Z^2$ without flattening to 1-dimension), but with the following modifications: 
   \begin{itemize}
       \item Since we only care about shifts in $C':= \{ (c,0): c\in \{0,1,\dots,n-m\} \}$, we can discard all elements in the candidate set $C$ that are not in $C'$. Therefore, 
\cref{lem:cminusc} can be trivially improved to $|C-C| \le |C'-C'| <  2(n-m+1) = O(n)$.
\item \cref{lem:bprimeplusc} relies the estimate from \cref{lem:cminusc}, so it can be improved accordingly to $|B'+C| = O(nR)$. Hence, the step of using 
\cref{lem:nickcount} to compute $(1_{A}\conv 1_{-B'})[c]$ for all $c\in C$ now runs in $\tilde O(|B'+C|) \le  \tilde O(n R)$ time.
\item For the next step of computing  $(1_{A}\conv 1_{ B_{-\mathrm{bad}}})[c]$ for all $c\in C\subseteq C'$, we apply 
\cref{thm:detsmalldoublethreesum} to $(A,B,C,S):=(C',B_{\mathrm{bad}},A,-C')$, with time complexity
    \begin{align*}
 &        \tilde O\Big (\frac{|C'-C'|}{\sqrt{|-C'|}}\cdot \sqrt{|A||B_{\mathrm{bad}}|}\Big ) + (|C'-C'|+|A|)N^{o(1)} 
 \\
 & \le   O\Big (\sqrt{n}\cdot \sqrt{|A|\cdot k} + n^{1+o(1)}\Big )\\
 & \le \sqrt{k}n^{1+o(1)}.
    \end{align*}
       \item Finally,  combined with the running time $\tilde O(|A|k/R)$ for the $|C|<|A|/R$ case,  the total time for solving \cref{prob:reducedprob} becomes
$ \tilde O(|A|k/R) +  \tilde O(n R) + \sqrt{k}n^{1+o(1)}  \le   \sqrt{k}n^{1+o(1)} $ by choosing $R = \Theta(\sqrt{k})$.  
\end{itemize}
   Again, for the randomized algorithm the $n^{o(1)}$ factor in the time complexity can be replaced by $\polylog(n)$. For the deterministic algorithm, we use the same argument as in  \cref{sec:constellationderand}.
\end{proof}

\section*{Acknowledgements}
We thank Gabriel Bathie, Panagiotis Charalampopoulos, and Tatiana Starikovskaya for bringing the $k$-Mismatch String Matching with Wildcards problem to our attention and answering our questions.

\bibliographystyle{alphaurl} 
\bibliography{main}

\end{document}